\documentclass[sigconf,dvipsnames]{acmart}

\usepackage{graphicx}
\usepackage{balance}  
\usepackage{booktabs}
\usepackage{listings}
\usepackage{enumitem}
\usepackage{tabularx}
\usepackage{bm}
\usepackage{clrscode3e}
\usepackage{footnote}
\usepackage{standalone} 
\usepackage{colortbl}
\usepackage{forest}
\usepackage[algo2e,linesnumbered,ruled,vlined]{algorithm2e}
\usepackage{xspace}
\usepackage{url}
\usepackage{textcomp}
\usepackage{multirow}
\usepackage{booktabs}
\usepackage{hyperref}
\usepackage{color}
\usepackage{tcolorbox}
\usepackage{xpatch}
\usepackage[skip=0.2pt]{subcaption}
\usepackage[nounderscore]{syntax}
\usepackage{tikz-qtree,tikz-qtree-compat}
\usepackage{tikz}
\usepackage{cleveref}
\usetikzlibrary{arrows.meta,bending,quotes}
\usepackage{mathtools}
\usepackage{ulem}
\usepackage{array}
\usepackage{float}
\tikzset{every tree node/.style={align=center, anchor=north}}

\definecolor{dkgreen}{rgb}{0,0.6,0}
\definecolor{gray}{rgb}{0.5,0.5,0.5}
\definecolor{mauve}{rgb}{0.58,0,0.82}
\AtBeginDocument{%
  \providecommand\BibTeX{{%
    \normalfont B\kern-0.5em{\scshape i\kern-0.25em b}\kern-0.8em\TeX}}}

\setcopyright{None}
\copyrightyear{2018}
\acmYear{2018}
\acmDOI{10.1145/1122445.1122456}

\acmConference[Woodstock '18]{Woodstock '18: ACM Symposium on Neural
  Gaze Detection}{June 03--05, 2018}{Woodstock, NY}
\acmBooktitle{Woodstock '18: ACM Symposium on Neural Gaze Detection,
  June 03--05, 2018, Woodstock, NY}
\acmPrice{15.00}
\acmISBN{978-1-4503-XXXX-X/18/06}

\newcommand{\narrow}[1]{\ensuremath{\text{\scalebox{0.7}[1.0]{\textsf{#1}}}}}
\newcommand{\Tables}{\narrow{Tables}}
\newcommand{\Table}{\narrow{Table}}
\newcommand{\Aliases}{\narrow{Aliases}}
\newcommand{\Attributes}{\narrow{Attributes}}

\newcommand{\Cost}{\narrow{Cost}}
\newcommand{\Sim}{\narrow{Sim}}
\newcommand{\mapping}{\ensuremath{\mathfrak{m}}}
\newcommand{\dist}{\narrow{dist}}

\newcommand{\Children}{\narrow{Children}}

\newcommand{\op}{\narrow{op}}
\newcommand{\RepairWhere}{\narrow{RepairWhere}}
\newcommand{\CreateBounds}{\narrow{CreateBounds}}
\newcommand{\DeriveFixes}{\narrow{DeriveFixes}}
\newcommand{\DeriveFixesOPT}{\narrow{DeriveFixesOPT}}
\newcommand{\DistributeFixes}{\narrow{DistributeFixes}}

\newcommand{\myComment}{\itshape\footnotesize\color{blue}// }

\let\oldnl\nl
\newcommand{\nonl}{\renewcommand{\nl}{\let\nl\oldnl}}

\SetCommentSty{myCommentFont}
\SetKwInOut{Input}{Input}
\SetKwInOut{Output}{Output}
\SetKw{Subroutine}{\underline{SUBROUTINE}}
\SetKw{Return}{return}
\SetKw{Continue}{continue}
\SetKw{Break}{break}
\SetKw{Let}{let}

\newcommand{\sql}[1]{\texttt{#1}}
\newcommand{\sqlquote}[1]{\texttt{\textquotesingle#1\textquotesingle}}
\newcommand{\SELECT}{\sql{SELECT}}
\newcommand{\FROM}{\sql{FROM}}
\newcommand{\WHERE}{\sql{WHERE}}
\newcommand{\GROUPBY}{\sql{GROUP} \sql{BY}}
\newcommand{\HAVING}{\sql{HAVING}}
\newcommand{\WITH}{\sql{WITH}}
\newcommand{\DISTINCT}{\sql{DISTINCT}}
\newcommand{\JOIN}{\sql{JOIN}}

\newcommand{\Fr}[1]{\ensuremath{\narrow{F}(#1)}}
\newcommand{\FrWh}[1]{\ensuremath{\narrow{FW}(#1)}}
\newcommand{\FrWhGr}[1]{\ensuremath{\narrow{FWG}(#1)}}
\newcommand{\FrWhGrHa}[1]{\ensuremath{\narrow{FWGH}(#1)}}

\newcommand{\mseq}{\ensuremath{\mathrel{\smash{\overset{\lower.5em\hbox{$\scriptscriptstyle\Box$}}{=}}}}}

\newcommand{\concat}{\ensuremath{\mathbin\Vert}}
\newcommand{\set}[1]{\ensuremath{\mathcal{#1}}}
\newcommand{\bool}[1]{\ensuremath{\mathsf{#1}}}
\newcommand{\True}{\narrow{true}}
\newcommand{\False}{\narrow{false}}

\newcommand{\MinFix}{\narrow{MinFix}}
\newcommand{\MinFixMult}{\narrow{MinFixMult}}
\newcommand{\MapAtomPreds}{\narrow{MapAtomPreds}}
\newcommand{\DontCare}{\ensuremath{\ast}}
\newcommand{\BuildTruthTable}{\narrow{BuildTruthTable}}
\newcommand{\InitFeasibility}{\narrow{InitFeasibility}}
\newcommand{\UpdateFeasibility}{\narrow{UpdateFeasibility}}
\newcommand{\PickSite}{\narrow{PickSite}}

\newcommand{\IsEquiv}{\narrow{IsEquiv}}
\newcommand{\IsUnSat}{\narrow{IsUnSatisfiable}}
\newcommand{\IsSat}{\narrow{IsSatisfiable}}
\newcommand{\Context}{\ensuremath{\mathcal{C}}}
\newcommand{\MinBoolExp}{\narrow{MinBoolExp}}
\newcommand{\FixGrouping}{\narrow{FixGrouping}}
\newcommand{\FixSelect}{\narrow{FixSelect}}

\newcommand{\ostar}{\ensuremath{o^{\smash{\scriptscriptstyle\star}}}}

\definecolor{black}{rgb}{0,0,0}
\definecolor{grey}{rgb}{0.8,0.8,0.8}
\definecolor{red}{rgb}{1,0,0}
\definecolor{green}{rgb}{0,1,0}
\definecolor{applegreen}{rgb}{0.55, 0.71, 0.0}
\definecolor{darkgreen}{rgb}{0,0.5,0}
\definecolor{darkpurple}{rgb}{0.5,0,0.5}
\definecolor{darkdarkpurple}{rgb}{0.3,0,0.3}
\definecolor{blue}{rgb}{0,0,1}
\definecolor{shadegreen}{rgb}{0.95,1,0.95}
\definecolor{shadeblue}{rgb}{0.95,0.95,1}
\definecolor{shadered}{rgb}{1,0.85,0.85}
\definecolor{shadegrey}{rgb}{0.85,0.85,0.85}
\definecolor{oddRowGrey}{rgb}{0.80,0.80,0.80}
\definecolor{evenRowGrey}{rgb}{0.85,0.85,0.85}

\newcommand{\cut}[1]{{}}

\SetCommentSty{mycommfont}

\newcommand{\oursys}{\textsc{Qr-Hint}}
\newcommand{\oursystitle}{Qr-Hint}

\newcommand{\mypar}[1]{\smallskip\noindent\textbf{{#1}.}}

\DeclareMathAlphabet{\mathbbold}{U}{bbold}{m}{n}

\newtheorem{Definition}{Definition}

\newtheorem{Example}{Example}

\newcommand{\card}[1]{|{#1}|}

\newcommand{\BigO}{\ensuremath{\mathrm{O}}}

\newcommand{\union}{\cup}

\newenvironment{sizeddisplay}[1]
 {\par\nopagebreak#1\noindent\ignorespaces}
 {\nopagebreak\ignorespacesafterend}

\definecolor{lstpurple}{rgb}{0.5,0,0.5}
\definecolor{lstred}{rgb}{1,0,0}
\definecolor{lstreddark}{rgb}{0.7,0,0}
\definecolor{lstredl}{rgb}{0.64,0.08,0.08}
\definecolor{lstmildblue}{rgb}{0.66,0.72,0.78}
\definecolor{lstblue}{rgb}{0,0,1}
\definecolor{lstmildgreen}{rgb}{0.42,0.53,0.39}
\definecolor{lstgreen}{rgb}{0,0.5,0}
\definecolor{lstorangedark}{rgb}{0.6,0.3,0}	
\definecolor{lstorange}{rgb}{0.75,0.52,0.005}
\definecolor{lstorangelight}{rgb}{0.89,0.81,0.67}
\definecolor{lstbeige}{rgb}{0.90,0.86,0.45}

\DeclareFontShape{OT1}{cmtt}{bx}{n}{<5><6><7><8><9><10><10.95><12><14.4><17.28><20.74><24.88>cmttb10}{}

\lstdefinelanguage{smtlib2}{
  alsoletter=-,
  morekeywords={declare-const,define-fun,assert,minimize,maximize,check-sat,get-objectives,and,or,not,distinct},
  extendedchars=false,
  keywordstyle=\bfseries\color{lstpurple},
  deletekeywords={Int,Bool},
  keywords=[2]{Int,Bool},
  keywordstyle=[2]\color{lstblue},
}

\lstdefinestyle{psql}
{
tabsize=2,
basicstyle=\scriptsize\upshape\ttfamily,
language=SQL,
morekeywords={PROVENANCE,BASERELATION,INFLUENCE,COPY,ON,TRANSPROV,TRANSSQL,TRANSXML,CONTRIBUTION,COMPLETE,TRANSITIVE,NONTRANSITIVE,EXPLAIN,SQLTEXT,GRAPH,IS,ANNOT,THIS,XSLT,MAPPROV,cxpath,OF,TRANSACTION,SERIALIZABLE,COMMITTED,INSERT,INTO,WITH,SCN,UPDATED},
extendedchars=false,
keywordstyle=\bfseries,
mathescape=true,
escapechar=@,
sensitive=true
}

\lstdefinestyle{psqlcolor}
{
tabsize=2,
basicstyle=\scriptsize\upshape\ttfamily,
language=SQL,
morekeywords={PROVENANCE,BASERELATION,INFLUENCE,COPY,ON,TRANSPROV,TRANSSQL,TRANSXML,CONTRIBUTION,COMPLETE,TRANSITIVE,NONTRANSITIVE,EXPLAIN,SQLTEXT,GRAPH,IS,ANNOT,THIS,XSLT,MAPPROV,cxpath,OF,TRANSACTION,SERIALIZABLE,COMMITTED,INSERT,INTO,WITH,SCN,UPDATED},
extendedchars=false,
keywordstyle=\bfseries\color{lstpurple},
deletekeywords={count,min,max,avg,sum},
keywords=[2]{count,min,max,avg,sum},
keywordstyle=[2]\color{lstblue},
stringstyle=\color{lstreddark},
commentstyle=\color{lstgreen},
mathescape=true,
escapechar=@,
sensitive=true
}

\lstdefinestyle{datalog}
{
basicstyle=\footnotesize\upshape\ttfamily,
language=prolog
}

\lstdefinestyle{pseudocode}
{
  tabsize=3,
  basicstyle=\small,
  language=c,
  morekeywords={if,else,foreach,case,return,in,or},
  extendedchars=true,
  mathescape=true,
  literate={:=}{{$\gets$}}1 {<=}{{$\leq$}}1 {!=}{{$\neq$}}1 {append}{{$\listconcat$}}1 {calP}{{$\cal P$}}{2},
  keywordstyle=\color{lstpurple},
  escapechar=&,
  numbers=left,
  numberstyle=\color{lstgreen}\small\bfseries, 
  stepnumber=1, 
  numbersep=5pt,
}

\lstdefinestyle{xmlstyle}
{
  tabsize=3,
  basicstyle=\small,
  language=xml,
  extendedchars=true,
  mathescape=true,
  escapechar=£,
  tagstyle=\color{keywordpurple},
  usekeywordsintag=true,
  morekeywords={alias,name,id},
  keywordstyle=\color{lstred}
}

\lstdefinestyle{smtlib2}
{
tabsize=2,
basicstyle=\scriptsize\upshape\ttfamily,
numbers=left,
stepnumber=1,
breaklines=true,
stringstyle=\color{lstreddark},
commentstyle=\color{lstgreen},
mathescape=true,
escapechar=@,
sensitive=true
}

\renewcommand\footnotetextcopyrightpermission[1]{} 

\begin{document}


\title{\oursystitle: Actionable Hints Towards Correcting Wrong SQL Queries} 

\author{Yihao Hu}
\email{yihao.hu@duke.edu}
\affiliation{
  \institution{Duke University}
  \city{Durham}
  \state{NC}
  \country{USA}
}

\author{Amir Gilad}
\email{amir.gilad1@mail.huji.ac.il}
\orcid{1234-5678-9012}
\affiliation{%
  \institution{The Hebrew University of Jerusalem}
  \city{Jerusalem}
  \country{Israel}
  \postcode{9190501}
}

\author{Kristin Stephens-Martinez}
\email{ksm@cs.duke.edu}
\affiliation{
  \institution{Duke University}
  \city{Durham}
  \state{NC}
  \country{USA}
}

\author{Sudeepa Roy}
\email{sudeepa@cs.duke.edu}
\affiliation{
  \institution{Duke University}
  \city{Durham}
  \state{NC}
  \country{USA}
}

\author{Jun Yang}
\email{junyang@cs.duke.edu}
\affiliation{
  \institution{Duke University}
  \city{Durham}
  \state{NC}
  \country{USA}
}


\begin{abstract}
We describe a system called \oursys\ that, given a (correct) target query $Q^\star$ and a (wrong) working query $Q$, both expressed in SQL,
provides actionable hints for the user to fix the working query so that it becomes semantically equivalent to the target.
It is particularly useful in an educational setting, where novices can receive help from \oursys\ without requiring extensive personal tutoring.
Since there are many different ways to write a correct query, we do not want to base our hints completely on how $Q^\star$ is written;
instead, starting with the user's own working query, \oursys\ purposefully guides the user through a sequence of steps
that provably lead to a correct query, which will be equivalent to $Q^\star$ but may still ``look'' quite different from it.
Ideally, we would like \oursys's hints to lead to the ``smallest'' possible corrections to $Q$.
However, optimality is not always achievable in this case due to some foundational hurdles
such as the undecidability of SQL query equivalence and the complexity of logic minimization.
Nonetheless, by carefully decomposing and formulating the problems and developing principled solutions,
we are able to provide provably correct and locally optimal hints through \oursys. 
We show the effectiveness of \oursys\ through quality and performance experiments as well as a user study in an educational setting.

\end{abstract}

\maketitle
\pagenumbering{arabic} 
\pagestyle{plain}   

\section{Introduction}
\label{sec:intro}

In an era of widespread database usage, SQL remains a fundamental skill for those working with data.
Yet, SQL's rich features and declarative nature can make it challenging to learn and understand.
When students encounter difficulties in debugging their SQL queries, they often turn to instructors and teaching assistants for guidance.
However, this one-on-one approach is limited in scalability.
Syntax errors are easy to fix, but many queries contain subtle semantic errors that may require careful and time-consuming debugging.
To save time, the teaching staff is often tempted to give hints based on how the reference solution query is written,
ignoring what students have written themselves,
but doing so misses opportunities for learning.
A SQL query can be written in many ways that are different in syntax but nonetheless equivalent semantically.
Seasoned teaching staff knows how to guide students through a sequence of steps that,
starting with their own queries, lead them to a corrected version that is equivalent to the solution query but without revealing the solution query.
Our goal is to build a system to help provide this service to students in a more scalable manner.

\begin{Example}\label{ex:intro1}
Consider the following database (keys are underlined) about beer drinkers and bars:
\narrow{Likes}(\underline{\narrow{drinker}}, \underline{\narrow{beer}}),
\narrow{Frequents}(\underline{\narrow{drinker}}, \underline{\narrow{bar}}),
\narrow{Serves}(\underline{\narrow{bar}}, \underline{\narrow{beer}}, \narrow{price}).
Suppose we want to write a SQL query for the following problem:
\emph{For each beer $b$ that Amy likes and each bar $r$ frequented by Amy that serves $b$,
show the rank of $r$ among all bars serving $b$ according to price
(e.g., if $r$ serves $b$ at the highest price, $r$'s rank should be $1$).
We assume that there are no ties.}

The reference solution query $Q^\star$ is given as follows:
\begin{lstlisting}[language=SQL, basicstyle=\small\ttfamily]
SELECT L.beer, S1.bar, COUNT(*)
FROM Likes L, Frequents F, Serves S1, Serves S2
WHERE L.drinker = F.drinker AND F.bar = S1.bar
  AND L.beer = S1.beer AND S1.beer = S2.beer
  AND S1.price <= S2.price
GROUP BY F.drinker, L.beer, S1.bar
HAVING F.drinker = 'Amy';
\end{lstlisting}

Now consider a wrong student query $Q$:
\begin{lstlisting}[language=SQL, basicstyle=\small\ttfamily]
SELECT s2.beer, s2.bar, COUNT(*)
FROM Likes, Serves s1, Serves s2
WHERE drinker = 'Amy'
  AND Likes.beer = s1.beer AND Likes.beer = s2.beer
  AND s1.price > s2.price
GROUP BY s2.beer, s2.bar;
\end{lstlisting}
\end{Example}

\vspace*{-2ex}

Suggesting good hints to help students fix $Q$ is not easy.
First, there are many ways to write a query that is equivalent to $Q^\star$,
and queries that look very different syntactically might be semantically similar or equivalent,
so relying solely on the syntactic difference between $Q$ and $Q^\star$ to propose fixes is ineffective and potentially misleading.
In Example~\ref{ex:intro1}, even though $Q^\star$ has a \HAVING\ clause, it would be confusing to suggest add \HAVING\ to $Q$,
because the condition $\narrow{drinker}{=}\sqlquote{Amy}$ in $Q$'s \WHERE\ serves the same purpose.
Also, even though $Q$ has $\narrow{Likes.beer}{=}\narrow{s2.beer}$ in \WHERE\ while $Q^\star$ has $\narrow{S1.beer}{=}\narrow{S2.beer}$,
the difference is non-consequential because of the transitivity of equality.
Yet another example is $\narrow{s1.price}{>}\narrow{s2.price}$ in $Q$ versus $\narrow{S1.price}{\le}\narrow{S2.price}$ in $Q^\star$.
It would be wrong to suggest changing $>$ to $\le$ in $Q$,
because an examination of the entire $Q$ would reveal that
the student intends \narrow{s2} (and \narrow{s1}) in $Q$ to serve the role of \narrow{S1} (and \narrow{S2}) in $Q^\star$.
The correct fix is actually changing $>$ to $\ge$.%
\footnote{Another wrong hint would be to suggest changing \sql{COUNT(*)} to \sql{COUNT(*)+1} in $Q$'s \SELECT\ instead of changing the inequality
because doing so misses the top-ranked bars.
\oursys\ will not make such a mistake.}

Second, it is often impossible to declare a part of $Q$ as ``wrong'' since one could instead fix the remainder of $Q$ to compensate for it.
For example, we could argue that $\narrow{s1.price}{>}\narrow{s2.price}$ in $Q$ is ``wrong,''
but there exists a correct query containing precisely this condition, e.g., with \sql{(}$\narrow{s1.price}{>}\narrow{s2.price}$ \sql{OR} $\narrow{s1.price}{=}\narrow{s2.price}$\sql{)}.
Hence, it is difficult to formally define what ``wrong'' means.
Instead of basing our approach heuristically on calling out ``wrong'' parts,
we formulate the problem as finding the ``smallest repairs'' to $Q$ that make it \emph{correct}.

Third, hints are for human users, so for a query with multiple issues---which is often the case in practice---%
we must be aware of the cognitive burden on users and not overwhelm them by asking them to make multiple fixes simultaneously.
This desideratum introduces the challenge of planning the sequence of hints and defining appropriate intermediate goals.

Finally, effective hinting faces several fundamental barriers.
Realistically, we cannot hope to always provide ``optimal'' hints
because doing so entails solving the query equivalence problem for SQL, which is undecidable~\cite{trahtenbrot1950impossibility, abiteboul1995foundations, KMO2022, stackexchange:query-equivalence};
even for decidable query fragments, Boolean expression minimization is known to be on the second level of the polynomial hierarchy (precisely $\mathbf{\Sigma}^p_2$~\cite{DBLP:journals/jcss/BuchfuhrerU11}).

To address the challenges, we propose \oursys, a system that,
given a target query $Q^\star$ and a working query $Q$, follows the logical execution flow (i.e., $\FROM {\rightarrow} \WHERE {\rightarrow} \GROUPBY {\rightarrow} \HAVING {\rightarrow} \SELECT$) and produces step-by-step hints for the user to edit the working query to eventually achieve $Q^\star$.
The sequence of steps is guaranteed to lead the user on a correct path to eventual correctness.
The following example shows \oursys\ helps fix the query in Example~\ref{ex:intro1}.

\begin{Example}\label{ex:intro2}
Continuing with Example~\ref{ex:intro1}, \oursys\ automatically generates the sequence of hints below.
Currently built for the teaching staff, \oursys\ only generates the ``repairs'' below;
using these repairs, the teaching staff would then hint the user in natural language.
With the recent advances in generative AI chatbots, it would not be difficult to automate the natural language hints as well;
the advantage of using \oursys\ in that setting would be to provide provable guarantees on the quality of hints,
which otherwise would be difficult, if not impossible, for generative AI to achieve by itself.

\vspace*{0.3ex}
\centerline{\small
\begin{tabular}{rp{0.2\columnwidth}p{0.6\columnwidth}}\hline
Stage & \oursys\ repair & Hint in natural language\\\hline
\FROM\ &  \narrow{Frequents} needed &
\emph{It looks like you are missing one table---read the problem carefully and see what other piece of information you need.}\\\hline
\WHERE\ & $\narrow{s1.price}{>}\narrow{s2.price} \mapsto \narrow{s1.price}{\ge}\narrow{s2.price}$ &
\emph{Your \WHERE\ has a small problem with $\narrow{s1.price}{>}\narrow{s2.price}$.
Think through some concrete examples and see how you may fix it.}\\\hline
\end{tabular}
}
\vspace*{0.3ex}

Note the sequential nature of the hints above; the working query constantly evolves.
\oursys\ first focuses on \FROM\ and will only proceed to \WHERE\ after \FROM\ is ``viable.''
After adding \narrow{Frequents} to \FROM, the user will also need to add appropriate join conditions in \WHERE;
if these were not added correctly, the second step above would suggest additional repairs.
It turns out that for this example, only the above two hints are needed to fix the query.
In particular, \oursys\ knows \emph{not} to suggest spurious hints
such as adding to \narrow{Frequents.drinker} to \GROUPBY\ or changing \narrow{s2.beer} to \narrow{Likes.beer} in \SELECT.
\end{Example}

We make the following contributions:
\begin{itemize}[leftmargin=*]
\item We develop a novel framework that allows \oursys\ to provide step-by-step hints to fix a working SQL query
    with the goal of making it equivalent to a target query.
    This framework formalizes the notion of ``correctness'' for a sequence of hints,
    allowing \oursys\ to guarantee that every hint is actionable and is on the right path to achieve eventual correctness.
    Further, by formulating the hinting problem in terms of finding repair sites in $Q$ with viable fixes,
    we are able to quantify the quality of the hints.
\item Since the optimality of hints, in general, is impossible to achieve due to the foundational hurdles discussed earlier,
    we aim to provide guarantees on the ``local'' optimality of \oursys\ in each step.
    We design practical algorithms with sensible trade-offs between optimality and efficiency.
\item We evaluate the performance and efficacy of \oursys\ experimentally.
    We further perform a user study involving students from current/past database courses offered at the authors' institution.
    Our findings indicate that \oursys\ finds repairs that are optimal or close to optimal in practice under reasonable time,
    and they lead to hints that are helpful for students.
\end{itemize}
\section{Related Work}\label{sec:related-work}

{\bf Debugging Query Semantics.}
There are two main lines of work toward debugging query semantics (as opposed to syntax or performance).
The first line helps debug a query but without knowing the correct (reference) query;
in this regard, it differs fundamentally from \oursys.
Qex~\cite{veanes2010qex} is a tool for generating input relations and parameter values for unit-testing parameterized SQL queries.
SQLLint~\cite{brass2003detecting, brass2004detecting, brass2005proving, brass2006semantic, goldberg2009you} detects suspected semantic errors in a query, alerting users to what may be indicative of efficiency, logical, or runtime errors.
The work highlights a list of common semantic errors made by students and SQL users~\cite{brass2006semantic},
but it does not suggest edits, and fixing the suspected errors will not guarantee that the query is correct.
Habitat~\cite{grust2013observing, dietrich2015sql} is a query execution visualizer that allows users to highlight parts of a query and view their intermediate results.
While it helps users spot possible errors, it gives no edit suggestions if errors exist.
More recently, QueryVis~\cite{leventidis2020queryvis} turns queries into intuitive diagrams, helping users better understand the semantics of the queries and spot potential errors.

The second line of work, more directly related to \oursys, focuses on checking a query against a reference query and/or helping to explain their difference.
However, previous work has not been able to suggest small fixes that will make the user query equivalent to the reference query.
XData~\cite{chandra2015data} checks the correctness of a query by running the query on self-generated testing datasets based on a set of pre-defined common errors, but it provides no guarantees beyond this pre-defined set.
Cosette~\cite{chu2017cosette, chu2017hottsql, chu2018axiomatic} uses constraint solvers and theorem provers to establish the equivalence of two queries or construct arbitrary instances that differentiate them.
From a large database instance, RATest~\cite{miao2019explaining} utilizes data provenance to generate a small, illustrative instance to differentiate queries.
C-instances~\cite{gilad2022understanding} aims at constructing small abstract instances based on c-tables~\cite{imielinski1989incomplete} that can differentiate two given queries in all possible ways.
While Cosette, RATest, and c-instances can provide examples illustrating how two queries are semantically different,
they can only indirectly help users pinpoint errors in the original query; none of them is able to suggest fixes.
Chandra et al. ~\cite{chandra2021edit} developed a grading system that canonicalizes queries by applying rewrite rules and then decides partial credits based on a tree-edit distance between logical plans.
However, as query syntax differs significantly from canonicalized plans after rewrite, edits on a canonicalized plan do not translate naturally to small fixes on the original query, making it hard for users to use these edits as hints.
Finally, SQLRepair~\cite{presler2021sqlrepair} fixes simple errors in an SPJ query using constraint solvers to synthesize/remove \WHERE\ conditions until the query produces correct outputs over all testing instances.
Its scope of error is much narrower than what we consider, and its tests-driven nature offers no guarantee of query equivalence.

{\bf Program Repair and Feedback for GPL.}
In the domain of program repair for general-purpose programming language (GPL), several types of approaches have been developed but none of them can be directly applied or easily transferred to cover SQL.
First, a wrong program is usually aligned with reference program(s) (\cite{ahmed2022verifix, gulwani2018automated, wang2018search}) and fixes are generated based on the selected reference program using various techniques. Such an approach is similar to \oursys, but SQL is essentially different from GPL as SQL is declarative and GPLs are usually procedural. While it is possible to write programs in GPL to simulate the execution of a specific SQL query, there is no well-defined mapping between the syntax of SQL and any GPL. As a result, it is impossible to apply such program repair techniques to SQL in general.
Another approach is to leverage test cases to synthesize ``patches'' for the wrong program so that it returns the same output as the reference program for all test cases (\cite{hua2018towards, perelman2014test, rivers2017data, singh2013automated, xiong2018identifying}). However, such an approach heavily relies on the test cases to cover all possible errors and thus usually fails to guarantee semantic equivalence. 
%
Besides the traditional approaches, recent work explores 
ML algorithms to provide feedback and correction (\cite{berabi2021tfix, bhatia2018neuro, chen2019sequencer, gupta2019deep, gupta2019neural, lutellier2020coconut, piech2015learning}). In addition, large language models such as GPT-3~\cite{brown2020language} have shown an ability to explain the semantics of SQL queries,  
but does not guarantee the correctness of fixes. 

{\bf Testing query equivalence.}
While the query equivalence problem in general is undecidable 
\cite{Shmueli1993, trahtenbrot1950impossibility, abiteboul1995foundations, ABLMP21-newdbtheorybook},
tools and algorithms are developed to check the equivalence of various classes of queries with restrictions and assumptions \cite{chandra1977optimal, aho1979equivalences, sagiv1980equivalences, klug1988conjunctive, ioannidis1995containment, jayram2006containment, chu2017cosette, chu2017hottsql, chu2018axiomatic, zhou2019automated, wang2022wetune}. Although they give a deterministic answer on equivalence, these tools/algorithms cannot provide any explanation on which parts of the users' queries cause semantic differences from the reference queries.


\section{The \oursystitle\ Framework}
\label{sec:framework}


\mypar{Queries}
We consider SQL queries that are select-project-join queries with an optional single level of grouping and aggregation.
For simplicity of presentation, we assume these are single-block SQL queries
with \SELECT, 
\FROM\ (without \JOIN\ operators),
and \WHERE\ (with condition defaulting to \sql{TRUE} if missing) clauses,%
\footnote{\label{footnote:subquery}We can handle a query with common table expressions (\WITH) and subqueries in \FROM\ that are aggregation-free,
as well as non-outer \JOIN{}s in \FROM,
by rewriting the query into single-block SQL.}
together with optional \GROUPBY\ and \HAVING\ clauses.
We refer to such a query as an \emph{SPJA} query if it contains grouping or aggregation or \DISTINCT;
otherwise, we will call it an \emph{SPJ} query.

We assume the default bag (multiset) semantics of SQL.
Given query $Q$, let $\Fr{Q}$ denote the cross product of $Q$'s \FROM\ tables (including multiple occurrences of the same table, if any);
and let $\FrWh{Q}$ denote the query that further filters $\Fr{Q}$ by $Q$'s \WHERE\ condition (i.e.,
$\FrWh{Q}$ is a \SELECT\ \sql{*} query with the same \FROM\ and \WHERE\ clauses as $Q$).
Furthermore, if $Q$ is SPJA, let $\FrWhGr{Q}$ denote the (non-relational) query%
\footnote{This query is non-relational because it returns, besides the underlying bag of rows from $\FrWh{Q}$, a partitioning of them into groups.}
that further groups the result rows of $\FrWh{Q}$ according to $Q$'s \GROUPBY\ expressions
(or $\emptyset$ if there are none but $Q$ contains aggregation nonetheless, in which case all result rows belong to a single group).
Finally, if $Q$ is SPJA, let $\FrWhGrHa{Q}$ denote the (non-relational) query
that filters the groups of $\FrWhGr{Q}$ according to $Q$'s \HAVING\ conditions (which defaults to \sql{TRUE} if missing).
When discussing equivalence (denoted $\equiv$) among above queries,
we require that they return the same bag of result rows (ignoring row and column ordering) for any underlying database instance,
and additionally, for queries returning groups, they return the same partitioning of result rows (ignoring group ordering).

\mypar{SMT Solvers}
As with previous work~\cite{chu2018axiomatic, miao2019explaining, zhou2019automated}, we leverage \emph{satisfiability modulo theory} (\emph{SMT}) solvers to
implement various primitives used by our system.
Such a solver can decide whether a formula, modulo the theories it references, is satisfiable, unsatisfiable, or unknown (beyond the solver's capabilities).
Specifically, we use the popular SMT solver Z3~\cite{de2008z3} to implement the following three primitives.
Given two quantifier-free expressions, $\IsEquiv(e_1, e_2)$ tests
whether $e_1 \Leftrightarrow e_2$ (for logic formulae such as those in \WHERE)
or $e_1 = e_2$ (for value experssions such as those in \SELECT\ or \GROUPBY).
Given a logic formula $p$, $\IsUnSat(p)$ and $\IsSat(p)$ return, respectively, whether $p$ is satisfiable or unsatisfiable, respectively.
All above primitives may return ``unknown'' when Z3 is unsure about its answer.
However, when they return true, Z3 guarantees that the answer is not a false positive.
Our algorithms in subsequent sections act only on (true) positive answers from these primitives.
For complex uses, it is often convenient to frame equivalence/satisfiability testing using a \emph{context} $\Context$,
or a set of logical assertions (e.g., types declaration, known constraints, and inference rules) under which testing is done.
We use subscripts to specify the context:
e.g., $\IsUnSat_\Context(p)$ is a shorthand for $\IsUnSat\left((\wedge_{c \in \Context} c) \wedge p\right)$.
\begin{Example}\label{ex:z3}
Consider a query with a \WHERE\ condition stipulating that $A>100$ for an \sql{INT}-typed column $A$,
as well as a \HAVING\ 
condition $\sql{MAX}(A) \ge 101$. We might wonder whether the \HAVING\ condition is unnecessary.
To this end, we call $\IsUnSat_\Context(p)$ with

\begin{sizeddisplay}{\footnotesize}
\begin{align*}
\allowdisplaybreaks
\Context&: \left\{\;\begin{aligned}
    \mathbf{A} \text{ has type }\narrow{Array}(\mathbb{Z})\\
    \forall i \in \mathbb{N}: \mathbf{A}[i] > 100\\
    \sql{MAX} \text{ has type }\narrow{Array}(\mathbb{Z}) \to \mathbb{Z}\\
    \forall i \in \mathbb{N}, \mathbf{X} \text{ of type }\narrow{Array}(\mathbf{Z}): \sql{MAX}(\mathbf{X}) \ge \mathbf{X}[i]
\end{aligned}\;\right\},&  p: \neg (\sql{MAX}(\mathbf{A}) \ge 101).
\end{align*}
\end{sizeddisplay}
The first two assertions in \Context\ are derived from the type of $A$ and the \WHERE\ conditions;
here the array-typed $\mathbf{A}$ refers to a collection of $A$ values.
The last two specify (some) general inference rules on the SQL aggregate function \sql{MAX}.
Z3 correctly returns true, meaning that $\sql{MAX}(A) \ge 101$ must be true under \Context\ and is therefore unnecessary.
\end{Example}

Our use of Z3 for reasoning with SQL aggregation, such as the example above, goes beyond the practice in previous work,
where aggregation functions are mostly treated as uninterpreted functions.
For example, to test the equality of two aggregates, \cite{zhou2019automated} conservatively checks whether input value sets or multisets for the aggregate function are equal.
In contrast, we encode properties of SQL aggregation functions in a way that allows Z3 to reason with them.
As formulae become more complicated, e.g., with quantifiers and arrays, Z3 no longer offers a complete decision procedure
(as there exists no decision procedure for first-order logic) and may return ``unknown'' more often.
Nonetheless, practical heuristics employed by Z3 allow it to handle many cases of practical uses to \oursys. 

\vspace{-2mm}
\subsection{Approach}\label{sec:framework:approach}
Given a ({\em syntactically correct}) working query $Q$ and a target query $Q^\star$,
\oursys\ provides hints in \emph{stages} to help the user edit the working query incrementally until it becomes \emph{semantically equivalent} to $Q^\star$.
Each stage focuses on one specific syntactic fragment of the working query.
\oursys\ gives actionable hints for the user to edit this fragment with the aim of bringing $Q$ a step ``closer'' to being equivalent to $Q^\star$.
\oursys\ strives to suggest the smallest edits possible and avoid suggesting unnecessary edits.
Upon passing a \emph{viability check}, the working query $Q$ clears the current stage and moves on to the next.
After clearing all stages, \oursys\ guarantees that $Q \equiv Q^\star$ (even if syntactically they are still different).

We now briefly outline the concrete stages of \oursys; the details will be presented in the subsequent sections.

For an SPJ query, there are three stages.
(1) We start with $Q$'s \FROM\ clause (Section~\ref{sec:from}) and make sure that its list of tables can eventually lead to a correct query;
following this stage, $\Fr{Q} \equiv \Fr{Q^\star}$.
(2) Next, we provide hints to repair $Q$'s \WHERE\ clause (Section~\ref{sec:where}) such that $\FrWh{Q} \equiv \FrWh{Q^\star}$,
i.e., the repaired query returns the same sub-multiset of rows as $Q^\star$ that satisfy the \WHERE\ clause, ignoring \SELECT.
(3) Finally, we handle $Q$'s \SELECT\ clause and ensure the working query returns correct output column values.
Importantly, we make inferences of equivalence under the premise that all rows before \SELECT\ already satisfy \WHERE;
this use of \WHERE\ allows us to infer more equivalent cases and avoid spurious hints.

For an SPJA query, there are five stages.
(1) The {\em first stage} handles \FROM\ as in the SPJ case.
(2) The {\em second stage} handles \WHERE, but with a twist.
As we have seen from Example~\ref{ex:intro1}, some condition can be either \WHERE\ or \HAVING, and it would be misleading to hint its absence from \WHERE\ to be wrong;
hence, \oursys\ will look ``ahead'' at the two queries' \HAVING\ and \GROUPBY\ clauses to avoid misleading the user.
At the end of this stage, instead of insisting that $\FrWh{Q} \equiv \FrWh{Q^\star}$ for the original $Q^\star$,
we may rewrite $Q^\star$ (by legally moving some conditions between \WHERE\ and \HAVING) as needed first.
(3) The {\em third stage} is \GROUPBY, where we provide hints to edit $Q$'s \GROUPBY\ expressions to achieve equivalent grouping, i.e., $\FrWhGr{Q} \equiv \FrWhGr{Q^\star}$.
Here, we infer equivalence under the premise that the rows to be grouped all satisfy \WHERE.
(4) The {\em fourth stage} is \HAVING, where we provide hints to repair $Q$'s \HAVING\ condition in the same vein as \WHERE;
however, inferences in this stage would additionally consider both \WHERE\ and \GROUPBY, and they are more challenging because of aggregation functions.
After this stage, we have $\FrWhGrHa{Q} \equiv \FrWhGrHa{Q^\star}$.
(5) The {\em fifth and final stage} is \SELECT, which is similar to the SPJ case,
but with the challenge of handling aggregation functions while simultaneously considering \WHERE, \GROUPBY, and \HAVING.

\mypar{Progress and Correctness}
Note that to clear a stage, the user only needs to come up with a fix to pass the viability checks up to this stage.
Even though \oursys\ may examine the queries in their entirety, the user does not have to think ahead about how to make the entire query correct.%
\footnote{In some cases, just to maintain syntactic correctness, a fix may necessitate trivial edits to fragments handled in future stages:
e.g., if we remove a table from \FROM, we will need to remove references to this table in the rest of the query.
However, the user never needs to worry about making those edits semantically correct---that responsibility falls on future stages.}
Moreover, once a stage is cleared, \oursys\ never requires the user to come back to fix the same fragment again.
This stage-by-stage design with ``localized'' hints helps limit the cognitive burden on the user.

The following theorem formalizes the intuition that this stage-based approach leads to steady, forward progress toward the goal of fixing the working query.
It follows from the observation that our solution for each stage ensures the properties asserted below,
which we will show stage by stage in the subsequent sections.
\begin{theorem}\label{theorem:main}
Let $Q_0 = Q$ denote the initial working query and $Q^\star$ denote the target query.
Let $V_i$ denote the viability check for stage $i$,
and $Q_i$ denote the working query upon clearing stage $i$, where $Q_i$ satisfies $V_1, V_2, \ldots, V_i$.
We say that two queries are \emph{stage-$i$ consistent} if they are identical syntactically except in the fragments that stage $i+1$ and beyond focus on.
For each stage $i$, the following 
hold:
\begin{itemize}[leftmargin=1em]
    \item[]\hspace*{-1em}\textbf{(Hint leads to fix)}
    If $Q_{i-1}$ fails to satisfy $V_i$, there exists a query $\hat{Q}_i$ such that
    $\hat{Q}_i$ satisfies $V_1, V_2, \ldots, V_i$,
    $\hat{Q}_i$ is stage-$(i-1)$ consistent with $Q_{i-1}$,
    and $\hat{Q}_i$ follows the stage-$i$ hint provided by \oursys.
    \item[]\hspace*{-1em}\textbf{(Fix leads to eventual correctness)}
    There exists a query $\hat{Q}$ such that $\hat{Q} \equiv Q^\star$ and $\hat{Q}$ is stage-$i$ consistent with $Q_i$.
\end{itemize}
\end{theorem}

We delegate all proofs in this paper to the appendix.

\mypar{Optimality}
Ideally, we would like \oursys\ to suggest the ``best possible'' hints, e.g., those leading to minimum edits to the working query.
Unfortunately, it is impossible for any system to provide such a guarantee in general,
because doing so entails being able to determine the equivalence of SQL queries:
if $Q \equiv Q^\star$ to begin with, the system should not suggest any fix.
It is well-known that the equivalence of first-order queries with only equality comparisons is undecidable~\cite{abiteboul1995foundations}.
Under bag semantics, even the decidability of equivalence of conjunctive queries has not been completely resolved~\cite{KMO2022}.
Once we open up to the full power of SQL, which can express integer arithmetic,
even equivalence of selection predicates becomes undecidable via a simple reduction to the satisfiability of Diophantine equations~\cite{stackexchange:query-equivalence}.

Given the foundational hurdles above, \oursys\ seeks a pragmatic solution.
Instead of offering any global guarantee on the optimality of its hints, which is impossible,
\oursys\ establishes, for each stage, guarantees on the necessity or minimality of its hints under certain assumptions.
For example, for the \FROM\ stage, \oursys\ guarantees its suggested fixes are optimal for SPJ queries,
but for some SPJA queries, it may suggest a fix that turns out to be unnecessary.
As another example, for the \WHERE\ stage, the optimality of \oursys\ depends on, among other things, Z3-based primitives offering \emph{complete} decision procedures.
In each subsequence section, we will state any such assumption explicitly.

Finally, it is important to note that \oursys's progress and correctness properties (Theorem~\ref{theorem:main}) do \emph{not} rely on these assumptions.
In the worst case, the user may be hinted to make some fixes that are unnecessary or unnecessarily big,
but \oursys\ will still ensure that the user gets a correct working query in the end.

%

\mypar{Limitations}
Following \Cref{theorem:main}, \oursys\ is guaranteed to generate correct hints for select-project-join queries with an optional single level of grouping and aggregation. On the other hand, \oursys\ currently has several limitations.
1)~\oursys\ may sometimes suggest suboptimal or even unnecessary fixes (even though they still lead to correct queries), as discussed above;
the reason lies in fundamental hurdles due to the undecidability of SQL query equivalence and the use of heuristics to tame complexity.
2)~\oursys\ currently does not handle \sql{NULL} values and assumes that all database columns are \sql{NOT} \sql{NULL}.
With some additional effort and complexity, \oursys\ can be extended to handle \sql{NULL} using the technique in~\cite{zhou2019automated} of encoding each column with a pair of variables in Z3 (one for its value and the other a Boolean representing whether it is \sql{NULL}).
The same applies to \sql{OUTER} \sql{JOIN}.
3)~Except the case of aggregation-free subqueries in \FROM\ mentioned in Footnote~\ref{footnote:subquery}, \oursys\ does not support subqueries in general.
Subqueries involving aggregation in general cannot be folded into the outer query block.
Subquery constructs such as \sql{NOT} \sql{EXISTS} and \sql{NOT} \sql{IN} entail supporting queries involving the difference operator, which we have not yet studied.
If we do not care about the number of duplicates in the result,
positive subqueries with \sql{EXISTS} and \sql{IN} could be rewritten as part of the join in the outer select-project-join query and supported as such.
However, this approach is unsatisfactory, especially since our handling of \FROM\ (Section~\ref{sec:from}) does assume that duplicates matter.
In general, more work is needed to develop a comprehensive solution for subqueries.
4)~Finally, \oursys\ does not consider database constraints such as keys and foreign keys.
While we can, in theory, encode some constraints as logical assertions and include them as part of the context when calling Z3,
these assertions (with quantifiers) can significantly hamper Z3's performance.
Future work is needed to develop more robust algorithms for incorporating constraints.

\section{\FROM\ Stage}
\label{sec:from}

This stage aims to ensure $\Fr{Q} \equiv \Fr{Q^\star}$.
Recall that a \sql{FROM} clause may reference a table $T$ multiple times, and each reference is associated with a distinct alias (which defaults to the name of $T$).
Each column reference must resolve to exactly one of these aliases.
Let $\Tables(Q)$ denote the multiset of tables in the \sql{FROM} clause of $Q$, and let $\Aliases(Q)$ denote the set of aliases they are associated with in $Q$.
With a slight abuse of notation, given table $T$, let $\Aliases(Q, T)$ denote the subset of $\Aliases(Q)$ associated with $T$ (a non-singleton $\Aliases(Q, T)$  implies a self-join involving $T$).
Given an alias $t \in \Aliases(Q)$, let $\Table(Q, t)$ denote the table that $t$ is associated with in $Q$.

The viability check (Theorem~\ref{theorem:main}, stage 1) for \FROM\ is simple:
$$V_1: \text{Check if}\; \Tables(Q) \mseq \Tables(Q^\star)$$
where $\mseq$ denotes multiset equality.
If the working query $Q$ fails the viability check, \oursys\ simply hints,
for each table $T$ whose counts in $\Tables(Q)$ and $\Tables(Q^\star)$ differ (including cases where $T$ is used in one query but not the other),
that the user should consider using $T$ more or less to make the counts the same.
It is straightforward to see that this hint leads to a fix that makes $\Tables(Q) \mseq \Tables(Q^\star)$,
which enables the user to further edit $Q$ into some $\hat{Q} \equiv Q^\star$ without retouching \FROM:
at the very least, one can make $\hat{Q}$ isomorphic to $Q^\star$ up to the substitution of table references with those in $\Aliases(Q)$.
This observation establishes the progress and correctness properties (see Theorem~\ref{theorem:main}) of \FROM-stage hints,
which we state below along with the remark that $\Fr{Q} \equiv \Fr{Q^\star}$ after this stage.

\begin{lemma}\label{lemma:from-correctness}
\oursys's \FROM-stage hint leads to a fixed working query $Q_1$ that
(1)~passes the viability check $V_1$ $\Tables(Q_1) \mseq \Tables(Q^\star)$;
(2)~satisfies $\Fr{Q_1} \equiv \Fr{Q^\star}$; and
(3)~leads to eventual correctness.
\end{lemma}

While the correctness of the \FROM-stage hint is straightforward, its optimality is surprisingly strong.
The following lemma states that the viability check is, in fact, necessary---regardless of what could be done in \WHERE\ and \SELECT---under reasonable assumptions.

\begin{lemma}\label{lemma:from-optimality}
Two SPJ queries $Q^\star$ and $Q$ cannot be equivalent under bag semantics if $\Tables(Q^\star) \not\mseq \Tables(Q)$ assuming no database constraints are present, and there exists some database instance for which either $Q^\star$ or $Q$ returns a non-empty result. 
\footnote{The assumption of a not-always-empty result may seem out of the blue but is necessary.
For example, queries \sql{SELECT} \sql{1} \sql{FROM} \sql{R} \sql{WHERE} \sql{FALSE} and \sql{SELECT} \sql{1} \sql{FROM} \sql{R,R} \sql{WHERE} \sql{FALSE} are equivalent---both always return empty results.
However, if at least one of $Q^\star$ and $Q$ can return non-empty results, $\Tables(Q^\star) \mseq \Tables(Q)$ becomes necessary for equivalence.
Our proof of Lemma~\ref{lemma:from-optimality}, in fact, builds on such a non-empty result.}
\end{lemma}

\mypar{Table Mappings}
To facilitate analysis in subsequent stages, \oursys\ needs a way to ``unify'' table and column references in $Q$ and $Q^\star$ so that all of them use the same set of table aliases.

\begin{Definition}\label{def:table-mapping}
    Given queries $Q^\star$ and $Q$ over the same schema where $\Tables(Q^\star) \mseq \Tables(Q)$,
    a \emph{table mapping} from $Q^\star$ to $Q$ is a bijective function $\mapping: \Aliases(Q^\star) \to \Aliases(Q)$
    with the property that two corresponding aliases are always associated with the same table, i.e., $\forall t \in \Aliases(Q^\star): \Table(Q^\star, t) = \Table(Q, \mapping(t))$.
\end{Definition}

If the queries have no self-joins, it is straightforward to establish this mapping by table names.
With self-joins, however, it can be tricky because we must match multiple roles played by the same table across queries.
The information contained in \FROM\ alone would be insufficient for matching.
One approach is to explore every possible table mapping and select the one that leads to the minimum fix.
Doing so would blow up complexity by a factor exponential in the number of self-joined tables.
\oursys\ instead opts for a heuristic that picks the single most promising table mapping.
Here we describe the heuristic briefly. For each alias,
we build a ``signature'' that captures how its columns are used by various parts of the query in a canonical fashion.
We define a distance (cost) metric for the signatures.
Then, for each table involved in self-joins, to determine the mapping between its aliases in $Q$ and $Q^\star$,
we construct a bipartite graph consisting of these aliases and solve the minimum-cost bipartite matching problem.
We illustrate the high-level idea using the example below.

\begin{Example}\label{ex:table-mapping}
Continuing with Example~\ref{ex:intro1}, the following are signatures (one per column) for \narrow{S1} and \narrow{S2} in $Q^\star$ and \narrow{s1} and \narrow{s2} in $Q$.

\centerline{\small\setlength{\tabcolsep}{0.2em}
\begin{tabular}{rr|l|l|l|l}
    &
            & \narrow{S1} in $Q^\star$
                    & \narrow{S2} in $Q^\star$
                            & \narrow{s1} in $Q$
                                    & \narrow{s2} in $Q$
\\\hline\hline
\WHERE\ \&
    & \narrow{bar}:
            & ${=}\{\narrow{F.bar}\}$
                    & ${=}\{\narrow{F.bar}\}$
                            & None
                                    & None
\\
\HAVING\
    & \narrow{beer}:
            & ${=}\{\narrow{L.beer}, \narrow{S2.beer}\}$
                    & ${=}\{\narrow{L.beer}, \narrow{S2.beer}\}$
                            & ${=}\{\narrow{Likes.beer}, \narrow{s2.beer}\}$
                                    & ${=}\{\narrow{Likes.beer}, \narrow{s2.beer}\}$
\\
    & \narrow{price}:
            & ${\le}\{\narrow{S2.price}\}$
                    & ${\ge}\{\narrow{S2.price}\}$
                            & ${>}\{\narrow{s2.price}\}$
                                    & ${<}\{\narrow{s1.price}\}$
\\\hline
\GROUPBY\
    &
            & $\{\narrow{bar}, \narrow{beer}\}$
                    & $\{\narrow{beer}\}$
                            & $\{\narrow{beer}\}$
                                    & $\{\narrow{beer}\}$
\\\hline
\SELECT\
    & \narrow{bar}:
            & $\{ 2 \}$
                    & $\emptyset$
                            & $\emptyset$
                                    & $\{ 2 \}$
\\
    & \narrow{beer}:
            & $\{ 1 \}$
                    & $\{ 1 \}$
                            & $\{ 1 \}$
                                    & $\{ 1 \}$
\\
    & \narrow{price}:
            & $\emptyset$
                    & $\emptyset$
                            & $\emptyset$
                                    & $\emptyset$
\\\hline\hline
\end{tabular}
}

For example, \narrow{S1.beer}'s \WHERE/\HAVING\ signature says that it is involved in an equality comparison with both \narrow{L.beer} and \narrow{S2.beer};
the latter is inferred---\oursys\ automatically adds column references and constants that obviously belong to the same equivalence class.
Likewise, \narrow{S1}'s \GROUPBY\ signature includes both \narrow{bar} and \narrow{beer}, with the latter added because of its equivalence to the \GROUPBY\ column \narrow{L.beer}.
When comparing signatures, all aliases are replaced by table names (which is a heuristic simplification);
therefore, all four \WHERE/\HAVING\ signatures above for \narrow{beer} are considered the same.
In this case, what makes the difference in bipartite matching turns out to be the \SELECT\ signatures for \narrow{bar},
which clearly favors the mapping with $\narrow{S1} \mapsto \narrow{s2}$ and $\narrow{S2} \mapsto \narrow{s1}$.
\end{Example}

Once we have selected the table mapping $\mapping$, we can then ``unify'' $Q^\star$ and $Q$.
For convenience, we simply rename each alias $a$ in $Q^\star$ to $\mapping(a)$;
in subsequent sections, we shall assume that $Q^\star$ and $Q$ have consistent column references.
\section{\WHERE\ Stage}
\label{sec:where}

\WHERE\ is our most involved stage,
aimed at making small edits to the \WHERE\ condition of $Q$ so that it becomes logically equivalent to that of $Q^\star$,
thereby ensuring $\FrWh{Q^\star} \equiv \FrWh{Q}$ (recall from \Cref{sec:framework:approach}).
Let $P$ and $P^\star$ denote the \WHERE\ predicates in $Q$ and $Q^\star$, respectively.
We assume that they have already been unified by the selected table mapping to have the same set of column references, as discussed in \Cref{sec:from}.
The viability check for the \WHERE\ stage (Theorem~\ref{theorem:main}, stage 2) is simply that $P$ is logically equivalent to $P^\star$: 
$$V_2: \text{Check if}\; P \Leftrightarrow P^\star$$

As discussed in \Cref{sec:intro}, if $P \not\Leftrightarrow P^\star$,
there are many different ways to modify $P$ so that becomes logically equivalent to $P^\star$,
and it is impossible to declare any part of $P$ as definitively ``wrong.''
Therefore, we suggest the smallest possible edits on $P$ to reduce the cognitive burden on the user.
We formalize the notion of ``small edits'' below.
We represent $P$ and $P^\star$ using syntax trees, where:
\begin{itemize}[leftmargin=*]
\item Internal (non-leaf) nodes represent logical operators $\land$, $\lor$, and $\lnot$.
    Let $\op(x)$ denote the operator associated with node $x$, and $\Children(x)$ denote the $x$'s child nodes.
    If $\op(x)$ is $\lnot$, $\card{\Children(x)} = 1$.
    If $\op(x) \in \{\land, \lor\}$, $\card{\Children(x)} \ge 2$.
\item Leaf nodes are atomic predicates involving column references and/or literals.
    We treat each unique column reference as a free variable over the domain of the referenced column.
    We support basic SQL types and as well as standard comparison, arithmetic, and string operators to the extent supported by Z3,
    e.g.: $A {>} 5$, $B {\le} 2C{-}10$, $D$ \sql{LIKE} \sql{'Eve\%'}.
\end{itemize}

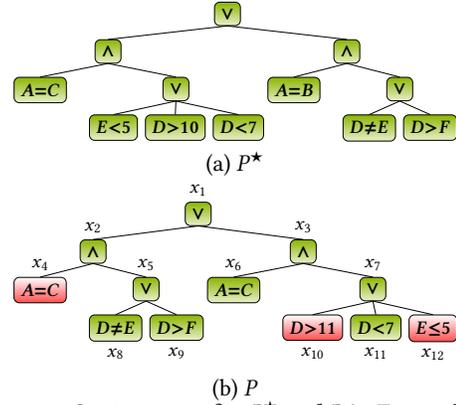
\begin{figure}[t]
\begin{minipage}[b]{0.75\linewidth}
     \centering
    \begin{tikzpicture}[
        scale=.5,
    level 1/.style={level distance=1cm, sibling distance=1mm},
    level 2/.style={sibling distance=6mm, level distance=1cm}, 
    level 3/.style={level distance=1cm,sibling distance=2mm},
    logic/.style = {shape=rectangle, rounded corners, draw, align=center, top color=applegreen, bottom color=applegreen!20},
    incorrect/.style = {shape=rectangle, rounded corners, draw, align=center, top color=white, bottom color=red!70}]
    \huge
    \Tree
    [.\node[logic] {$\boldsymbol{\lor}$}; 
        [.\node[logic] {$\boldsymbol{\land}$}; 
            [.\node[logic] {$\boldsymbol{A{=}C}$};]
            [.\node[logic] {$\boldsymbol{\lor}$};
                [.\node[logic] {$\boldsymbol{E{<}5}$};]
                [.\node[logic] {$\boldsymbol{D{>}10}$};]
                [.\node[logic] {$\boldsymbol{D{<}7}$};]
            ]
        ]
        [.\node[logic] {$\boldsymbol{\land}$}; 
            [.\node[logic] {$\boldsymbol{A{=}B}$};]
            [.\node[logic] {$\boldsymbol{\lor}$};
                [.\node[logic] {$\boldsymbol{D{\neq}E}$};]
                [.\node[logic] {$\boldsymbol{D{>}F}$};]
            ]
        ]
    ]
    \end{tikzpicture}\\
    (a)~$P^\star$
\end{minipage}

\begin{minipage}[b]{0.75\linewidth}
     \centering
    \begin{tikzpicture}[
        scale=.5,
    level 1/.style={level distance=1cm, sibling distance=1mm},
    level 2/.style={sibling distance=6mm, level distance=1cm}, 
    level 3/.style={level distance=1cm,sibling distance=2mm},
    logic/.style = {shape=rectangle, rounded corners, draw, align=center, top color=applegreen, bottom color=applegreen!20},
    incorrect/.style = {shape=rectangle, rounded corners, draw, align=center, top color=white, bottom color=red!70}]
    \huge
    \Tree
    [.\node[logic, label=above:{$x_1$}] {$\boldsymbol{\lor}$}; 
        [.\node[logic, label=above:{$x_2$}] {$\boldsymbol{\land}$}; 
            [.\node[incorrect, label=above:{$x_4$}] {$\boldsymbol{A{=}C}$};]
            [.\node[logic, label=above:{$x_5$}] {$\boldsymbol{\lor}$};
                [.\node[logic, label=below:{$x_8$}] {$\boldsymbol{D{\neq}E}$};]
                [.\node[logic, label=below:{$x_9$}] {$\boldsymbol{D{>}F}$};]
            ]
        ]
        [.\node[logic, label=above:{$x_3$}] {$\boldsymbol{\land}$}; 
            [.\node[logic, label=above:{$x_6$}] {$\boldsymbol{A{=}C}$};]
            [.\node[logic, label=above:{$x_7$}] {$\boldsymbol{\lor}$};
                [.\node[incorrect, label=below:{$x_{10}$}] {$\boldsymbol{D{>}11}$};]
                [.\node[logic, label=below:{$x_{11}$}] {$\boldsymbol{D{<}7}$};]
                [.\node[incorrect, label=below:{$x_{12}$}] {$\boldsymbol{E{\leq}5}$};]
            ]
        ]
    ]
    \end{tikzpicture}\\
    (b)~$P$
\end{minipage}
\caption{Syntax trees for $P^\star$ and $P$ in Example~\ref{ex:syntax-tree}.} 

\label{fig:disjoint-tree}
\end{figure}
\begin{Example}\label{ex:syntax-tree}
Consider the following logical formulae $P^\star$ and $P$ where $A,B,C,D,E$ are integers:
{\small
\begin{itemize}[leftmargin=*]
    \item[] $P^\star$: $(A{=}C \land (E{<}5 \lor D{>}10 \lor D{<}7) \lor (A{=}B \land (D{\neq}E \lor D{>}F))$
    \item[] $P$: $(A{=}C \land (D{\neq}E \lor D{>}F)) \lor (A{=}C \land (D{>}11 \lor D{<}7 \lor E{\leq}5))$
\end{itemize}
}
The syntax tree of $P$ is shown in \Cref{fig:disjoint-tree}.
The syntax tree of $P^\star$ would have the positions of nodes $x_2$ and $x_3$ reversed and red nodes replace with $A{=}B$, $D{>}10$, and $E{<}5$ respectively.
\end{Example}

\begin{Definition}[Repair for SQL Predicate]\label{def:where-repair-site}
Given a quantifier-free logical formulae $P$ represented as a tree, a \emph{repair} of $P$ is a pair $(\set{S}, \set{F})$
where $\set{S}$ is a set of disjoint subtrees of $P$ called the \emph{repair sites},
and $\set{F}$ is function that maps each site $x \in \set{S}$ to a new formulae $\set{F}(x)$ called the \emph{fix} for $x$.
Given a target predicate $P^\star$, a repair $(\set{S}, \set{F})$ for $P$ is \emph{correct} if applying it to $P$---i.e., replacing each $x \in \set{S}$ with $\set{F}(x)$---results in a formulae $P'$ such that $P' \Leftrightarrow P^\star$.
\end{Definition}

\begin{Definition}[Cost of a repair]\label{defn:repair-cost}
Given target predicate $P^\star$, the \emph{cost} of a repair $(\set{S}, \set{F})$ for $P$ is:
\vspace{-1mm}
\begin{equation}
\Cost(\set{S}, \set{F}) = w \cdot |\set{S}|
  + \sum_{s \in \set{S}} \frac{\dist(s, \set{F}(s))}{|P| + |P^\star|},
  \text{where}\label{eq:cost}
\vspace{-1mm}
\end{equation}
\begin{equation}
\dist(s, \set{F}(s)) = |s| + |\set{F}(s)|, \text{and}
\vspace{-1mm}
\end{equation}
$w \in \mathbb{R}^+$ controls the relative weights of the cost components.
\end{Definition}
Here, we simply define $\dist(\cdot, \cdot)$ to be the number of nodes deleted and inserted by the repair;
other notions of edit distance could be used too.
The denominator under $\dist(\cdot, \cdot)$ serves to normalize the measure relative to the sizes of the queries.
Also, note that the $w \cdot |\set{S}|$ term adds a fixed penalty for each additional repair site.
Intuitively, \oursys\ will present all repair sites (without the associated fixes) to the user as a hint.
Even a moderate number of repair sites will pose a significant cognitive challenge---if there were so many issues with $P$, we might as well ask the user to rethink the whole predicate (which would be a single repair site at the root).
In our experiments (\Cref{sec:experiments}), we set $w = 1/6$, and the number of repair sites per \WHERE\ rarely goes above two or three.

\begin{Example}\label{ex:rs-fix-cost}
Consider \Cref{fig:disjoint-tree}.
One correct repair for $P$ consists of three sites $(x_4, x_{10}, x_{12})$ and the corresponding fixes $(A{=}B, D{>}10, E{<}5)$.
The cost for this repair is $3w + \frac{3\times(1+1)}{12+12} = \frac{1}{2} + \frac{1}{4} = 0.75$.

Another correct repair for $P$ consists of two sites $(x_5, x_3)$ and the corresponding fixes $E{<}5 \lor D{>}10 \lor D{<}7$ and $A{=}B \land (D{\neq}E \lor D{>}F)$.
The cost for this repair is $2w + \frac{(4+3)+(5+6)}{12+12} = \frac{1}{3} + \frac{3}{4} \approx 1.08$.

A trivial single-site repair that replaces the entire $P$ with $P^\star$ would have cost $1w + \frac{(12+12)}{12+12} \approx 1.16$.
\end{Example}

{
\begin{algorithm2e}[t] \small

\caption{$\RepairWhere(x, x^\star, n)$}
\label{alg:find-where-repair}
\Input{a wrong predicate $x$, a correct predicate $x^\star$, and a cap $n$ on the number of repair sites}
\Output{a repair $(\set{S}, \set{F})$ with minimum cost}

\Let $\set{S}_\circ = \emptyset, \set{F}_\circ = \emptyset$\;
\Let $c_\circ$ denote the minimum cost so far, and $\infty$ initially\;
\ForEach{set $\set{S}$ of ${\le}n$ disjoint subtrees in $x$, in ascending $\card{\set{S}}$ order}{ 
    \If(\tcp*[h]{cost due to \# sites alone is already too big}){$\Cost(\set{S}, \cdot) \ge c_\circ$}{
        \Return ($\set{S}_\circ$, $\set{F}_\circ$)\tcp*{safe to stop now}
    }
    \If{$x^\star \in \CreateBounds(x, \set{S})$\label{l:createbounds-verify}}{ 
        \Let $\_, \set{F} = \DeriveFixes(x, \set{S}, x^\star, x^\star)$\;\label{l:createbounds-derive-fixes}
        \If{$c_\circ > \Cost(\set{S}, \set{F})$}{
            \Let $\set{S}_\circ = \set{S}$, $\set{F}_\circ = \set{F}$\;
        }
    }
}
\Return $(\set{S}_\circ, \set{F}_\circ)$\;
\end{algorithm2e}
}

\Cref{alg:find-where-repair} is our overall procedure for computing a minimum-cost repair for a predicate.
It considers all possible sets of repair sites,
prioritizing smaller ones because the number of repair sites heavily influences the repair cost,
and stopping once the lowest cost found so far is no greater than a conservative lower bound on the cost of the repairs to be considered.
In the worst case, the number of repairs to be considered is exponential in the size of $P$,
but in practice, the early stopping condition usually kicks in when the number of repair sites is $2$ or $3$, so the number of repairs considered is usually quadratic or cubic in $\card{P}$.

The two key building blocks of \Cref{alg:find-where-repair} are $\CreateBounds$ and $\DeriveFixes$, which we describe in more detail in the remainder of this section.
Intuitively, $\CreateBounds$ (\Cref{sec:where:bounds}) provides a quick and ``exact'' test to determine whether a given set of repair sites could ever lead to a correct repair.
If yes, $\DeriveFixes$ (\Cref{sec:where:fixes}) then finds the ``optimal'' fixes for these repair sites.
Our algorithms use Z3, so their exactness and optimality depend on Z3's completeness for the types of predicates they are given.
$\DeriveFixes$'s optimality further hinges on a Boolean minimization procedure ($\MinBoolExp$)  that it also uses.
On the other hand, since Z3 inferences are sound, progress and correctness (\Cref{sec:framework:approach}) are guaranteed.

\begin{lemma}\label{lemma:where-correctness}
\WHERE-stage hint leads to a fixed working query $Q_2$ with \WHERE\ condition that
1)~passes the viability check $P \Leftrightarrow P^\star$;
2)~satisfies $\FrWh{Q_2} \equiv \FrWh{Q^\star}$; and
3)~leads to eventual correctness.
\end{lemma}

\begin{lemma}\label{lemma:where-optimality}
Given $P$ and $P^\star$,
assuming that Z3 inference is complete with respect to the logic exercised by $P$ and $P^\star$,
and that $\MinBoolExp$ finds a minimum-size Boolean formula equivalent to its given input,
the repair returned by $\RepairWhere(P, P^\star, |P|)$ is optimal (i.e., has the lowest possible cost)
if there exists an optimal repair that either contains a single site or has all its sites sharing the same parent in $P$.
\end{lemma}
Note that \Cref{lemma:where-optimality} provides optimality for two important cases that commonly arise in practice:
1)~$P$ makes a single (presumably small) mistake;
2) $P$ is either conjunctive or disjunctive (because all atomic-predicate nodes share the same $\land$ or $\lor$ parent node).

\subsection{Viability of Repair Sites}
\label{sec:where:bounds}

\begin{savenotes}
\begin{algorithm2e}[t] \small
\caption{$\CreateBounds(x, \set{S})$}
\label{alg:create-bounds}
\Input{a predicate $x$, and a set $\set{S}$ of disjoint subtrees (repair sites) of $x$}
\Output{lower and upper bounds for $x$ achievable by fixing $\set{S}$}
\lIf{$x \in \set{S}$}{
    \Return $[\False, \True]$
}\lElseIf{$x$ is atomic}{
    \Return $[x, x]$
}\uElseIf{$\op(x) \in \{\land, \lor\}$}{
    \ForEach{$c \in \Children(x)$}{
        \Let $[l_c, u_c] = \CreateBounds(c, \set{S}[c])$;%
        \footnote{For node $x$ in $P$, $\set{S}[x]$ denotes the subset of $\set{S}$ that belong to the subtree rooted at $x$.}
    }
    \Return $[\Theta_{c \in \Children(x)} l_c,\; \Theta_{c \in \Children(x)} u_c]$ where $\Theta = \op(x)$\nllabel{l:CreateBounds:and-or};
}\Else(\tcp*[h]{$\op(x)$ is $\lnot$}){
    \Let $c = \Children(x)[0]$\tcp*{the only child of $x$}
    \Let $[l_c, u_c] = \CreateBounds(c, \set{S}[c])$\;
    \Return $[\lnot u_c,\; \lnot l_c]$\nllabel{l:CreateBounds:not};
}
\end{algorithm2e}
\end{savenotes}

The key idea is that, given a set of repair sites in $P$, we can quickly compute a ``bound'' that precisely defines what can be accomplished by \emph{any} fixes at these sites (and only at these sites).
We first give the definition of bounds and introduce some notations.
Give quantifier-free logical formulae $P_\bot$, $P$, and $P_\top$ such that $P_\bot \Rightarrow P \Rightarrow P_\top$,
we say that $[P_\bot, P_\top]$ is a \emph{bound} for $P$, denoted $P \in [P_\bot, P_\top]$.
We call $P_\top$ an \emph{upper bound} of $P$ and $P_{\bot}$ a \emph{lower bound} of $P$.

$\CreateBounds(P, \set{S})$ (\Cref{alg:create-bounds}) computes a precise bound for any predicate that can be obtained by fixing $P$ at the given set $\set{S}$ of repair sites.
It works by computing a bound for each node in $P$ in a bottom-up fashion, starting from the repair sites or leaves of $P$.
We call these bounds \emph{repair bounds}.
Intuitively, the repair bound at a repair site would be $[\False, \True]$, because a fix can change it to any logical formula.
If a subtree contains no repair sites underneath, it would have a very tight repair bound of $[p, p]$, where $p$ denotes the formulae corresponding to the subtree, which is unchangeable by the given repair.
The internal logical nodes combine and transform these bounds in expected ways in \Cref{alg:create-bounds}.

\begin{Example}\label{ex:create-bounds}
Given repair sites $\{ x_4, x_{10}, x_{12} \}$ for $P$ in \Cref{fig:disjoint-tree}, $\CreateBounds$ computes the repairs bounds shown below.

\centerline{ \small
\begin{tabular}{|c|c|c|}\hline
Node(s) & repair lower bound & repair upper bound\\\hline\hline
$x_4$ & $\False$ & $\True$\\\hline
$x_8, x_9, x_5$ & \multicolumn{2}{c|}{original predicate in $P$}\\\hline
$x_2$ & $\False$ & $D{\neq}E \lor D{>}F$\\\hline
$x_6$ & \multicolumn{2}{c|}{original predicate in $P$}\\\hline
$x_{10}$ & $\False$ & $\True$\\\hline
$x_{11}$ & \multicolumn{2}{c|}{original predicate in $P$}\\\hline
$x_{12}$ & $\False$ & $\True$\\\hline
$x_7$ & $D{<}7$ & $\True$\\\hline
$x_3$ & $A{=}C \land D{<}7$ & $A{=}C$\\\hline
$x_1$ ($P$) & $A{=}C \land D{<}7$ & $D{\neq}E \lor D{>}F \lor A{=}C$\\\hline
\end{tabular}
}
\end{Example}
The following shows that repair bounds computed by \CreateBounds\ are valid.
The proof uses an induction on the structure of $P$.
\begin{lemma}[Validity of Repair Bounds]\label{lemma:correctness-create-bounds}
Given a predicate $P$ and a set $\set{S}$ of repair sites,
$\CreateBounds(P, \set{S})$ outputs two predicates $P_{\bot}$ and $P_{\top}$,
such that applying any repair $(\set{S}, \set{F})$ (with the given $\set{S}$) will result in a predicate $P' \in [P_{\bot}, P_{\top}]$.
\end{lemma}
\Cref{lemma:correctness-create-bounds} immediately yields a method for deciding whether a candidate set $\set{S}$ of repair sites is viable:
if the target formula $P^\star \notin [P_\bot, P_\top]$ given $\set{S}$, then there does not exist a set of correct fixes $\set{F}$ for $\set{S}$.
The next natural question to ask is: if the target formula $P^\star \in [P_\bot, P_\top]$ given $\set{S}$, is it always possible to find some correct fixes?
The answer to this question is yes---and \Cref{sec:where:fixes} will provide constructive proof.
Hence, repair bounds provide a \emph{precise} test of whether a set $\set{S}$ of repair sites is viable.

For example, continuing from Example~\ref{ex:create-bounds}, using Z3, it is easy to verify that 
$P^\star \in [A{=}C \land D{<}7,\; D{\neq}E \lor D{>}F \lor A{=}C]$;
therefore, $\{x_4, x_{10}, x_{12}\}$ is a viable set of repair sites for $P$ with respect to $P^\star$.

\subsection{Derivation of Fixes}
\label{sec:where:fixes}

{
\begin{algorithm2e}[t] \small
\caption{$\DeriveFixes(x, \set{S}, l^\star, u^\star)$}
\label{alg:derive-fixes}
\Input{a predicate $x$, a set $\set{S}$ of disjoint subtrees (repair sites) of $x$,
    and a target bound $[l^\star, u^\star]$ for $x$ to achieve by fixes}
\Output{a repair represented as a set of $(s,f)$ pairs, one for each $s \in \set{S}$}

\lIf{$x \in \set{S}$ \nllabel{l:DeriveFixes:base1:1}}{
    \Return $\{(x, \MinFix(l^\star, u^\star))\}$ \nllabel{l:DeriveFixes:base1:2}
}\lElseIf{$x$ is atomic \nllabel{l:DeriveFixes:base2:1}}{
    \Return $\emptyset$ \nllabel{l:DeriveFixes:base2:2}
}\ElseIf{$\op(x)$ is $\lnot$ \nllabel{l:DeriveFixes:not:start}}{
    \Let $c = \Children(x)[0]$\tcp*{the only child of $x$}
    \Return $\DeriveFixes(c, \set{S}[c], \lnot u^\star_0, \lnot l^\star_0)$\nllabel{l:DeriveFixes:not:end};
}

\Let $\Theta = \op(x)$ \tcp*{either $\land$ or $\lor$ at this point}
\ForEach{$c \in \Children(x)$}{
    \Let $[l_c, u_c] = \CreateBounds(c, \set{S}[c])$\;
}
\Let $\set{R} = \Children(x) \cap \set{S}$ \tcp*{children of $x$ being repaired}
\lIf{$\set{R} = \emptyset$}{
    \Let $r = \emptyset$ and $\set{C} = \Children(x)$%
}\Else(\tcp*[h]{treat all children being repaired as one}){
    \Let $r = \Theta_{c \in \set{R}} c$ and $[l_r, u_r] = [\False, \True]$\;
    \Let $\set{C} = \Children(x) \setminus \set{R} \union \{r\}$\;
}
\Let $\set{F} = \emptyset$ \tcp*{result set of $(s,f)$ pairs to be computed}
\ForEach{$c \in \set{C}$\nllabel{l:DeriveFixes:recurse:start}}{
    \tcp{Combine bounds from all other children:}
    \Let $[l', u'] = [\Theta_{c' \in \set{C} \setminus \{c\}} l_{c'},\;\Theta_{c' \in \set{C} \setminus \{c\}} u_{c'}]$\;
    \uIf{$\Theta$ is $\land$ \nllabel{l:DeriveFixes:and:start}}{
        \Let $l_c^\star = l^\star $;
        \Let $u_c^\star = u_c \land (u^\star \lor \lnot u')$\;\nllabel{l:DeriveFixes:and:end}
    }\Else(\tcp*[h]{$\Theta$ is $\lor$} \nllabel{l:DeriveFixes:or:start}){
        \Let $l_c^\star = l_c \lor (l^\star \land \lnot l') $;
        \Let $u_c^\star = u^\star$\; \nllabel{l:DeriveFixes:or:end}
    }
    \lIf{$c$ is not $r$}{
        \Let $\set{F} = \set{F} \cup \DeriveFixes(c, \set{S}[c], l_c^\star, u_c^\star)$%
    }\lElse{\nllabel{l:DeriveFixes:recurse:end}
        \Let $\set{F} = \set{F} \cup \DistributeFixes(\MinFix(l_c^\star, u_c^\star), \set{C})$ %
    }
}
\Return $\set{F}$\;

\end{algorithm2e}
}

Suppose the target formula $P^\star$ falls within the repair bound $[P_\bot, P_\top]$ computed by $\CreateBounds(P, \set{S})$. 
We now introduce \DeriveFixes\ (\Cref{alg:derive-fixes}) that computes correct fixes $\set{F}$ for $\set{S}$.
The idea is to traverse $P$'s syntax tree top-down and derive a \emph{target bound} for each node $x$.
As long as we repair subtrees rooted at $x$'s children such that the resulting predicates fall within their respective target bounds,
we will have a repair for $x$ that makes its result predicate fall within $x$'s target bound.
We start from $P$'s root with the desired target bound $[P^\star, P^\star]$ and ``push it down'';
whenever we reach a repair site, its fix would simply be the smallest formula (found by \MinFix)
that falls within the target bound we have derived for the repair site.

The intuition behind how to ``push down'' the target bound at node $x$ to its children is as follows.
First, the repair bound on a child $c$ of $x$ dictates what repairs are possible---the target bound we set for $c$ must be bound by its repair bound.
However, we want to tighten the repair bound as little as possible because a looser target bound gives \MinFix\ more freedom in finding a small formula.
As a simple example, consider the target bound $[a_1 \land a_2,\; (a_1 \land a_2) \lor a_3]$, where $a_1, a_2, a_3$ are independent atomic predicates.
The smallest formula within this bound is $a_1 \land a_2$.
However, if the target bound were looser, e.g., $[a_1 \land a_2 \land a_3,\; (a_1 \land a_2) \lor a_3]$,
the smallest formula within this new bound would be just $a_3$, smaller than before.

Lines~\ref{l:DeriveFixes:recurse:start}--\ref{l:DeriveFixes:recurse:end} of \Cref{alg:derive-fixes} spells out our strategy.
We will illustrate the key ideas with \Cref{ex:create-bounds} and \Cref{fig:disjoint-tree}.
Consider pushing down the target bound of $[P^\star, P^\star]$ at $x_1$ to $x_2$ and $x_3$.
Note that our choices of target bounds for $x_2$ and $x_3$ are constrained by their respective repair bounds in the table of \Cref{ex:create-bounds};
in general, we will need to raise these lower bounds and/or lower these upper bounds in a way such that any repairs on $x_2$ and $x_3$ within these bounds ensure that $x_1$'s target bound is met.
Let us focus on setting the target bound for $x_2$.
As argued above, we would like it to be as loose as possible.
Thankfully, because $x_1 \Leftrightarrow x2 \lor x3$, $x_3$ can help ``cover'' some of $x_1$.
Specifically, no matter how we end up repairing $x_3$, we know it is lower-bounded by $A{=}C \land D{<}7$ (denote this formula by $l'$).
Hence, $x_3$ will certainly cover the $P^\star \land l'$ part of $P^\star$,
leaving $x_2$ responsible to cover only $P^\star \land \lnot l'$.
This observation motivates us to set the lower target bound for $x_2$ by raising its lower repair bound (denote it by $l_c$)
to $l_c \lor (P^\star \land \lnot l')$ (Line~\ref{l:DeriveFixes:or:end}) instead of all the way up to $l_c \lor P^\star$.
On the other hand, $x_3$ does not help with setting the upper target bound for $x_2$.
We have to set $x_2$'s upper target bound to $P^\star$, because if $x_2$ ``overshoots'' $P^\star$,
$\lor$-ing it with any $x_3$ formula will not bring it down.
In sum, we set the target bound for $x_2$ as $[l_c \lor (P^\star \land \lnot l'), P^\star] = [P^\star \lnot(A{=}C \land D{<}7), P^\star]$.
A symmetric argument leads to setting the target bound for $x_3$ as $[(A{=}C \land D{<}7) \lor P^\star, P^\star]$
(in this case $x_2$ offers no help to $x_3$ because it is lower-bounded only by \False).
The intuition behind pushing the target bound through $\land$ is analogous to that described above for $\lor$
but instead boils down to lowering upper bounds as little as possible (as opposed to raising lower bounds).
Completing the rest of \Cref{ex:create-bounds}, we show the target bounds derived by \DeriveFixes\ for $P$ given repair sites $\{ x_4, x_{10}, x_{12} \}$ in Table~\ref{tab:derivefixes-bounds}.

\begin{table}[t]
    \centering    \centerline{\small
\begin{tabular}{|c|c|c|}\hline
Node(s) & target lower bound & target upper bound\\\hline\hline
$x_1$ ($P$) & $P^\star$ & $P^\star$\\\hline
$x_2$ & $P^\star \land \lnot (A{=}C \land D{<}7)$ & $P^\star$\\\hline
$x_4$ & $P^\star \land \lnot (A{=}C \land D{<}7)$ & $P^\star \lor \lnot (D{\neq}E \lor D{>}F)$\\\hline
$x_5, x_8, x_9$ & \multicolumn{2}{c|}{same as in original predicate}\\\hline
$x_3$ & $P^\star \lor (A{=}C \land D{<}7)$ & $P^\star$\\\hline
$x_6$ & \multicolumn{2}{c|}{same as in original predicate}\\\hline
$x_7$ & $P^\star \lor (A{=}C \land D{<}7)$ & $P^\star \lor \lnot(A{=}C)$\\\hline
$x_{10} \land x_{12}$ & $P^\star \land \lnot(D{<}7)$ & $P^\star \lor \lnot(A{=}C)$\\\hline
$x_{11}$ & \multicolumn{2}{c|}{same as in original predicate}\\\hline
\end{tabular}
}
\caption{Target lower and upper bounds in Example~\ref{ex:syntax-tree}}
    \label{tab:derivefixes-bounds}
\end{table}



Another aspect of \DeriveFixes\ worth mentioning is its handling of the case when multiple repair sites have the same $\land$ or $\lor$ parent
(which is common because many queries in practice are conjunctive; therefore, their trees have only two levels- the root and the leaves).
Since $\land$ and $\lor$ are commutative, all such sites can be combined into effectively one site ($r$ in \Cref{alg:derive-fixes}) to be fixed.
In \Cref{ex:syntax-tree} above, $x_{10}$ and $x_{12}$ are handled in this manner.
Once we obtain a fix for $r$ using \MinFix\ (in conjunctive normal form for $\land$ or disjunctive normal form for $\lor$),
\DistributeFixes\ distributes the $r$'s clauses to the repair sites (Line \ref{l:DeriveFixes:recurse:end})
based on syntactic similarities between them.\label{sec:where:distribute-fixes}

The following is the main result of this section, which affirms that so long as a candidate set $\set{S}$ of repair sets passes the repair bound check in \Cref{sec:where:bounds}, there must exist a correct repair for $\set{F}$ and \DeriveFixes\ will find it.
This lemma and \Cref{lemma:correctness-create-bounds} together imply that our repair bound check is \emph{exact}.

\begin{lemma}[Existence of Correct Repair]\label{lemma:correctness-derive-fixes}
Suppose $P^\star \in \CreateBounds(P, \set{S})$.
$\DeriveFixes(P, \set{S}, P^\star, P^\star)$ returns $\set{F}$
such that applying $(\set{S},\set{F})$ to $P$ yields a formula equivalent to $P^\star$.
\end{lemma}

In the remainder of this section, we first focus on \MinFix, which \DeriveFixes\ uses to find the smallest formula within a target bound.
We end with a discussion of complexity, optimality, and, when we cannot guarantee optimality, techniques to mitigate suboptimality.

\mypar{Finding Smallest Formula with a Bound}
Given a target bound $[l^\star, u^\star]$ for a repair site, \MinFix\ needs to find a formula $g$ with the smallest size possible such that $g \in [l^\star, u^\star]$.
This goal is intimately related to the \emph{Boolean minimization} problem, which has been well studied and known to be hard~\cite{DBLP:journals/jcss/BuchfuhrerU11}.
Many practically effective tools have been developed over the years, so our strategy is to leverage these tools for \oursys.
There are two technical challenges:
1)~Boolean minimization is formulated in terms of expressions involving independent Boolean variables, while our formulae involve atomic predicates whose truth values are not independent.
2)~Our minimization problem is given a bound as opposed to a single expression that Boolean minimization typically expects.

To address (1), we run a heuristic procedure using Z3 to identify a set $\set{A}$ of ``unique'' atomic predicates that appear in $l^\star$ and $u^\star$;
those that are logically equivalent to others or can be expressed easily in terms of others (e.g., with a negation) are excluded.
This procedure does not need to detect or remove intricate dependencies (such that $A{>}C$ follows from $A{>}B$ and $C{\le}B$); any such dependencies will still be caught later.
Then, we map each predicate in $\set{A}$ to a unique Boolean variable and convert $l^\star$ and $u^\star$ into Boolean expressions involving these variables.

To address (2), we note that many practical Boolean minimization tools accept the specification of Boolean expressions as truth tables with possible \emph{don't-care} output entries.
Our idea is to use \emph{don't-cares} to encode the constraint implied by the target bound.
Specifically, we generate a truth table whose rows correspond to truth assignments of the Boolean variables for $\set{A}$.
If a particular assignment is not feasible (which is testable in Z3) due to interacting atomic predicates, we mark the output for the row as \emph{don't-care}.
For each feasible assignment, if $l^\star$ and $u^\star$ evaluate to the same truth value, we designate the output for that row to be this value.
If $l^\star$ evaluates to \False\ and $u^\star$ evaluates to \True, we mark the output as \emph{don't-care}---reflecting the flexibility offered by the bound.
(Note that because $l^\star \Rightarrow u^\star$, the case where $l^\star$ and $u^\star$ evaluate to \True\ and \False\ respectively cannot occur.)

The current implementation of \oursys\ uses \emph{ESPRESSO}~\cite{brayton1982comparison} as the primitive \MinBoolExp\ for finding a minimum-size Boolean expression given a truth table with \emph{don't-cares}.

{\bf Complexity and Optimality.}
In our analysis below, let $\kappa$ denote the combined size of formulae $P$ and $P^\star$.
\DeriveFixes's
main cost comes from calls to \MinFix\ and Z3.
The number of times that \MinFix\ is invoked is $\card{\set{S}}$, which is $\BigO(\kappa)$ but is usually a small constant in practice.
\MinFix\ runs in time exponential in the number of Boolean variables, which is capped at $\kappa$.
To construct the input truth table for \MinBoolExp, \MinFix\ will also call Z3 $O(2^\kappa)$ times.
Each Z3 call may take time exponential in the length of its input, though in practice, we time out with an inconclusive answer.
Finally, as discussed at the beginning of \Cref{sec:where}, the number of calls to \DeriveFixes\ by \RepairWhere\ can be worst-case exponential in $\kappa$, but in practice it will be $\BigO(\kappa^3)$.
Regardless, the overall complexity of \RepairWhere\ is exponential in the complexity of the \WHERE\ predicates.
Although this worst-case complexity seems daunting, we have found that \oursys\ delivers acceptable performance in practice:
thankfully, $\kappa$ is often small,
and the structures of $P$ and $P^\star$ and the interdependencies among their atomic predicates tend to be much simpler than, e.g., our \Cref{ex:syntax-tree}.

The optimality result is presented earlier as \Cref{lemma:where-optimality}.
Intuitively, the guarantees (which still depend on the primitives Z3 and \MinBoolExp) stem from two observations:
1)~if repair is limited to a single site, the target bound computed by \DeriveFixes\ is indeed the best one can do; and
2)~if all sites share the same parent, \DeriveFixes\ would effectively process them as a single site.
However, target bounds for non-combinable repair sites cannot be set optimally in an independent manner;
the approach taken by \DeriveFixes, which essentially assumes that siblings receive the least amount of help possible from each other
when pushing down target bounds, cannot guarantee a minimum-size repair.
Indeed, our running example \Cref{ex:syntax-tree} with repair sites $\{ x_4, x_{10}, x_{12} \}$ is an instance where \DeriveFixes\ fails to set target bounds optimally, because $x_4$ has a different parent from $x_{10}$ and $x_{12}$.
To mitigate this problem, we have developed a more sophisticated algorithm (called \DeriveFixesOPT) for finding fixes for multiple sites holistically.
A full discussion of \DeriveFixesOPT\ is in the appendix.
%
%
\DeriveFixesOPT\ increases the complexity by another factor of $2^{\card{\set{S}}}$.
It is heuristic in nature (as it prioritizes repair sites by how constrained they are) and cannot guarantee optimality beyond \Cref{lemma:where-optimality}.
However, it does well in practice and better than \DeriveFixes.
Since $\card{\set{S}}$ is small in practice, the complexity overhead is a good price to pay.

\begin{Example}\label{ex:derive-fixes}
In \Cref{ex:syntax-tree}, for repair sites $\{ x_4, x_{10}, x_{12} \}$,
\DeriveFixes\ returns fixes
$x_4 \mapsto A{=}B \lor (A{=}C \land D{>}10) \lor (A{=}C \land D{<}7)$;
$x_{10} \mapsto (A{=}B \land D{\neq}E) \lor (A{=}B \land D{>}F)$;
$x_{12} \mapsto (A{=}C \land D{>}10) \lor (A{=}C \land E{<}5)$.

On the other hand, \DeriveFixesOPT\ finds the optimal fixes 
$x_4 \mapsto A{=}B$; $x_{10} \mapsto D{>}10$; $x_{12} \mapsto E{<}5$.
\end{Example}

\section{\GROUPBY\ Stage}\label{sec:handling-groupby}

We check the \GROUPBY\ equivalence assuming $Q^\star$, $Q$ have equivalent \FROM\ and \WHERE\ clauses.
We focus on ensuring $\FrWhGr{Q} \equiv \FrWhGr{Q}$,
regardless of the order and the number of expressions involved in their \GROUPBY\ clauses.

\cut{
The following lemma
shows the necessity of fixing \sql{GROUP BY} clause (proof in the full version). 

\begin{lemma}\label{lemma:aggr-necessity}
Consider two single-block SQL queries $Q_1$ and $Q_2$, where $Q_1$ has no \sql{GROUP} \sql{BY} or aggregation,
while $Q_2$ has \sql{GROUP} \sql{BY} and/or aggregation but no \sql{HAVING}.
$Q_1$ and $Q_2$ cannot be equivalent under bag semantics, assuming that no database constraints are present and there exists some database instance for which either $Q_1$ or $Q_2$ returns a non-empty result.
\end{lemma}
}

In the following, we consider the case where both $Q$ and $Q^\star$ have grouping and/or aggregation.
Suppose we have unified the \sql{WHERE} conditions and \sql{GROUP} \sql{BY} expressions in the two queries according to the table mapping $\mapping$.
Let $P$ denote the resulting formula for $Q^\star$'s \sql{WHERE} condition (which at this point is logically equivalent to $Q$'s),
and let $\vec{o}$ and $\vec{\ostar}$ denote the resulting lists of \sql{GROUP} \sql{BY} expressions for $Q$ and $Q^\star$, respectively.
Note that the ordering of the \sql{GROUP} \sql{BY} expressions is unimportant.
Also, if a query involves aggregation but has no \sql{GROUP} \sql{BY}, we consider the list of \sql{GROUP} \sql{BY} expressions to be an empty list.
Same column references across $P$, $\vec{o}$, and $\vec{\ostar}$ are treated as same variables.
Our goal is to compute a subset $\Delta^-$ of \sql{GROUP} \sql{BY} expressions to be removed from $Q$,
as well as a set $\Delta^+$ of additional \sql{GROUP} \sql{BY} expressions to be added to $Q$,
such that the resulting query will always produce the same grouping of intermediate result tuples (produced by \sql{FROM}-\sql{WHERE}) as $Q^\star$.
In practice, we may not want to reveal $\Delta^+$, but instead simply hint that $Q$ misses some \sql{GROUP} \sql{BY} expressions.
We may repeat the hinting process several times until \sql{GROUP} \sql{BY} is completely fixed.

Repairing grouping is trickier than it seems because seemingly very different \sql{GROUP} \sql{BY} lists can produce equivalent grouping, as illustrated by the following example.
\begin{example}\label{ex:groupby}
Consider two queries over tables $\sql{R}(\sql{A},\sql{B})$ and $\sql{S}(\sql{C},\sql{D})$:
\begin{lstlisting}[language=SQL,escapeinside=``,basicstyle=\small\ttfamily]
SELECT B FROM R, S WHERE B=C GROUP BY B, D; -- `$Q^\star$`
SELECT C FROM R, S WHERE B=C GROUP BY C+D, C; -- `$Q$`
\end{lstlisting}
The two queries are equivalent, even though none of the pairs of \sql{GROUP} \sql{BY} expressions are equivalent when examined in isolation.
\end{example}
To address this challenge, instead of comparing pairs from $\vec{\ostar}$ and $\vec{o}$ in isolation,
we holistically consider these lists as well as the \sql{WHERE} condition,
and go back to the definition of \sql{GROUP} \sql{BY} as computing a partitioning of intermediate result tuples.
Formally, the viability check for this stage is that $\vec{o}$ and $\vec{\ostar}$ achieve the same partitioning,
or more precisely:
$$\textstyle V_3: \text{Check if}\; \forall t_1, t_2 \in \FrWh{Q^\star}: \bigwedge_i (o_i[t_1]\!=\!o_i[t_2]) \Leftrightarrow \bigwedge_i (\ostar_i[t_1]\!=\!\ostar_i[t_2])$$
Here, $t_1$ and $t_2$ denote intermediate result tuples, which are known to satisfy $P$;
we use $o[t]$ to denote 
evaluating $e$ over $t$.%
\footnote{Formally, we treat $t$ as an assignment of variables (column references) in $e$ to variables representing corresponding column values in $t$.
Hence, $e[t]$ is an expression obtained from $e$ by replacing each variable (column reference) $v$ with variable $t(v)$.}
This approach underlines our algorithm \FixGrouping\ (Algorithm~\ref{alg:fix-grouping}).
\begin{Example}\label{ex:groupby-logic}
Consider the two queries in Example~\ref{ex:groupby}.
The table mapping is trivial and we simply use column names to name variables.
We have: $P$ is $B=C$, $\vec{\ostar} = [ B, D ]$, and $\vec{o} = [ C+D, C ]$.
The logical statement that establishes the equivalence of grouping is
\begin{sizeddisplay}{\footnotesize}
\begin{gather*}
\forall (A_1, B_1, C_1, D_1), (A_2, B_2, C_2, D_2):\\
\left( B_1{=}C_1 \land B_2{=}C_2 \right) \;\;\text{\myComment both $(A_1, B_1, C_1, D_1)$ and $(A_2, B_2, C_2, D_2)$ satisfy $P$}\\
\Rightarrow
\left(
\begin{matrix*}
& \left( B_1{=}B_2 \land D_1{=}D_2 \right) \;\;\text{\myComment $Q^\star$'s grouping criterion}\\
& \Leftrightarrow \left( C_1{+}D_1{=}C_2{+}D_2 \land C_1{=}C_2 \right) \;\;\text{\myComment $Q$'s grouping criterion}
\end{matrix*}
\right).
\end{gather*}
\end{sizeddisplay}
Note that instead of referring to tuples $t_1$ and $t_2$, we simply refer to variables representing their column values in the above.
\end{Example}

\begin{algorithm2e}[t] \small
\caption{$\FixGrouping(P, \vec{o}, \vec{\ostar})$}
\label{alg:fix-grouping}
\Input{a formula $P$ and two expression lists $\vec{o}$ and $\vec{\ostar}$}
\Output{a pair $(\Delta^-, \Delta^+)$,
where $\Delta^- \subseteq [1..\dim(\vec{o})]$ is a subset of indices of $\vec{o}$
and $\Delta^+ \subseteq [1..\dim(\vec{\ostar})]$ is a subset of indices of $\vec{\ostar}$}
\Let $\vec{v}$ denote the set of variables in $P$, $\vec{o}$, and $\vec{\ostar}$\;
\Let $t_1, t_2$ be two assignments of $\vec{v}$ to new sets of variables $\vec{v}_1$ and $\vec{v}_2$\;
\Let $G^\star$ denote the formula $\bigwedge_i (\ostar_i[t_1]\!=\!\ostar_i[t_2])$\;
\Let $\Delta^- = \emptyset$\;
\ForEach{$o_i \in \vec{o}$}{
    \If{$\IsSat(P[t_1] \land P[t_2] \land G^\star \land o_i[t_1]\!\neq\!o_i[t_2])$ \nllabel{l:fix-grouping:redundant}}{
        \Let $\Delta^- = \Delta^- \union \{i\}$\;
    }
}
\Let $G$ denote the formula $\bigwedge_{i \not\in \Delta^-} (o_i[t_1]\!=\!o_i[t_2])$\;
\Let $\Delta^+ = \emptyset$\;
\ForEach{$\ostar_i \in \vec{\ostar}$}{
    \If{$\IsSat(P[t_1] \land P[t_2] \land G \land \ostar_i[t_1]\!\neq\!\ostar_i[t_2])$ \nllabel{l:fix-grouping:missing}}{
        \Let $\Delta^+ = \Delta^+ \union \{i\}$\;
        \Let $G = G \land \ostar_i[t_1]\!\neq\!\ostar_i[t_2]$\;
    }
}
\Return $(\Delta^-, \Delta^+)$\;
\end{algorithm2e}

In \FixGrouping, to find $\Delta^-$, which are ``wrong'' expressions in $\vec{o}$,
we check, for each $o_i$, whether it is possible that given $P[t_1] \land P[t_2]$, we can have $\bigwedge_i (\ostar_i[t_1]\!=\!\ostar_i[t_2])$ but not $o_i[t_1]\!=\!o_i[t_2]$.
If yes, that means $o_i$ is wrong with respect to $\ostar$, because while $t_1$ and $t_2$ should belong to the same group per $\ostar$, grouping by $o_i$ alone would have forced them into separate groups instead.
After identifying all wrong expressions in $\vec{o}$ and removing them, we are left with a partitioning potentially coarser than $\ostar$ but otherwise consistent with $\ostar$. We then find $\Delta^+$ to be further added in a similar fashion. 

\cut{
We then find $\Delta^+$, a set of \sql{GROUP} \sql{BY} expressions to be further added to make the resulting partitioning equivalent to that of $\vec{\ostar}$.
To this end, we iterate through $\ostar$, and check, for each $\ostar_i$,
whether it is possible for two tuples $t_1$ and $t_2$ in the same group as defined by the current partitioning to have $\ostar_i[t_1]\!\neq\!\ostar_i[t_2]$.
If yes, we add $\ostar_i$ to $\Delta^+$ to refine the current partitioning.}

\cut{As Lemma~\ref{lemma:group-fix} below shows, at the end of this process,
we guarantee that $(\vec{o} \setminus \Delta^- \union \Delta^+)$ produces the same partitioning as $\vec{\ostar}$,
that $\Delta^-$ and $\Delta^+$ are ``minimal'' under reasonable assumptions.}
\begin{lemma}\label{lemma:group-fix}
We say that two lists of \sql{GROUP} \sql{BY} expressions are \emph{equivalent} if they produce the same partitioning for the above query over any database instance.
Let $(\Delta^-, \Delta^+) = \FixGrouping(P, \vec{o}, \vec{\ostar})$.
Assuming that subroutine \IsSat\ returns no false positives, we have:

{\bf Correctness:} \GROUPBY-stage hint leads to a fixed working query $Q_3$ that 1) passes the viability check ($\vec{o}, \vec{\ostar}$ are equivalent), 2) satisfies $\FrWhGr{Q_3} \equiv \FrWhGr{Q^\star}$; and 3) leads to eventual correctness. 

Further assuming that \IsSat\ returns no false negatives, we have:

{\bf Strong Minimality of $\Delta^-$:} Let $(\Delta^-_\circ, \Delta^+_\circ)$ denote the minimal $\Delta^-$ and $ \Delta^+$ respectively, then for any $(\Delta^-_\circ, \Delta^+_\circ)$ such that $\vec{o} \setminus \Delta^-_\circ \union \Delta^+_\circ$ is equivalent to $\vec{\ostar}$, $\Delta^- \subseteq \Delta^-_\circ$.

{\bf Weak Minimality of $\Delta^+$:} If $\Delta^+ \neq \emptyset$, then there exists no $\Delta^-_\circ$ such that
$\vec{o} \setminus \Delta^-_\circ$ is equivalent to $\vec{\ostar}$.

\end{lemma}

The strong minimality of $\Delta^-$ means that we can hint each expression therein as a ``must-fix.''
The weak minimality of $\Delta^+$ works perfectly as we simply hint that the wrong query needs some additional \sql{GROUP} \sql{BY} expressions.

\cut{
We note that a stronger sense for minimality for $\Delta^+$ is both semantically trickier to define and computationally harder to find.
First, we can always use a single \sql{GROUP} \sql{BY} expression to achieve the same partitioning as $\vec{\ostar}$---namely,
\sql{ROW($\ostar_1$,$\ostar_2$,...)} where \sql{ROW} is the SQL row/tuple constructor.
Therefore, minimizing the cardinality of $\Delta^+$ by itself is not meaningful.
Second, even if we place restrictions on $\Delta^+$ (e.g., requiring it to be a subset of $\vec{\ostar}$ as does \FixGrouping),
or develop a better measure for the complexity of $\Delta^+$, 
the optimization problem is hard in general, as illustrated by Example~\ref{ex:groupby-minimize}.
\begin{example}\label{ex:groupby-minimize}
Consider the list of \sql{GROUP} \sql{BY} expressions $\vec{\ostar} = [ \sql{ROW(A,B,C)}, \sql{ROW(B,C,D)}, \sql{ROW(A,C,E)} ]$.
If we require $\Delta^+ \subseteq \vec{\ostar}$, the $\Delta^+$ with minimum cardiality is $\{ \sql{ROW(B,C,D)}, \sql{ROW(A,C,E)} \}$.
It is not difficult to see that finding $\Delta^+$ for similarly constructed examples amounts to solving the set cover problem.

As another example, consider $\vec{\ostar} = [ A{+}2B{+}3C, 2A{+}4B{-}C, C ]$.
A reasonable $\Delta^+$, without requiring $\Delta^+ \subseteq \vec{\ostar}$, would be $\{ A{+}2B, C \}$.
Computing this $\Delta^+$ amounts to computing a basis for a vector space.
\end{example}
Because of these complications, we believe that the weak minimality of $\Delta^+$ is a simpler and more practical guarantee.
}

\section{\sql{HAVING} Stage}\label{sec:handling-having}

At \HAVING\ stage, we aim at further ensuring that $\FrWhGrHa{G} \equiv \FrWhGrHa{G^\star}$ assuming that $Q^\star$ and $Q$ unified by a table mapping and have equivalent \FROM, \WHERE, and \GROUPBY.
While \HAVING\ can also be modeled as a logical formula, there are new challenges:
1)~unlike \WHERE, inputs to \HAVING\ formulae are arrays of tuples $[t_1, ..., t_n]$ instead of single tuples,
2)~we need to consider aggregate functions, and
3)~we cannot test \HAVING\ alone without considering \WHERE's effect.

\begin{Example}\label{ex:having-1}
Consider two queries over $\sql{R}(\sql{A},\sql{B})$ and $\sql{S}(\sql{C},\sql{D})$:
\begin{lstlisting}[language=SQL, basicstyle=\small\ttfamily]
SELECT A FROM R, S WHERE A=C AND A>4 GROUP BY A, B 
  HAVING A > B + 3 AND 2*SUM(D) > 10; -- `$Q^\star$`
SELECT A FROM R, S WHERE A=C GROUP BY A, B, C 
  HAVING C > B + 3 AND SUM(D * 2) > 10 AND A>4; -- `$Q$`
\end{lstlisting}
The two queries are equivalent because \sql{A=C} in \WHERE, because \sql{2*} distributes over \sql{SUM},
and because \sql{A>4} can be either in \WHERE\ or \HAVING.
\end{Example}

Our strategy is to construct two formulae $H^\star, H$ for the \sql{HAVING} conditions of $Q^\star, Q$ respectively,
such that equivalence of $H^\star$ and $H$ implies $\FrWhGrHa{G} \equiv \FrWhGrHa{G^\star}$.
To this end, for each reference to a \GROUPBY\ column in \HAVING, we replace it with a variable from the same domain,
and we translate \HAVING\ expressions outside aggregate function calls in the same way as we handle \WHERE:
e.g., \sql{A>B+3} becomes $A{>}B{+}3$.
For each reference to a column not in \GROUPBY, we introduce an array variable to capture the fact that it refers to a collection of values from rows in the same group.
Moreover, for each aggregate function call, we introduce a new array variable to represent the collection of input values if they are computed from an expression, and we use a universally quantified assertion to relate this variable to the source column values:
e.g., for \sql{SUM(D*2)} we introduce array-valued $\mathbf{D}_2$ to represent \sql{D*2} values,
and we related it to the array-valued $\mathbf{D}$ representing \sql{D} values by asserting
$\forall i \in \mathbb{N}: \mathbf{D}_2[i] = \mathbf{D}[i] \times 2$.
Such assertions, along with the \WHERE\ condition and additional inference rules for aggregate functions, go into a context as discussed in \Cref{sec:framework} and illustrated in \Cref{ex:z3}.

\begin{Example}\label{ex:having-2}
For \Cref{ex:having-1}, \HAVING\ formulae for $Q^\star, Q$ are:
\begin{sizeddisplay}{\footnotesize}
\begin{align*}
(H^\star)&& A{>}B{+}3 \land (2{\times}\sql{SUM}(\mathbf{D}){>}10) \\
(H)&& C{>}B{+}3 \land \sql{SUM}(\mathbf{D}_2){>}10 \land A{>}4
\end{align*}
\end{sizeddisplay}
We test their equivalence under the following context:
\begin{sizeddisplay}{\footnotesize}
\begin{align*}
\Context&: \left\{\;\begin{aligned}
    \mathbf{D}, \mathbf{D}_2 \text{ have type }\narrow{Array}(\mathbb{Z})\\
    A = C \land A > 4\\
    \forall i \in \mathbb{N}: \mathbf{D}_2[i] = \mathbf{D}[i] \times 2\\
    \cline{1-1}
    \sql{SUM} \text{ has type }\narrow{Array}(\mathbb{Z}) \to \mathbb{Z}\\
    \forall c \in\mathbb{Z}, \mathbf{X} \text{ and } \mathbf{Y} \text{ of type }\narrow{Array}(\mathbf{Z}):\hspace*{10em}\\
    (\forall i \in \mathbb{N}: \mathbf{X}[i] \times c = \mathbf{Y}[i])
    \Rightarrow
    \sql{SUM}(\mathbf{X}) \times c = \sql{SUM}(\mathbf{Y})
\end{aligned}\;\right\},
\end{align*}
\end{sizeddisplay}
In the above, the assertions underneath the horizontal line are generic assertions encoding properties of aggregate functions
useful for inferring equivalences.
Only those relevant to \Cref{ex:having-1} are listed here; for a complete list see~\cite{fullversion}.
\end{Example}%


The viability check for \HAVING\ (Theorem~\ref{theorem:main}, stage 4) is that $H$ is logically equivalent to $H^\star$ under \HAVING\ base context $\Context$, i.e.:
$$V_4: \textit{Check if } H \Leftrightarrow H^\star  \textit{under \Context}$$
Note that this check implicitly applies to all groups.
If a constraint solver fails to establish equivalence, we invoke the exact same procedures as for \WHERE\ to find a repair.



\begin{lemma}\label{lemma:having-fix}
\HAVING-stage hint leads to a fixed working query $Q_4$ with \HAVING\ condition that
1)~passes the viability check;
2)~satisfies $\FrWhGrHa{Q_4} \equiv \FrWhGrHa{Q^\star}$; and
3)~leads to eventual correctness.
\end{lemma}

As with \WHERE, the correctness of the above lemma relies only on the fact that Z3 inference is sound with respect to the logic exercised by $H$, $H^\star$, and $\Context$ and that $\MinBoolExp$ always finds a Boolean formula equivalent to its given input.
We could additionally guarantee optimality similar to \Cref{lemma:where-optimality} by making the same assumptions therein (completeness of Z3 inference and optimality of $\MinBoolExp$) plus the additional assumption that the context \Context\ encodes all properties of aggregate functions relevant to inference.

\section{\SELECT\ Stage}\label{sec:handling-select}

This stage aims at fixing \SELECT\ as needed to ensure $Q \equiv Q^\star$, assuming that they already have equivalent \FROM, \WHERE, \GROUPBY\ and \HAVING.
We test the equivalence between \SELECT\ expressions with a context $\Context$ dependent on the type of the query:
if the queries are SPJ, we simply assert the \WHERE\ condition in $\Context$;
if the queries are SPJA, we use the same $\Context$ defined by the \HAVING-stage.

Let $\vec{o}$ and $\vec{\ostar}$ denote the resulting ordered lists of \SELECT\ expressions for $Q, Q^\star$, respectively. The viability check ($V_5$) is that $\dim(\vec{o}) = \dim(\vec{\ostar})$ and $\vec{o}[i]$ is equivalent to $\vec{\ostar}[i]$ for $1 \leq i \leq \dim(\vec{\ostar})$, i.e. both \SELECT{}s have the same number of expressions and expressions on the same index position are equivalent. 
If \SELECT\ clauses are not equivalent between $Q^\star, Q$, our goal becomes to compute $\Delta^-$ of \SELECT\ expression to be removed from $Q$ at the corresponding index position and $\Delta^+$ of expressions to be added to $Q$ at the corresponding index position.

The algorithm checks the equivalence between $(\vec{o}[i], \vec{\ostar}[i])$ and add $\Delta^-$ and $\Delta^+$ respectively if they are inequivalent. Finally, excessive expressions in $Q$ or $Q^\star$ will also be added to $\Delta^-$ and $\Delta^+$ respectively. After fixing \SELECT, we guarantee $Q^\star \equiv Q$.

\cut{
\begin{lemma}\label{lemma:select-fix}
We say that two lists of \SELECT\ expression are equivalent if they produce the same set of columns in the same ordering. Let $(\Delta^-, \Delta^+) = \FixSelect(P, \vec{o}, \vec{\ostar})$. Assuming that subroutine $\IsSat_\Context$ returns no false positive, we have:

{\bf Correctness:} \oursys's \SELECT-stage hint leads to a fixed working query $Q_5$ that 1) passes the viability check ($\vec{o}$ and $\vec{\ostar}$ are equivalent); 2) satisfies $Q_5 \equiv Q^\star$. This applies to both \emph{SPJ} and \emph{SPJA} queries. 

{\bf Strong minimality of $(\Delta^-, \Delta^+)$ for \emph{SPJ}} Let $(\Delta^-_\circ, \Delta^+_\circ)$ denote the minimal $(\Delta^-, \Delta^+)$ respectively, then for any $(\Delta^-_\circ, \Delta^+_\circ)$ that make $\vec{o}$ and $\vec{\ostar}$ equivalent, $\Delta^- \subseteq \Delta^-_\circ, \Delta^+ \subseteq \Delta^+_\circ$.

\end{lemma}
}

\section{Experiments}\label{sec:experiments}

We test three aspects of \oursys: coverage, accuracy, and running time.
For coverage, we test the ability of \oursys\ to fix wrong queries that arise in real-world classroom settings.
For accuracy and running time, we focus on \Cref{alg:find-where-repair}, which is the bottleneck of \oursys\ due to calls to \DeriveFixes\ or \DeriveFixesOPT.
As fix minimization incurs exponential time, we examine
1)~how the number of unique predicates affects running time,
2)~how close the generated repairs are to the optimal if queries are not conjunctive,
3)~a comparison between the running time and optimality of \DeriveFixes\ and \DeriveFixesOPT.
In general, \DeriveFixesOPT\ strives for smaller fixes and hence incurs longer running time than \DeriveFixes.

\mypar{Implementation/Test Environment} \label{sec:experiments:implementation} We implemented \oursys\ in Python 3.10 using Apache Calcite~\cite{begoli2018apache} to parse SQL queries and Z3 SMT Solver \cite{de2008z3} to test constraint satisfiability. We use ESPRESSO in PyEDA~\cite{pyeda} for fix minimization. We run the experiments locally on a 64-bit Ubuntu 20.04 LTS server with 3.20GHz Intel Core i7-8700 CPU and 32GB 2666MHz DDR4.

\mypar{Test Data Preparation} \label{sec:experiments:test-data-prep}
To prepare the first test dataset, denoted \textsf{Students}, we examined 2,000+ real student queries from an undergraduate database course in one semester at the first author's institution.
These queries came from 4 introductory-level SQL questions (with 4 reference queries), and altogether they included 341 wrong queries.
Out of these, 35 (11\%) used SQL features not supported by \oursys\ (see limitations at the end of \Cref{sec:framework}).
Hence, we end up with 306 supported wrong queries in \textsf{Students}.
(At the time of writing, we are still exploring with the institutional review board the possibility of making this dataset publicly available.)

To further expand coverage of errors, we cross-checked \textsf{Students} queries with the list of SQL issues indicative of semantic errors categorized by Brass et al.~\cite{brass2006semantic} (which did not publish a query dataset).
Out of the 43 issues in~\cite{brass2006semantic}, 18 involve SQL features not currently supported by \oursys,
but they only make up for a small minority (11.4\%) of the observed instances as reported by~\cite{brass2006semantic}.
Out of the 25 issues \oursys\ should support, 17 are already represented in the 306 \textsf{Students} queries.
To cover the remaining 8, we handcrafted two queries according to each issue and added to the dataset;
we also handcrafted corresponding reference queries (free from any issue in \cite{brass2006semantic}).
We denote the resulting dataset \textsf{Students+}, with 322 queries having errors/issues.

Our second test dataset, denoted \textsf{TPCH}, is based on TPC-H~\cite{tpch} schema and queries, with synthetic errors injected.
This dataset allows us to stress-test \oursys\ with queries that are more complex than \textsf{Students}.
Also, because errors are synthetic, we have the ``ground-truth'' repair sites and fixes, allowing us to easily assess the optimality of \oursys\ fixes.
Most \WHERE\ conditions in TPC-H queries are conjunctive:
we chose 7 TPC-H queries with conjunctions of 4,5,6,7,9,10,11 atomic predicates (TPC-H Query 4,3,10,9,5,8,21 respectively).
Since we did not find a TPC-H query with exactly 8 predicates, we synthesized one by removing one predicate from TPC-H Query 5. 
For each query, we then introduced errors into two atomic predicates to make the wrong query, which remained conjunctive.
Thus, each pair of wrong and reference queries has 6-13 unique atomic predicates. 
Furthermore, to test cases beyond conjunctive \WHERE\ conditions,
we chose TPC-H Query 7, whose \WHERE\ contains multiple nested \sql{AND} and \sql{OR},
and created 5 wrong queries by injecting 1-5 errors by changing atomic predicates or logical operators.
For fair comparison, we ensured that the number of unique atomic predicates is always 10 between the reference query and each wrong query.

\subsection{Results and Discussion}

\begin{figure}[t]
\centering
\begin{subfigure}{.53\linewidth}
\centering
\includegraphics[scale=0.28]{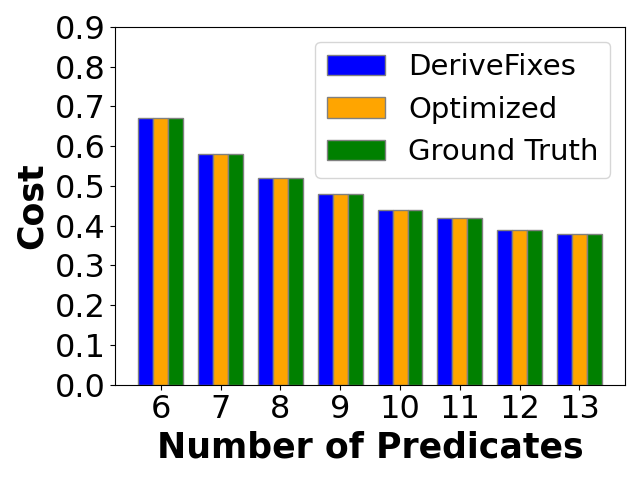}
\caption{Repair cost}
\label{fig:runtime-test-cost}
\end{subfigure}
\begin{subfigure}{.46\linewidth}
\includegraphics[scale=0.28]{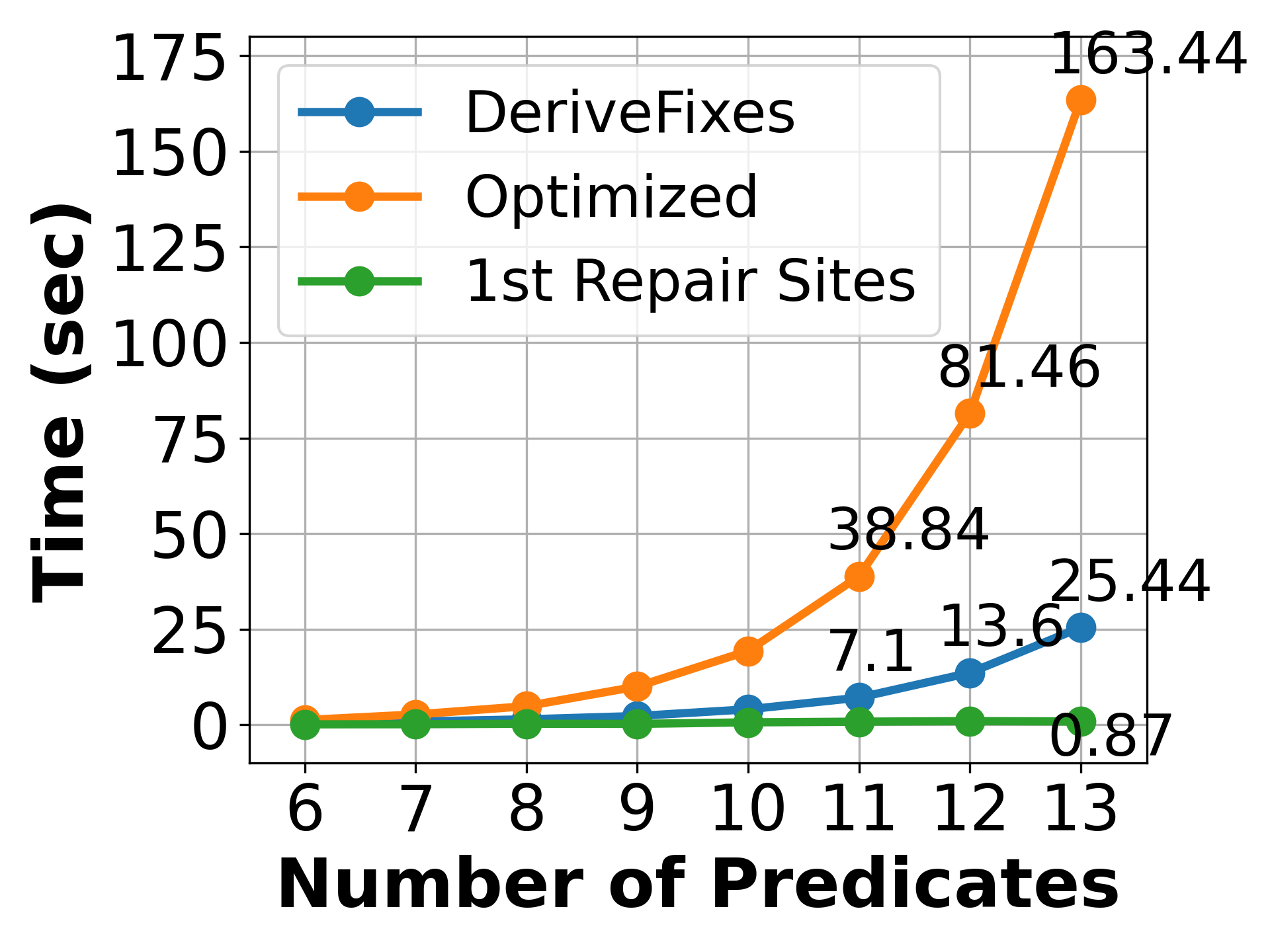}
\caption{Running time}
\label{fig:runtime-test-runtime}
\end{subfigure}
\caption{\DeriveFixes\ vs.\ \DeriveFixesOPT\ (Optimized) for conjunctive \WHERE\ (\textsf{TPCH})}
\label{fig:runtime-test}
\end{figure}

\begin{figure}[t]
\centering
\begin{subfigure}{.53\linewidth}
\includegraphics[scale=0.28]{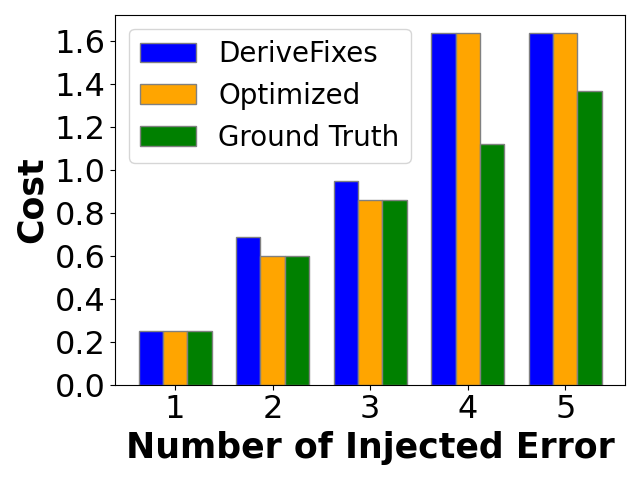}
\caption{Repair cost}
\label{fig:cost-q7}
\end{subfigure}
\begin{subfigure}{.46\linewidth}
\centering
\includegraphics[scale=0.28]{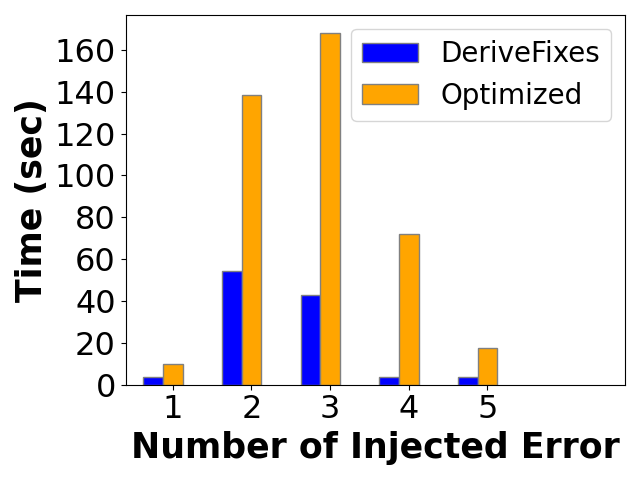}
\caption{Running time}
\label{fig:time-q7}
\end{subfigure}
\caption{\DeriveFixes\ vs.\ \DeriveFixesOPT\ (Optimized) for nested \sql{AND}/\sql{OR} (\textsf{TPCH})}\label{fig:q7-cost_vs_time}
\end{figure}

\begin{figure}[t]
\centering
\begin{subfigure}{.53\linewidth}
\includegraphics[scale=0.28]{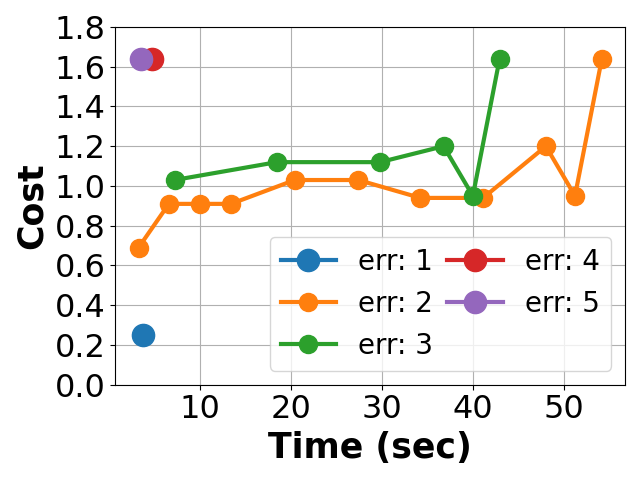}
\caption{\DeriveFixes}
\label{fig:qr-hint-cost-time}
\end{subfigure}
\begin{subfigure}{.46\linewidth}
\centering
\includegraphics[scale=0.28]{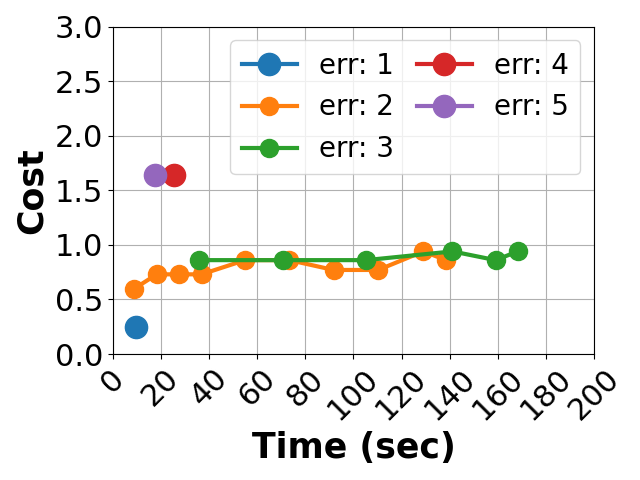}
\caption{\DeriveFixesOPT}
\label{fig:qr-hint-optimized-cost-time}
\end{subfigure}
\caption{Cost of repairs found during course of execution}
\label{fig:cost_vs_time}
\end{figure}

\mypar{\textsf{Student+}}\label{sec:experiments:coverage-test-result}
To test coverage and optimality of \oursys,
we ran \oursys\ for the 322 \textsf{Student+} queries with errors/issues, along with their reference queries,
and examined all \oursys\ fixes.
For the 25 issues in \cite{brass2006semantic} that \oursys\ should support, we found that they were handled in three ways:
1)~11 of them were indeed errors, and \oursys\ correctly identified and fixed them all;
2)~3 of them were efficiency/stylistic issues where the queries were semantically still correct (e.g., logically correct \WHERE\ containing some tautological conditions, such as \sql{A >= B OR A < B}), and \oursys\ did not flag any error;
3)~the remaining 11 of them were also efficiency/stylistic issues (e.g., unnecessarily joining a primary key with its corresponding foreign key but only projecting the foreign key column), but \oursys\ failed to detect query equivalence in this case and suggested some fixes.
This last category is the only case where \oursys\ showed suboptimal behavior, though its suggested fixes still lead to correct queries, and with the interesting side effect of resolving efficiency/stylistic issues.
The detailed analysis can be found in the appendix.
It is worth noting that \oursys\ perfectly handles all of the 10 most common issues in \cite{brass2006semantic}.

\oursys's average running time per query on \textsc{Student+} is 0.2 seconds, using \DeriveFixes.
However, note that most \textsc{Student+} queries are rather simple, with conjunctive \WHERE\ (which does not need \DeriveFixesOPT\ for optimality) and at most 5 unique atomic predicates.
Therefore, we further stress-tested \oursys\ using \textsc{TPCH}.

\mypar{\textsf{TPCH}, conjunctive \WHERE\ with varying number of atomic predicates}\label{sec:experiments:runtime-result}
Here, we study \oursys's running time and optimality (as measured by repair cost, the lower the best) as we vary the number of atomic predicates involved in repairing \WHERE.
We compare versions of \oursys\ using \DeriveFixes\ vs.\ \DeriveFixesOPT, both set to explore up to two repair sites.
\Cref{fig:runtime-test-cost} confirms that for conjunctive queries, both always return optimal repairs according to the ground truth, regardless of the size of \WHERE.
(Note that the repair cost is not proportional to the number of atomic predicates because it is normalized by the query sizes per \Cref{eq:cost}).
\Cref{fig:runtime-test-runtime} shows that as expected, both have running times exponential in the number of unique atomic predicates,
but \DeriveFixes\ runs much faster than \DeriveFixesOPT.
Furthermore, the plot labeled ``1st Repair Sites'' shows that it takes less than one second for \oursys\ to find the first \emph{viable} (not necessarily optimal) repair site, so there is additional room to trade optimality for faster running time.

\mypar{\textsf{TPCH}, \WHERE\ with nested \sql{AND}/\sql{OR} and varying number of injected errors}\label{sec:experiments:accuracy-result}
As shown in \Cref{fig:cost-q7}, when the optimal repair (according to the ground truth) involves only one repair site (a single error),
both \DeriveFixes\ and \DeriveFixesOPT\ are able to find this optimal repair, confirming \Cref{lemma:where-optimality}.
When there are more errors (2-3), \DeriveFixes\ returns suboptimal repairs while \DeriveFixesOPT\ is still able to find optimal or near-optimal repairs (for the cases of 2 and 3 errors, respectively).
However, with 4-5 errors---which are arguably not the cases \oursys\ targets---both suffer from suboptimality because they are set to explore up to two repair sites; in fact, both decided that it was best to just repair the whole \WHERE\ condition.
\Cref{fig:time-q7} shows that \DeriveFixesOPT's better optimality comes at the expense of slower speed than \DeriveFixes, however.
Interestingly, with 4-5 errors, both run faster than with 2-3 errors, because the large numbers of errors severely limit the number of possibilities of single- and 2-site repairs, speaking to the effectiveness of \CreateBounds\ in quickly spotting and bailing out of difficult situations.

Finally, \Cref{fig:cost_vs_time} shows all unpruned viable repairs found during \oursys's course of execution,
in terms of when they were found and how much they cost;
there is one trace for each execution.
Traces for 1 (blue), 4 (red), and 5 (purple) errors degenerate into single dots because \oursys\ eventually finds only one solution as viable repair options are limited.
Recall that we heuristically prioritize the viable repairs to consider,
but there is no guarantee that a cheaper repair will always be found earlier.
Hence, there are fluctuations in the repair costs over time, although the general trends are up, confirming the effectiveness of our heuristic.
Furthermore, note that the lowest-cost repairs tend to surface early during execution.
\label{sec:experiments:feasibility}
In closing, while the total and worst-case running times of \oursys\ grow exponentially in query size,
in practice the running times are reasonable considering that \oursys\ is intended for education settings,
where returning hints instantaneously may not be necessary or desirable for learning.
With the observation that \oursys\ often returns some low-cost repairs early,
we can offer them as preliminary hints to get students thinking, while \oursys\ continues to look for better repairs in the meantime.

\section{User Study}\label{sec:user-study}

We conducted a small-scale user study to evaluate \oursys: 1)~whether students can understand what is wrong with the suggested hints,
and 2)~how the hints generated by \oursys\ compare with ones provided by ``{\em expert users}'' (teaching assistants in our study). 

\mypar{Participants} \label{sec:user-study:participants}
We recruited 38 students who have taken/are taking a graduate or undergraduate database course.
Except for an incentive of receiving a small gift card and practicing SQL, the participation was voluntary.
In the end, we collected 15 complete and valid answers.
A possible explanation for the low completion rate was the significant effort required to debug SQL queries with subtle mistakes
(we observed that some participants took more than an hour to finish).
We considered the possibility of recruiting participants from other sources (e.g., Amazon Mechanical Turk),
but decided against it because they would not represent our targeted population (students).
Furthermore, given the significant effort required from the participants as observed above,
it would be hard to incentivize participants who are not actively learning SQL:
a low reward would turn them away, while a high reward might encourage undesirable behaviors.

\mypar{Preparation}\label{sec:user-study:prep} 
To design the survey, we first performed an analysis of the \textsf{Students} queries to get a sense of what the common errors were.
Overall, most errors came from \WHERE\ and \HAVING\ (130 out of 341 are wrong due to \WHERE);
students often missed join conditions for queries involving many tables.
Other common errors include incorrect/redundant/missing tables in \FROM, incorrect order and missing/redundant expressions in \SELECT, and incorrect expressions in \GROUPBY.
We decided not to use the same queries from \textsf{Students}, as our participants had done the same/similar homework previously, which might bias the results.
Nonetheless, based on these observations, we designed four SQL questions using the DBLP schema (details in the appendix).
For each question, we crafted a wrong solution containing one or more mistakes:
two \WHERE\ errors for $Q_1$,
one \GROUPBY\ error and one \SELECT\ error for $Q_2$,
one \WHERE\ error for $Q_3$,
and one each \WHERE\ and \HAVING\ errors in $Q_4$.
Even though the queries are over a different schema, the errors above faithfully reflect real errors from \textsf{Students},
and they are consistent with the common errors found by others~\cite{ahadi2016students, brass2006semantic}.

Then, we performed a small study with four graduate teaching assistants (TAs) to generate hints for these queries.
Each TA was asked to pinpoint all mistakes in each query and offer hints, as if they were helping students debug wrong queries.
To simulate an office-hour setting, we asked TAs to finish all four questions in one sitting, with no help from \oursys.
We collected all hints provided by the TAs as ``expert'' hints.

Next, we ran \oursys\ on all wrong queries to obtain repair sites and fixes.
We removed fixes and only showed repair sites to the participants as hints.
To prevent participants from recognizing the source of hints (experts vs.\ \oursys) by their wording,
we paraphrased all hints to use a common template ``In [\emph{SQL clause}], [\emph{hint}]'' and standard wording.

\mypar{Tasks}
Using the four queries, 
each participant saw and completed three questions.
Students were required to complete questions on Q1 and Q2, and they completed one of Q3 and Q4 at random.
For each question, students were given the database schema, problem statement in English, and the wrong SQL query, and were asked to explain what is wrong with the query.  For creating {\em treatment} and {\em control} groups, students received hints from \oursys\ for either Q1 or Q2 (not both) at random, and for the other one they were asked to detect errors without any hints provided; the order of the two questions with and without hints was also chosen at random.
For the last question, participants received Q3 or Q4 at random, and we showed the union of hints (mixed together) 
generated by the TAs as well as by \oursys, and asked participants to categorize each hint as one of the following: ``{\em Unhelpful or incorrect}'', ``{\em Helpful but require thinking}'', and ``{\em Obvious and giving away the answer}''. Participants were asked to finish all questions in one sitting. We recorded the time a participant spent on each question\footnote{$Q_1$ without/with hints took 704s/460s on average; $Q_2$ took 756s/658s. Students completed the survey asynchronously, so the time recorded may not be accurate.}. In our study, for Q1, 8 students answered it with no hints and 7 with hints from \oursys. For Q2, these numbers are 7 and 8 respectively. For the third question, 7 received Q3 and 8 received Q4.

\begin{figure}[t]
\vspace{-3mm}
\centering
\begin{subfigure}{.55\linewidth}
\includegraphics[scale=0.25]{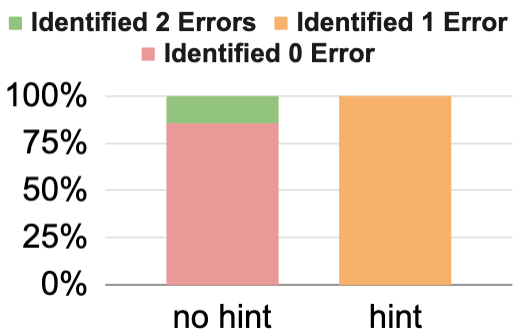}
\caption{Q1: No Hints vs. \oursys}
\label{fig:q1}
\end{subfigure}%
\begin{subfigure}{.45\linewidth}
\centering
\includegraphics[scale=0.25]{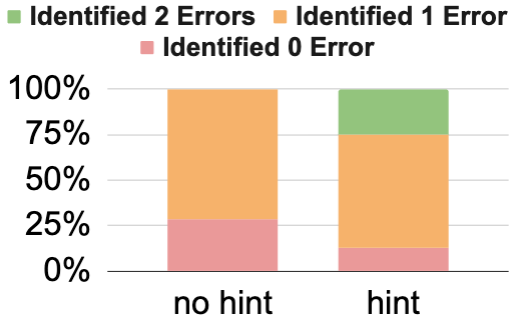}
\caption{Q2: No Hints vs. \oursys}
\label{fig:q2}
\end{subfigure}
\caption{User performance with/without \oursys } 
\end{figure}

\begin{figure}[t]
\vspace{-3mm}
\centering
\begin{subfigure}{.5\linewidth}
\includegraphics[scale=0.23]{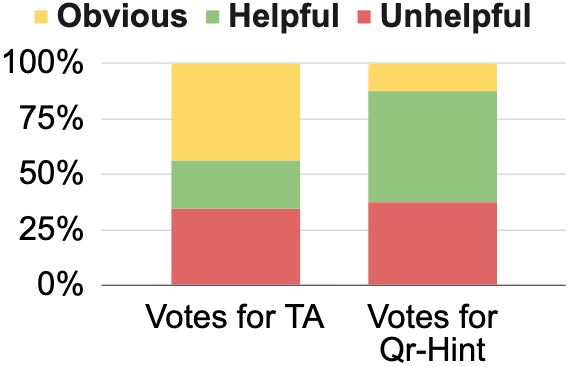}
\caption{Q3 Vote Cast}
\label{fig:q3-vote}
\end{subfigure}
\begin{subfigure}{.45\linewidth}
\centering
\includegraphics[scale=0.23]{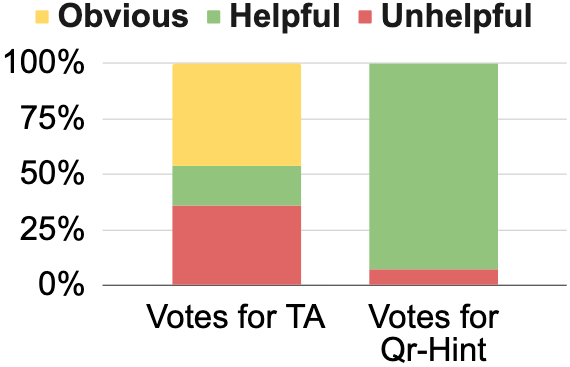}
\caption{Q4: Vote Cast}
\label{fig:q4-vote}
\end{subfigure}
\caption{Hint categorization from participants for Q3, Q4}
\label{fig:hint-categorization}
\end{figure}

\mypar{Result and Analysis} 
Our results 
for Q1 and Q2 show that participants were better at identifying at least one error in the query given the hints provided by \oursys\  compared to no hints. As shown in Figure~\ref{fig:q1} and Figure~\ref{fig:q2}, 100\% and 87.3\% of the participants were able to identify at least one of the two errors in the wrong query in Q1 and Q2 respectively after receiving hints from \oursys, as opposed to 14.3\% and 71.4\% who were able to do so without a hint. While there is a single participant who correctly identified both errors without any hint for Q1, this participant spent more than 20 minutes doing so, while most participants spent no more than 10 minutes on the same question without hints. 

Q3 and Q4 are used to evaluate whether \oursys\ provided hints that are comparable to the ones given by teaching assistants in terms of their quality. 
For Q3, there are four TA hints and one hint from \oursys; and there are four TA hints and two hints generated by \oursys\ for Q4. For all responses, we sum up the number of times participants vote for each of the three categories of hint ranks: ``Obvious'', ``Unhelpful'', and ``Helpful''. The results are shown in Figures~\ref{fig:q3-vote}, \ref{fig:q4-vote}. In summary, the quality of TAs' hints varies greatly as perceived by participants. On the other hand, \oursys\ is consistently perceived by participants as ``helpful but require thinking'', which might be best suited for classroom settings. 

\section{
Conclusion and 
Future Work}\label{sec:future-work}

We presented \oursys, a framework for automatically generating hints and suggestions for fixes for a wrong SQL query with respect to a reference query.  
We developed techniques to fix all clauses in a query and gave theoretical guarantees.
There are multiple intriguing directions of future work, including the support of more complex constructs such as subqueries, outer-joins (\sql{NULL}), and database constraints. 
There are many steps where the framework evaluates all possible options (e.g., repair sites), hence improving the scalability of the system is also a future work.
It will also be interesting to develop techniques to avoid the limitations of SMT solvers in our framework.
We are implementing a graphical user interface so that \oursys\ can better assist students/TAs in database courses. Conducting a larger-scale user study to further understand the effectiveness of \oursys\ is also important for future work.







\clearpage
\balance
\bibliographystyle{ACM-Reference-Format}
\bibliography{main}

\clearpage
\appendix

\section{Related Work Supplement}
{\bf Extended Discussion on Testing query equivalence.}
There are several classical results on query equivalence. 
Chandra et al. \cite{chandra1977optimal} show that equivalence testing of conjunctive queries is NP-complete. 
Aho et al.~\cite{aho1979equivalences} propose tableau to represent the value of a query, which is used to give algorithms for checking equivalence of SPJ queries. Sagiv et al.~\cite{sagiv1980equivalences} give a procedure for testing the equivalence of Select-Project-Join-Difference-Union (SPJDU) queries. Klug~\cite{klug1988conjunctive} presents algorithms for checking the equivalence of conjunctive queries with inequalities. The equivalence problem of some classes of queries under bag semantics has been proved to be undecidable \cite{ioannidis1995containment, jayram2006containment}.
While the query equivalence problem in general is undecidable \cite{trahtenbrot1950impossibility, abiteboul1995foundations}, tools are developed to check the equivalence of various classes of queries with restrictions and assumptions. Cosette~\cite{chu2017cosette, chu2017hottsql, chu2018axiomatic} transforms SQL queries to algebraic expressions and uses a decision procedure and rewrite rules to check if the resulting expressions of two queries are equivalent. EQUITAS~\cite{zhou2019automated} develops a symbolic representation of SQL queries in first-order logic, and uses satisfiability modulo theories (SMT) to check query equivalence. 
WeTune~\cite{wang2022wetune} 
builds a query equivalence verifier by utilizing the U-semiring structure~\cite{chu2018axiomatic}. 
These tools for query equivalence do not give hints on how to modify one query to become equivalent to another. 

\section{\FROM\ Stage Supplement}

\subsection{Finding Table Mapping}
In this section, we describe our heuristics for determining a table mapping. Since looking at \FROM\ alone does not have enough information for determining a mapping (e.g. when self-join exists), we gather information from other clauses (i.e., \WHERE, \GROUPBY, \SELECT) to make a decision. Given that both $Q$ and $Q^\star$ share the same multiset of tables, the general idea is to turn the table mapping into establishing one-to-one matching for elements between the two multisets.


To uniquely identify each table in both multisets, we first introduce \textit{table signature}. We first describe the construction of a table signature for a table in query $Q$:
\begin{enumerate}[leftmargin=*]
    \item Scanning $Q$'s \WHERE\ and \HAVING, for each selected operator ($=, <, >, \leq, \geq, \sql{LIKE}$) and each attribute $a$ in the table, we create a set of attributes that ``interact'' with $a$ in some atomic predicates in \WHERE\ and/or \HAVING\ (note: we rewrite the predicates to make sure $a$ is on the left-hand side of the operator in the case of inequality). If $a$ does not appear in \WHERE\ or \HAVING\, or it does not appear in a predicate with the selected operator, the corresponding set will be empty. After creating each set, we expand it to the entire equivalence class of its current attributes. We then replace each attribute in the set with the name of the original table they belong to. 
    \item Scanning $Q$'s \GROUPBY, create a set of attributes from the table that appear in any \GROUPBY\ expression.
    \item Scanning $Q$'s \SELECT, for each attribute $a$ in the table, create a set of indices such that this attribute appears in the \SELECT\ expression at the indexed position.
\end{enumerate}

In summary, let $t$ denote a table in a query $Q$, a \textit{table signature} is a triple $\sigma = (W_t, G_t, S_t)$, where $W_t$ is a function $W_t(a, o) \mapsto T, a \in \Attributes(t), o \in \{=, <, >, \leq, \geq, \sql{LIKE}\}$, and $T$ is the set of tables that interact with $a$ in $t$ through $o$; $G_t \mapsto E$ s.t. $E \subseteq \Attributes(t)$ and $\forall e \in E$ appear in at least one \GROUPBY\ expression in $Q$; and $S_t$ is a function $S_t(a) \mapsto I, a \in \Attributes(t)$, and $I$ is the set of integers whose indexed \SELECT\ expression in $Q$ contain $a$.

With table signatures, let $O = \{=, <, >, \leq, \geq, \sql{LIKE}\}$ denote the set of operators, we then define the following metric for calculating a normalized similarity between two signatures $\sigma = (W,G,S), \sigma' = (W', G', S')$:

\begin{align*}
\Sim(\sigma, \sigma') =& \frac{ \sum_{a \in \Attributes(t), o \in O} \dist(W(a,o), W'(a,o))  }{|\Attributes(t)| \times |O|} \\
                       & + \dist(G, G') \\
                       & + \frac{ \sum_{a \in \Attributes(t)} \dist(S(a), S'(a))  }{|\Attributes(t)|}
\end{align*}

Here we define $\dist$ as the Jaccard similarity between two sets. Each component is a normalized Jaccard similarity between the corresponding sets, and we take the sum of three Jaccard similarities (i.e. for \WHERE, \GROUPBY, \SELECT\ respectively) as our final similarity metric. Note that when two sets are empty, we count their Jaccard similarity as 1.

With the similarity metric, we build a bipartite graph where
\begin{itemize}
    \item Partitions 1 and 2 contain the table signatures for tables in $Q^\star$ and $Q$, respectively.
    \item Each node in partition 1 is connected with all nodes in partition 2 that refer to the same table. The weight of the edge between two nodes is the similarity metric between their signatures.
\end{itemize}

Given the bipartite graph, we select a table mapping that has the maximum cumulative similarity by solving the minimum-weight perfect matching problem. We now use an example to demonstrate the heuristic.

\begin{Example}\label{ex:table-mapping-full}
Continuing with \Cref{ex:intro1}, the initial signatures for \narrow{S1}, \narrow{S2}, \narrow{s1}, \narrow{s2} are shown in \Cref{ex:table-mapping}. After replacing attribute names with table names, the final signatures are the following:

\centerline{\footnotesize\setlength{\tabcolsep}{0.2em}
\begin{tabular}{rr|l|l|l|l}
    &
            & \narrow{S1} in $Q^\star$
                    & \narrow{S2} in $Q^\star$
                            & \narrow{s1} in $Q$
                                    & \narrow{s2} in $Q$
\\\hline\hline
\WHERE\ \&
    & \narrow{bar}:
            & ${=}\{\narrow{Frequents}\}$
                    & ${=}\{\narrow{Frequents}\}$
                            & None
                                    & None
\\
\HAVING\
    & \narrow{beer}:
            & ${=}\{\narrow{Likes}, \narrow{Serves}\}$
                    & ${=}\{\narrow{Likes}, \narrow{Serves}\}$
                            & ${=}\{\narrow{Likes}, \narrow{Serves}\}$
                                    & ${=}\{\narrow{Likes}, \narrow{Serves}\}$
\\
    & \narrow{price}:
            & ${\le}\{\narrow{Serves}\}$
                    & ${\ge}\{\narrow{Serves}\}$
                            & ${>}\{\narrow{Serves}\}$
                                    & ${<}\{\narrow{Serves}\}$
\\\hline
\GROUPBY\
    &
            & $\{\narrow{bar}, \narrow{beer}\}$
                    & $\{\narrow{beer}\}$
                            & $\{\narrow{beer}\}$
                                    & $\{\narrow{beer}\}$
\\\hline
\SELECT\
    & \narrow{bar}:
            & $\{ 2 \}$
                    & $\emptyset$
                            & $\emptyset$
                                    & $\{ 2 \}$
\\
    & \narrow{beer}:
            & $\{ 1 \}$
                    & $\{ 1 \}$
                            & $\{ 1 \}$
                                    & $\{ 1 \}$
\\
    & \narrow{price}:
            & $\emptyset$
                    & $\emptyset$
                            & $\emptyset$
                                    & $\emptyset$
\\\hline\hline
\end{tabular}
}

The weight of the edge between \narrow{S1} and \narrow{s1} is calculated as followed:
$\frac{4+5+3}{15} + \frac{1}{2} + \frac{0+1+1}{3} \approx 1.97$.
Comparing the \WHERE\ and \HAVING\ in \narrow{S1} and \narrow{s1}'s signatures, the Jaccard distance between $W(\narrow{bar}, =)$ and $W'(\narrow{bar}, =)$ is 0 as they do not have common element. Since the other 4 operators do not involve \narrow{bar}, their corresponding sets are all empty and thus yield a Jaccard distance of 1 between \narrow{S1} and \narrow{s1}. Thus the total Jaccard distance for \narrow{bar} is 4. Similarly, we obtain 5, 3, respectively for \narrow{beer} and \narrow{price}, and the sum of these Jaccard distances is normalized. The Jaccard distance between the \GROUPBY\ is simply $\frac{1}{2}$ as there is only one common element. For \SELECT, \narrow{beer} and \narrow{price} have the same sets between \narrow{S1} and \narrow{s1}, so the normalized Jaccard distance is $\frac{2}{3}$.

Following the same fashion, the weight of the rest of edges are below:
\begin{itemize}
    \item $\narrow{S1} \mapsto \narrow{s2}: \frac{4+5+3}{15} + \frac{1}{2} + \frac{1+1+1}{3} = 2.3$
    \item $\narrow{S2} \mapsto \narrow{s1}: \frac{4+5+3}{15} + 1 + \frac{1+1+1}{3} = 2.8$
    \item $\narrow{S2} \mapsto \narrow{s2}: \frac{4+5+3}{15} + 1 + \frac{0+1+1}{3} \approx 2.47$
\end{itemize}

With the signatures, the corresponding bipartite graph is shown below:
\begin{figure}[hbp]
    \centering
    \includegraphics[scale=0.3]{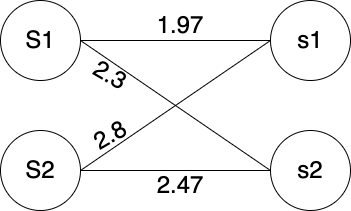}
    \caption{\small Constructed bipartite graph for \Cref{ex:table-mapping-full}}
    \label{fig:table-mapping}
\end{figure}
By negating the weight of each edge, we then convert the problem into solving the minimum-weight perfect matching on the graph, \oursys\ eventually choose the following mapping: $\narrow{S1} \mapsto \narrow{s2}, \narrow{S2} \mapsto \narrow{s1}$, which has the minimum overall weight ($-2.3 + -2.8 = -5.1$) among all matchings, and it is the best one for suggesting fixes in the downstream stages.
\end{Example}

\subsection{Proof for Lemma~\ref{lemma:from-optimality}}
\begin{proof}[Proof of \Cref{lemma:from-optimality}]
Without loss of generality, suppose that $Q^\star$ returns a non-empty result over some database instance $D$.
We construct a new database instance $D'$ as follows: 
\begin{enumerate}
\item For each unique table $T \in \Tables(Q^\star)$, duplicate each row in $T$ a number of times equal to a unique prime $p_T$ (such that $p_{T_1} \neq p_{T_2}$ for any $T_1 \neq T_2$).
\item For each table $T \not\in \Tables(Q^\star)$, make $T$ empty.
\end{enumerate}
We have
\[
|Q^\star(D')|/|Q^\star(D)| = \prod_{\text{unique}\,T \in \Tables(Q^\star)} p_T^{|\Aliases(Q^\star, T)|}.
\]
For any $Q$ equivalent to $Q^\star$, it must be the case that $\Tables(Q) \subseteq \Tables(Q^\star)$ (disregarding counts); otherwise $Q(D')$ would be empty by construction of $D'$. 
Furthermore, note that for $Q$ to be equivalent to $Q^\star$, we must have $|Q(D')|/|Q(D)| = |Q^\star(D')|/|Q^\star(D)|$, so
\[
\prod_{\text{unique}\,T \in \Tables(Q)} p_T^{|\Aliases(Q, T)|} = \prod_{\text{unique}\,T \in \Tables(Q^\star)} p_T^{|\Aliases(Q^\star, T)|}.
\]
It then follows from the Prime Factorization Theorem~\cite{gauss1966disquisitiones} that $\Tables(Q)$ and $\Tables(Q)$ contain the same \emph{set} of tables, and that for each unique $T \in \Tables(Q)$ (or $\Tables(Q)$), $|\Aliases(Q, T)| = |\Aliases(Q^\star, T)|$.
In other words, $\Tables(Q) \mseq \Tables(Q^\star)$.
\end{proof}

\section{\WHERE\ Stage Supplement}

In this section, we give a complete story for the derivation of fixes (i.e. \Cref{alg:derive-fixes}) and its optimization, as well as give proof to \Cref{lemma:where-correctness}, \Cref{lemma:where-optimality}, \Cref{lemma:correctness-create-bounds} and \Cref{lemma:correctness-derive-fixes}.

\subsection{\DeriveFixes\ Revisited}

\subsubsection{Target Bound Derivation}
We first completely map out the story for lines~\ref{l:DeriveFixes:recurse:start}--\ref{l:DeriveFixes:recurse:end} of \Cref{alg:derive-fixes}, which contain the strategy for pushing down target bound based on the repair bounds of the children of a node.

For illustration and without loss of generality, we first assume that all $\land, \lor$ nodes have only two children $c, c'$. While addressing a particular node $x$, we use the same notation as in \Cref{sec:where}, where $[l_c, u_c], [l_{c'}, u_{c'}]$ represent the repair bound of the child $c, c'$ respectively, and $[l^\star, u^\star]$ denotes the $x$'s target bound, and $[l^\star_c, u^\star_c], [l^\star_{c'}, u^\star_{c'}]$ denote the target bound of $c, c'$ respectively. 

The principle that guides the formulation of the target bound is simple: knowing that we can only pick a formula within the repair bounds at any node (as the repair bounds dictate the possible formula), we would like to tighten the repair bounds as little as possible to form target bounds so that \MinFix\ have more freedom to find a small formula. We use an example to demonstrate why loose bounds potentially yield small formulas. 

\begin{Example}\label{ex:min-formula-optimization}
Let $l \equiv (a_1 \lor a_2) \land a_3$ and $u \equiv a_1 \lor a_2$ be the lower and upper bounds, where $a_1, a_2, a_3$ are independent atomic predicates. The minimal formula within $[l, u]$ is $a_1 \lor a_2$. However, let $u' \equiv a_1 \lor a_2 \lor a_3$ be the new upper bound, the minimal formula that falls within $[l, u']$ becomes $a_3$. $a_3$ is smaller than $a_1 \lor a_2$ in terms of the size of their syntax trees. Here we expand the range of $[l, u]$ by raising the upper bound (i.e., adding a disjunct to the upper bound).

Symmetrically, we can create a case when range expansion is done by restricting the lower bound (i.e. adding a conjunct to the lower bound). Let $l \equiv a_1 \land a_2$ and $u \equiv (a_1 \land a_2) \lor a_3$ be the lower and upper bounds. If we restrict $l$ further by constructing $l' \equiv a_1 \land a_2 \land a_3$, the minimal formula moves from $a_1 \land a_2$ to $a_3$ for bounds $[l, u]$ and $[l', u]$ respectively. 
\end{Example}

With such a principle, the next question becomes how to tighten children's repair bounds to form target bounds such that we still guarantee combining the children yields the target formula. Situations differ based on the logical operator at each node.

\begin{figure}
\vspace{-3mm}
\centering
\begin{subfigure}{.35\linewidth}
\includegraphics[scale=0.12]{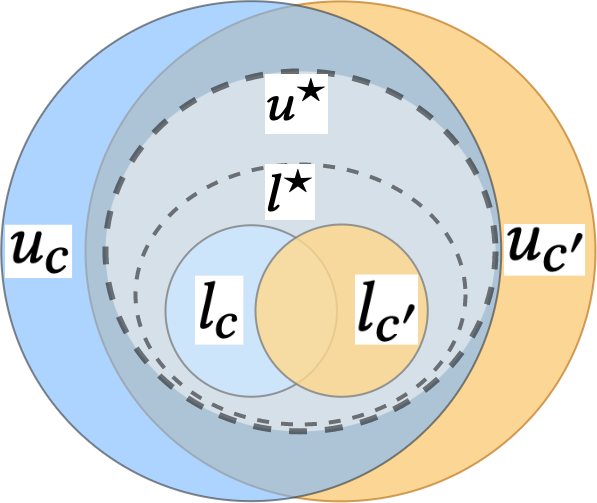}
\caption{The bounds for $\land$}
\label{fig:and-vis}
\end{subfigure}
\begin{subfigure}{.45\linewidth}
\centering
\includegraphics[scale=0.12]{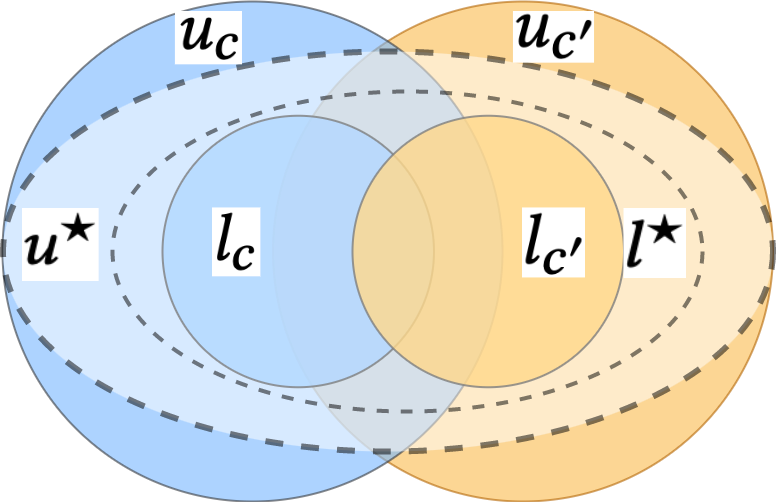}
\caption{The bounds for $\lor$}
\label{fig:or-vis}
\end{subfigure}
\caption{Venn Diagrams for visualizing relationship among formulae in Algorithm \ref{alg:derive-fixes}}
\end{figure}

{\bf When $x$ is rooted at $\land$}, the formula has a general form of $c \land c'$. Assuming both $c, c'$ contain some repair sites, we want to have formulas at the children $p_c, p_{c'}$ satisfying the following constraints after fixes are applied:
\begin{enumerate}[leftmargin=*]
    \item $p_c \in [l_c, u_c]$ and $p_{c'} \in [l_{c'}, u_{c'}]$, i.e., $p_c$ and $p_{c'}$ are within their respective repair bounds.
    \item $p_c \land p_{c'} \in [l^\star, u^\star]$, i.e., combining $p_c, p_{c'}$ using $\land$ forms a new formula that falls within the target bounds at $x$.
    \item The target bounds to be pushed down to $c, c'$ are contained in their repair bounds, i.e., (i) $[l^\star_c u^\star_c] \in [l_c, u_c]$, and (ii) $[l^\star_{c'}, u^\star_{c'}] \in [l_c, u_{c'}]$.
\end{enumerate}

Given the above constraints, the relationship among target bounds and repair bounds are depicted in \Cref{fig:and-vis}: $l_c \land l_{c'} \Rightarrow l^\star \Rightarrow u^\star \Rightarrow u_c \land u_{c'}$. 

We now make the following observation that helps us determine how we tighten $[l_c, u_c]$ and $[l_{c'}, u_{c'}]$ with $[l^\star, u^\star]$ to form $[l^\star_c, u^\star_c]$ and $[l^\star_{c'}, u^\star_{c'}]$, respectively.
\begin{itemize}
    \item $x \Leftrightarrow p_c \land p_{c'} \Leftrightarrow (p_c \lor \lnot p_{c'}) \land p_{c'} \Leftrightarrow p_c \land (p_{c'} \lor \lnot p_c)$
\end{itemize}
Such observation indicates that $p_c$ does not have to cover the semantics of $p_{c'}$ as $p_{c'}$ will be $\land$-ed back when reconstructing $x$. Since $x$ is currently represented by the target bound $[l^\star, u^\star]$, and adding a disjunct is to raise the upper bound as shown in \Cref{ex:min-formula-optimization}, we thus raise $u^\star$ to be $u^\star \lor \lnot p_{c'}$ for constructing the target upper bound for $c$. At this moment, however, we are not certain what $p_{c'}$ will be in the end, and this fact hints us to pick a formula whose semantics must be covered by all possible $p_{c'}$. Recalling the repair bound dictates possible fixes, we thus pick $u_{c'}$ as it is guaranteed to be implied by all possible $p_{c'}$. Finally, as $u^\star \lor \lnot u_{c'} \notin [l_c, u_c]$, we have to further tighten $u_c$ by adding $u^\star \lor \lnot u_{c'}$ as a conjunct, forming the target upper bound $u^\star_c = u_c \land (u^\star \lor \lnot u_{c'})$. As for the target lower bound, we have to pick $l^\star_c = l^\star$, because if $p_c$ ``undershoot'' $l^\star$ (i.e., $p_c \Rightarrow l^\star$ but $p_c \nLeftarrow l^\star$), the missing semantics cannot be covered by $\land$-ing any $p_{c'} \in [l_{c'}, u_{c'}]$ even if $l_{c'} = l^\star$ (i.e., $l^\star \nRightarrow p_c \land l^\star$, causing the resulting formula to fall out of $[l^\star, u^\star]$). Following a similar reasoning, we can obtain $[l^\star_{c'}, u^\star_{c'}] = [l^\star, u_{c'} \land (u^\star \lor \lnot u_c)]$. 

Finally, when there are more than two children under the $\land$ node, we process one child $c$ at a time while combining all other children into $c'$. Since $\land$ is commutative, the new repair bound for the combined $c'$ is simply formed by combining the repair lower/upper bounds from all other children.

{\bf When $x$ is rooted at $\lor$}, the formula has a general form of $c \lor c'$. Similar to an $\land$ node, we want to have $p_c, p_{c'}$ satisfying the following constraints:
\begin{enumerate}[leftmargin=*]
    \item $p_c \in [l_c, u_c]$ and $t_1 \in [l_{c'}, u_{c'}]$, i.e., $p_c$ and $p_{c'}$ are within their respective repair bounds.
    \item $p_c \lor p_{c'} \in [l^\star, u^\star]$, i.e., combining $p_c, p_{c'}$ (after fixes are applied) falls within the target bounds at the root $\lor$ node.
    \item The target bounds to be pushed down to $p_c, p_{c'}$ are contained in their repair bounds, i.e., (i) $[l^\star_c, u^\star_c] \in [l_c, u_c]$, and (ii) $[l^\star_{c'}, u^\star_{c'}] \in [l_{c'}, u_{c'}]$.
\end{enumerate}

Given the above constraints, the relationship among target bounds and repair bounds are depicted in \Cref{fig:or-vis}: $l_c \lor l_{c'} \Rightarrow l^\star \Rightarrow u^\star \Rightarrow u_c \lor u_{c'}$. 

We now make the following observation that helps us determine how we tighten $[l_c, u_c]$ and $[l_{c'}, u_{c'}]$ with $[l^\star, u^\star]$ to form $[l^\star_c, u^\star_c]$ and $[l^\star_{c'}, u^\star_{c'}]$, respectively.
\begin{itemize}
    \item $x \Leftrightarrow p_c \lor p_{c'} \Leftrightarrow (p_c \land \lnot p_{c'}) \lor p_{c'} \Leftrightarrow p_c \lor (p_{c'} \land \lnot p_c)$
\end{itemize}
Such observation indicates that $p_c$ does not have to cover the semantics of $p_{c'}$ as $p_{c'}$ will be $\lor$-ed back when reconstructing $x$. Since $x$ is currently represented by the target bound $[l^\star, u^\star]$, and adding a conjunct is to lower the lower bound as shown in \Cref{ex:min-formula-optimization}, we thus lower $l^\star$ to be $l^\star \land \lnot p_{c'}$ for constructing the target lower bound for $c$. With symmetric reasoning to the $\land$ case, we can obtain $[l^\star_c, u^\star_c] = [l_c \lor (l^\star \land \lnot l_{c'}), u^\star]$ and $[l^\star_{c'}, u^\star_{c'}] = [l_{c'} \lor (l^\star \land \lnot l_c), u^\star]$. Since $\lor$ is also commutative, we handle multiple children in the exact same manner as for $\land$.

{\bf When $x$ is rooted at $\lnot$}, it has only one child, and we thus push down the target bounds by setting them to be $[\lnot u^\star, \lnot l^\star]$ as negation inverts the direction of implication.

\subsubsection{Fix Minimization}
\begin{algorithm2e}[t] \small
\caption{$\MapAtomPreds(\set{P})$}
\label{alg:map-atom-preds}
\Input{a set $\set{P}$ of predicates}
\Output{a list of atomic predicates $\vec{a} = [ a_1, a_2, \ldots, a_{\dim(\vec{a})} ]$
    denoted by Boolean variables $\vec{\bool{a}} = [ \bool{a}_1, \bool{a}_2, \ldots, \bool{a}_{\dim(\vec{\bool{a}})} ]$,
    and a mapping $\phi$ such that for any Boolean subexpression $s$ in any predicate in $\set{P}$,
    $\phi(s)$ returns a Boolean function (with variables in $\vec{\bool{a}}$)
    that is equivalent to $s$ after replacing each variable $\bool{a}_i$ with predicate $a_i$}
\Let $\vec{a} = []$, $\vec{\bool{a}} = []$, and $\phi = \text{empty mapping}$\;
\ForEach{predicate in $\set{P}$ and each atomic predicate $t$ therein}{
    \ForEach(\tcp*[h]{see if $t$ is expressible by a chosen predicate}){$a_i \in \vec{a}$}{
        \uIf{$\IsEquiv(t, a_i)$}{
            \Let $\phi(t) = \bool{a}_i$; \Break;
        }\ElseIf{$\IsEquiv(t, \neg a_i)$}{
            \Let $\phi(t) = \neg \bool{a}_i$; \Break;
        }
    }
    \If(\tcp*[h]{a ``new'' atomic predicate found}){$\phi(t)$ is undefined}{
        $\vec{a}.\narrow{append}(t)$;
        $\vec{\bool{a}}.\narrow{append}(\bool{a}_{\dim(\vec{a})})$;
        \Let $\phi(t) = \bool{a}_{\dim(\vec{\bool{a}})}$;
    }
}
\ForEach{predicate in $\set{P}$ and each Boolean subexpression $s$ therein}{
    \If{$\phi(s)$ is undefined}{
        \Let $\phi(s)$ be the Boolean function obtained from $s$ by replacing each atomic predicate $t$ with $\phi(t)$;
    }
}
\Return $\vec{a}, \vec{\bool{a}}, \phi$;
\end{algorithm2e}


\begin{algorithm2e}[t]\small
\caption{$\MinFix(l^\star, u^\star)$}
\label{alg:min-fix}
\Input{predicates $l^\star$ and $u^\star$ together defining a target bound $[l^\star, u^\star]$}
\Output{a predicate bounded by $[l^\star, u^\star]$ that is as simple as possible}
\Let $\vec{a}, \vec{\bool{a}}, \phi = \MapAtomPreds(\{l^\star, u^\star\})$\; \nllabel{l:MinFix:MapAtomPreds}
\Let $g_l = \phi(l^\star)$ and $g_u = \phi(u^\star)$\tcp*{both are Boolean functions of $\vec{\bool{a}}$} \nllabel{l:MinFix:bounds}
\Let $g^\star = \BuildTruthTable(\vec{a}, \vec{\bool{a}}, g_l, g_u)$\; \nllabel{l:MinFix:target}
\Let $g = \MinBoolExp(g^\star)$\;
\Return predicate obtained from $g$ by replacing each variable $\bool{a}_i$ with atomic predicate $a_i$;

\BlankLine\setcounter{AlgoLine}{0}\nonl
\Subroutine $\BuildTruthTable(\vec{a}, \vec{\bool{a}}, g_l, g_u)$\DontPrintSemicolon\;\PrintSemicolon
\Output{a partial Boolean function $g^\star$ of $\vec{\bool{a}}$,
    consistent with the bounds defined by $g_l$ and $g_u$,
    represented as a mapping $\smash{\{0,1\}^{\dim(\vec{\bool{a}})}} \to \{\DontCare, 0, 1\}$}
\Let $g^\star = $ empty mapping\;
\ForEach{assignment $\vec{v} \in \smash{\{0,1\}^{\dim(\vec{\bool{a}})}}$ of $\vec{\bool{a}}$ \nllabel{l:BuildTruthTable:iteration}}{
    \Let $x = \bigwedge_{i \in [1..\dim(\vec{a})]} x_i$,
        where $x_i$ is $a_i$ if $\vec{v}$ assigns $\bool{a}_i$ to $1$,
        or $\neg a_i$ if $\vec{v}$ assigns $\bool{a}_i$ to $0$\;
    \uIf(\tcp*[h]{input setting not possible}){$\IsUnSat(x)$ \nllabel{l:BuildTruthTable:unsat1}}{
        \Let $g^\star(\vec{v}) = \DontCare$; \nllabel{l:BuildTruthTable:unsat2}
    }\uElseIf( \tcp*[h]{both true or both false}){$g_l(\vec{v}) = g_u(\vec{v})$ \nllabel{l:BuildTruthTable:sat1}}{
        \Let $g^\star(\vec{v}) = g_l(\vec{v})$;
    }\Else(\tcp*[h]{$g_l(\vec{v}) = 0$ and $g_u(\vec{v}) = 1$, because $l \Rightarrow u$}){
        \Let $g^\star(\vec{v}) = \DontCare$; \nllabel{l:BuildTruthTable:sat2}
    }
}
\Return $g^\star$;
\end{algorithm2e}
At each repair site, \Cref{alg:derive-fixes} calls \MinFix\ to compute the minimal fix given the target bound. The pseudocode of \MinFix\ is shown in \Cref{alg:min-fix}. It takes a lower bound $l^\star$ and an upper bound $u^\star$ as inputs, builds a desired truth table, and utilizes the \emph{ESPRESSO}~\cite{brayton1982comparison} to find a formula $t$ in disjunctive normal form, having the minimum number of minterms among all formulas that fall within the bounds.

However, since the truth table and ESPRESSO only work with Boolean variables instead of atomic predicates, \MinFix\ leverages a subroutine \MapAtomPreds\ (Algorithm~\ref{alg:map-atom-preds}, line~\ref{l:MinFix:MapAtomPreds} in Algorithm~\ref{alg:min-fix}) to 
\begin{enumerate}
    \item scan through both $l^\star$ and $u^\star$ and extracts all semantically unique atomic predicates.
    \item determines a mapping that maps each semantically unique atomic predicate in both $l^\star$, $u^\star$ to a set of unique Boolean variables that represent the atomic predicates (e.g., $a = b$ is semantically equivalent to $a+1 = b+1$, thus they will be mapped to the same Boolean variable).
    \item construct a mapping $\phi$ which maps any subexpression in a predicate (formula) to a Boolean function so that $\phi(l^\star)$ and $\phi(u^\star)$ return the Boolean functions that represent the truth table of $l^\star$ and $u^\star$ respectively.
\end{enumerate}

Given the Boolean functions for $l^\star$ and $u^\star$ (line~\ref{l:MinFix:bounds}), \MinFix\ then calls a subroutine \BuildTruthTable\ to construct the Boolean function representing the truth table of the minimal formula $t \in [l^\star, u^\star]$ (line~\ref{l:MinFix:target} in \MinFix) by going through all possible truth value assignments for $l^\star$ and $u^\star$ (line~\ref{l:BuildTruthTable:iteration} in \BuildTruthTable) with the following criteria: 
\begin{itemize}
    \item If the conjunction of all atomic predicates in a row is not satisfiable (e.g. $a=b$ and $a>b$ cannot both be true simultaneously), we mark the truth value for $t$ by ``$\DontCare$'' (don't-care) since such situation can never happen (line~\ref{l:BuildTruthTable:unsat1}-\ref{l:BuildTruthTable:unsat2} in \BuildTruthTable). 
    \item If the conjunction of all atomic predicates in a row is satisfiable (line~\ref{l:BuildTruthTable:sat1}-\ref{l:BuildTruthTable:sat2} in \BuildTruthTable):
    \begin{enumerate}
        \item if $l^\star$ and $u^\star$ are evaluated to the same truth value (i.e. both true or both false), then the same truth value will also be assigned to $t$.
         \item if $l^\star$ and $u^\star$ are evaluated to different truth values (i.e. false for $l^\star$ and true for $u^\star$), then a don't-care is assigned to $t$. The purpose of such a don't-care assignment is to allow more flexibility for Quine-McCluskey's method to minimize $t$ as much as possible. Note that the situation where $l^\star$ is evaluated to true and $u^\star$ is evaluated to false can never happen due to $l^\star \Rightarrow u^\star$.
    \end{enumerate}
\end{itemize}

After obtaining the Boolean function (i.e. truth table) for $t$, \MinFix\ feeds it to a subroutine \MinBoolExp\ which minimizes a given Boolean function $f$ with possible ``don't-care outputs'' (denoted \DontCare). Boolean minimization in general is NP-complete, but good heuristics exist that often find near-optimal solutions for even a large number of variables. Our implementation uses \emph{ESPRESSO}~\cite{brayton1982comparison}, but alternatives such as the Quine-McCluskey algorithm~\cite{mccluskey1956minimization} can also be used. We demonstrate how \MinFix\ works using the following example.

\begin{Example}\label{ex:minimize-bounds}
\begin{figure}[h]\scriptsize\setlength{\tabcolsep}{3pt}
    \begin{minipage}[b]{0.6\linewidth}\centering
    {\scriptsize
      \begin{tabular}[b]{|c|c|c|c|c|c|c|}\hline
        {\tt $a \geq b$} & {\tt $f=e$} & {\tt $a=b$} & {\tt $a>b$} & {\tt $l^\star$} & {\tt $u^\star$} & {\tt $t$} \\ \hline
        0 & 0 & 0 & 0 & 0 & 0 & 0 \\
        0 & 0 & 0 & 1 & $\DontCare$ & $\DontCare$ & $\DontCare$ \\
        0 & 0 & 1 & 0 & $\DontCare$ & $\DontCare$ & $\DontCare$ \\
        0 & 0 & 1 & 1 & $\DontCare$ & $\DontCare$ & $\DontCare$ \\
        0 & 1 & 0 & 0 & 0 & 1 & $\DontCare$ \\
        0 & 1 & 0 & 1 & 0 & 1 & $\DontCare$ \\
        0 & 1 & 1 & 0 & $\DontCare$ & $\DontCare$ & $\DontCare$ \\
        0 & 1 & 1 & 1 & $\DontCare$ & $\DontCare$ & $\DontCare$ \\
        1 & 0 & 0 & 0 & $\DontCare$ & $\DontCare$ & $\DontCare$ \\
        1 & 0 & 0 & 1 & 0 & 1 & $\DontCare$ \\
        1 & 0 & 1 & 0 & 1 & 1 & 1 \\
        1 & 0 & 1 & 1 & $\DontCare$ & $\DontCare$ & $\DontCare$ \\
        1 & 1 & 0 & 0 & $\DontCare$ & $\DontCare$ & $\DontCare$ \\
        1 & 1 & 0 & 1 & 1 & 1 & 1 \\
        1 & 1 & 1 & 0 & 1 & 1 & 1 \\
        1 & 1 & 1 & 1 & $\DontCare$ & $\DontCare$ & $\DontCare$ \\\hline
      \end{tabular}
      }
    \end{minipage}
\caption{\label{fig:minimize-bounds}\small
  Compact Truth tables for $l^\star$, $u^\star$ and $t$ in Example~\ref{ex:minimize-bounds}. }
\end{figure}
Consider the following bounds for a minimal formula $t$:
$l^\star \equiv (a \geq b \land f=e) \lor a=b$, and 
$u^\star \equiv a=b \lor e=f \lor a > b$.
Thus, we can construct a truth table for each formula shown in Figure~\ref{fig:minimize-bounds}. Based on the truth tables, it is clear that contradiction occurs when

\begin{itemize}
    \item $a=b$ and $a>b$ are both true.
    \item Either $a=b$ or $a>b$ is true but $a \geq b$ is false.
    \item $a \geq b$ is true but both $a=b$ and $a>b$ are false.
\end{itemize}

Running the Quine-McCluskey method on the final truth table for $t$ yields $t \equiv a \geq b$, and $t \in [l^\star, u^\star]$.
\end{Example}

Given \MinFix, the fixes computed for the repair sites in \Cref{ex:derive-fixes} are the following:
\begin{itemize}
    \item $x_4: a=b \lor (a=c \land d>10) \lor (a=c \land d<7)$
    \item $x_{10}, x_{12}: (N_{10}, N_{12}): (a=b \land d \neq e) \lor (a=b \land d>f) \lor (a=c \land d>10) \lor (a=c \land e<5)$
\end{itemize}

While the sizes of these fixes are still quite large, it is verifiable that the resulting formula is equivalent to $P^\star$ in \Cref{ex:syntax-tree} after inserting the fixes into $P$.

\subsection{Optimization for Algorithm~\ref{alg:derive-fixes}}
The suboptimality of multiple fixes comes from the independent derivation of the target formula bounds (Algorithm \ref{alg:derive-fixes}). Using the same notations as in Algorithm~\ref{alg:derive-fixes}, consider the target bounds for the children of an $\lor$ node with only two children and their Venn Diagram in Figure~\ref{fig:or-vis}:

{\small
\begin{itemize}[leftmargin=2em]
    \item[] $[l^\star_c, u^\star_c] = [l_c \lor (l^\star \land \lnot l_{c'}), u^\star]$, $[l^\star_{c'}, u^\star_{c'}] = [l_{c'} \lor (l^\star \land \lnot l_c), u^\star]$
\end{itemize}
}

Here we use the Venn diagram to illustrate the source of suboptimality. The regions representing $l^\star_c$ and $l^\star_{c'}$ overlap (same for $u^\star_c, u^\star_{c'}$). Though their union matches exactly the region of $l^\star$ ($u^\star$ respectively), such overlap indicates an overlap in their semantics, meaning semantic redundancy exists in the target bounds of \Children($x$). This implies that any combination of $p_c \in [l^\star_c, u^\star_c]$ and $p_{c'} \in [l^\star_{c'}, u^\star_{c'}]$ have semantic overlaps, causing suboptimality in the subsequent derivation of fixes. To illustrate this, reconsider Example~\ref{ex:derive-fixes}. The predicates $a=c$, $d > 10$, and $d < 7$ appear in the fix for $x_4$ unnecessarily.

To reduce such suboptimality, we propose to use a new procedure {\bf instead of Algorithm \ref{alg:derive-fixes}} to consider all repair sites simultaneously by leveraging the fact that Algorithm \ref{alg:derive-fixes} returns optimal fix for a single repair site (\Cref{lemma:where-optimality}). The overall routine is shown in \Cref{alg:min-fix-mult}. 

The general intuition of \Cref{alg:min-fix-mult} is to start with $l^\star$ and $u^\star$ being the reference formula and updates them in the same manner as Algorithm~\ref{alg:derive-fixes} until the lowest common ancestor (LCA) of all repair sites. Treating the LCA as a single repair site, the corresponding optimal fix lies within its $[l^\star, u^\star]$, and \Cref{alg:min-fix-mult} aims to make up such optimal fix by collectively deriving fixes for all actual repair sites. We next describe the major steps of \Cref{alg:min-fix-mult}. For a concise and clear illustration, we denote the optimal single fix at the LCA with $T$.

\mypar{Build consistency table} We first replace each repair site with a unique Boolean variable, forming a new Boolean predicate $P'$ at the LCA, and the goal is to make $P' \equiv T$. Because two Boolean formulas are equivalent if they share the same truth table, we then want to observe under what value assignments $T$ and $P'$ are consistent (i.e., evaluated to the same truth value) before determining how to construct each individual fix. Therefore, \Cref{alg:min-fix-mult} first generates a ``{\bf consistency table}'' (line~\ref{l:consistency-table}) where the inputs are all unique Boolean variables and atomic predicates from $T, P'$, and the formulas being evaluated are $l^\star, u^\star, T, P'$, here $l^\star, u^\star$ are present for the derivation of $T$ which follows the same procedures as in \MinFix\ (\Cref{alg:min-fix}).

\begin{Example}\label{ex:qr-hint-optimized}
Consider the following formulae: $P^\star \equiv a=1 \lor (b=2 \land c=3)$, $P \equiv c=3 \lor (b=2 \land a=1)$. 
           
Let the repair sites in $P$ be $c=3$ and $a=1$ with Boolean variables $r_1, r_2$, and their lowest common ancestor be the top-level $\lor$. We can obtain $l^\star \Leftrightarrow u^\star \Leftrightarrow P^\star$. Therefore, the consistency table can be constructed as shown in Figure~\ref{fig:consistency-table}. 

\begin{figure}[hbp]\scriptsize\setlength{\tabcolsep}{3pt}
    \begin{minipage}[b]{0.6\linewidth}\centering
    {\scriptsize
      \begin{tabular}[b]{|c|c|c|c|c|c|c|c|c|}\hline
        {\tt $r1$} & {\tt $r2$} & {\tt $a=1$} & {\tt $b=2$} & {\tt $c=3$} & {\tt $l^\star$} & {\tt $u^\star$} & {\tt $T$} & {\tt $P'$} \\ \hline
        0 & 0 & 0 & 0 & 0 & 0 & 0 & 0 & 0 \\
        0 & 0 & 0 & 0 & 1 & 0 & 0 & 0 & 0 \\
        0 & 0 & 0 & 1 & 0 & 0 & 0 & 0 & 0 \\
        0 & 0 & 0 & 1 & 1 & 1 & 1 & 1 & 0 \\
        0 & 0 & 1 & 0 & 0 & 1 & 1 & 1 & 0 \\
        0 & 0 & 1 & 0 & 1 & 1 & 1 & 1 & 0 \\
        0 & 0 & 1 & 1 & 0 & 1 & 1 & 1 & 0 \\
        \rowcolor{red!50}
        0 & 0 & 1 & 1 & 1 & 1 & 1 & 1 & 0 \\
        0 & 1 & 0 & 0 & 0 & 0 & 0 & 0 & 0 \\
        0 & 1 & 0 & 0 & 1 & 0 & 0 & 0 & 0 \\
        0 & 1 & 0 & 1 & 0 & 0 & 0 & 0 & 1 \\
        0 & 1 & 0 & 1 & 1 & 1 & 1 & 1 & 1 \\
        0 & 1 & 1 & 0 & 0 & 1 & 1 & 1 & 0 \\
        0 & 1 & 1 & 0 & 1 & 1 & 1 & 1 & 0 \\
        0 & 1 & 1 & 1 & 0 & 1 & 1 & 1 & 1 \\
        \rowcolor{blue!50}
        0 & 1 & 1 & 1 & 1 & 1 & 1 & 1 & 1 \\
        
        1 & 0 & 0 & 0 & 0 & 0 & 0 & 0 & 1 \\
        1 & 0 & 0 & 0 & 1 & 0 & 0 & 0 & 1 \\
        1 & 0 & 0 & 1 & 0 & 0 & 0 & 0 & 1 \\
        1 & 0 & 0 & 1 & 1 & 1 & 1 & 1 & 1 \\
        1 & 0 & 1 & 0 & 0 & 1 & 1 & 1 & 1 \\
        1 & 0 & 1 & 0 & 1 & 1 & 1 & 1 & 1 \\
        1 & 0 & 1 & 1 & 0 & 1 & 1 & 1 & 1 \\
        \rowcolor{blue!50}
        1 & 0 & 1 & 1 & 1 & 1 & 1 & 1 & 1 \\
        1 & 1 & 0 & 0 & 0 & 0 & 0 & 0 & 1 \\
        1 & 1 & 0 & 0 & 1 & 0 & 0 & 0 & 1 \\
        1 & 1 & 0 & 1 & 0 & 0 & 0 & 0 & 1 \\
        1 & 1 & 0 & 1 & 1 & 1 & 1 & 1 & 1 \\
        1 & 1 & 1 & 0 & 0 & 1 & 1 & 1 & 1 \\
        1 & 1 & 1 & 0 & 1 & 1 & 1 & 1 & 1 \\
        1 & 1 & 1 & 1 & 0 & 1 & 1 & 1 & 1 \\
        \rowcolor{blue!50}
        1 & 1 & 1 & 1 & 1 & 1 & 1 & 1 & 1 \\\hline
      \end{tabular}
      }
    \end{minipage}
  \caption{\label{fig:consistency-table}\small
  The consistency table for Example~\ref{ex:qr-hint-optimized} }
\end{figure}
\end{Example}
The consistency table gives a direct view of the occasions where $T$ and $P'$ are consistent (e.g. blue highlights, evaluated to the same truth value) and inconsistent (e.g. red highlights, evaluated to different truth values). This helps us determine how to construct each fix in later steps as we want to avoid any inconsistency between $T$ and $P'$ to achieve $T \Leftrightarrow P'$.

\begin{algorithm2e}[t] \small
\caption{$\MinFixMult(x, \set{S}, l^\star, u^\star)$}
\label{alg:min-fix-mult}
\Input{a formula $x$, a set $\set{S}$ of disjoint subtrees (repair sites) of $x$,
    and a target bound $[l^\star, u^\star]$ for $x$ to achieve by fixes}
\Output{a repair $(\set{S}, F)$, where $F$ maps each site in $\set{S}$ to a formula}
\Let $\set{U}$ denote the set of atomic formulas in $x$ that belong to none of the subtrees in $\set{S}$\; \nllabel{l:MinFixMult:consistency1}
\Let $\vec{a}, \vec{\bool{a}}, \phi = \MapAtomPreds(\set{U} \cup \{l^\star, u^\star\})$\;
\Let $g_l = \phi(l^\star)$ and $g_u = \phi(u^\star)$\tcp*{both are Boolean functions of $\vec{\bool{a}}$}
\Let $g^\star = \BuildTruthTable(\vec{a}, \vec{\bool{a}}, g_l, g_u)$\;
\tcp{Compute feasibility of truth values for repair sites:} \nllabel{l:consistency-table}
\Let $\vec{s} = [s_1, s_2, \ldots]$ be the list of subexpressions in $\set{S}$,
    denoted by the list of Boolean variables $\vec{\bool{s}} = [\bool{s}_1, \bool{s}_2, \ldots]$\;
\Let $g_x$ be a Boolean function with variables $\vec{\bool{a}} \concat \vec{\bool{s}}$,
    obtained from $x$ by replacing each subexpression $s_i$ with variable $\bool{s}_i$,
    and replacing each atomic formula $t \in \set{U}$ by $\phi(t)$\; \nllabel{l:MinFixMult:consistency2}
\Let $\mathfrak{C} = \InitFeasibility(\vec{\bool{a}}, \vec{\bool{s}}, g_x, g^\star)$\;
\tcp{Fix one site at a time, and incrementally update feasibility:} \nllabel{l:constraint-table}
\Let $F = $ empty mapping, and $\set{I} = [1..\dim(\vec{\bool{s}})]$\;
\While{$\set{I} \neq \emptyset$}{
    \Let $d, g^\star_d = \PickSite(\vec{\bool{a}}, \vec{\bool{s}}, \mathfrak{C}, \set{I})$\; \nllabel{l:greedy-start}
    \Let $g_d = \MinBoolExp(g^\star_d)$\;
    \Let $F(s_d) = $ formula obtained from $g_d$ by replacing each variable $\bool{a}_i$ with atomic formula $a_i$\;
    \Let $\mathfrak{C} = \UpdateFeasibility(\vec{\bool{a}}, \vec{\bool{s}}, \mathfrak{C}, d, g_d)$\;
    \Let $\set{I} = \set{I} \setminus \{ d \}$; \nllabel{l:greedy-end}
}
\Return $(\set{S}, F)$;
\end{algorithm2e}

\mypar{Build constraint table} After acquiring the consistency table, the next question becomes how to use all available atomic predicates to construct fixes for each repair site while avoiding inconsistencies. For this purpose, we turn all atomic predicates as ``constraints'' for all Boolean variables that represent the repair sites, thus constructing a ``{\bf constraint table}'' (line~\ref{l:constraint-table}). A constraint table is a truth table where inputs are all the atomic predicates in the consistency table, and the output is the concatenation of truth values of Boolean variables. For each row (i.e. truth assignment of all atomic predicates), the constraint table aggregates and lists all truth assignments of all Boolean variables where $T$ and $P'$ are consistent, and these are the potential truth values to be assigned individually to each Boolean variable.

\begin{Example}\label{ex:qr-hint-optimized-1}
Consider the consistency table in \Cref{fig:consistency-table} from \Cref{ex:qr-hint-optimized}. The corresponding constraint table is shown in \Cref{fig:constraint-table}. 

\begin{figure}[h]\scriptsize\setlength{\tabcolsep}{3pt}
  
    \begin{minipage}[b]{0.6\linewidth}\centering
    {\scriptsize
      \begin{tabular}[b]{|c|c|c|c|}\hline
        {\tt $a=1$} & {\tt $b=2$} & {\tt $c=3$}  & {\tt $r_1,r_2$} \\ \hline
        
        0 & 0 & 0 & 00,01  \\
        0 & 0 & 1 & 00,01  \\
        0 & 1 & 0 & 00  \\
        0 & 1 & 1 & 01,10,11  \\
        1 & 0 & 0 & 10,11  \\
        1 & 0 & 1 & 10,11  \\
        1 & 1 & 0 & 01,10,11  \\
        \rowcolor{blue!50}
        1 & 1 & 1 & 01,10,11  \\\hline
      \end{tabular}
      }
    \end{minipage}
  
  \caption{\label{fig:constraint-table}\small
  Constraint table for Example~\ref{ex:qr-hint-optimized} }
\end{figure}

The blue-highlighted row reflects the highlighting in the consistency table, where only the truth assignments $(0,1)$, $(1,0)$ or $(1,1)$ for $(r_1, r_2)$ produce consistent evaluations for $T$ and $P$. This indicates that when constructing $r_1$ and $r_2$ using the available atomic predicates $a=1, b=2$ and $c=3$, we must guarantee the truth assignment $(1,1,1)$ for $(a=1, b=2, c=3)$ would cause $(r_1, r_2)$ to be evaluated to either $(0,1)$, $(1,0)$ or $(1,1)$ respectively. All other rows in the constraint table are constructed in the same manner and carry the same implication. 

\end{Example}

\clearpage

\mypar{Compute minimal fixes} The final step is to compute a fix for each repair site according to the constraint table. While a constraint table lists all possible simultaneous truth assignments for all repair sites (e.g. last column of \Cref{fig:constraint-table}), we cannot make independent choices for the truth value of each repair site as dependencies exist. For example, in the highlighted row in \Cref{fig:constraint-table}, if $r_1$ is assigned $0$, then $r_2$ can only be assigned $1$ for consistency. On the other hand, if $r_1$ is assigned $1$, then $r_2$ can be assigned either $0$ or $1$. At this point, it is unclear how truth values can be assigned to each repair site to obtain minimum fixes, we thus follow a greedy procedure (line~\ref{l:greedy-start}-\ref{l:greedy-end}):
\begin{enumerate}
    \item Randomly pick a repair site $r$.
    \item Iterate over each row in the last column of the constraint table and give $r$ the most flexible assignment (i.e. if both $0$ and $1$ are available, assign a don't-care).
    \item Use Quine-McCluskey's method to compute the minimal fix for the repair site.
    \item Update the assigned don't-cares to a determined value by evaluating the fix with the corresponding truth value assignment. 
    \item Update the available options in the last column of the constraint table based on the truth assignment of $r$.
\end{enumerate}


\begin{algorithm2e}[htb] \small
\caption{Helper functions for \MinFixMult}
\label{alg:min-fix-mult-helpers}
\BlankLine\setcounter{AlgoLine}{0}\nonl
\Subroutine $\InitFeasibility(\vec{\bool{a}}, \vec{\bool{s}}, g_x, g^\star)$\DontPrintSemicolon\;\PrintSemicolon
\Output{$\mathfrak{C}: \smash{\{0,1\}^{\dim(\vec{\bool{a}})}} \to \{\DontCare\} \cup \mathbb{P}(\smash{\{0,1\}^{\dim(\vec{\bool{s}})}})$,
    a mapping such that for each assignment $\vec{v}$ of variables in $\vec{\bool{a}}$,
    $\mathfrak{C}(\vec{v})$ returns the set of feasible truth value settings for $\vec{\bool{s}}$
    (such that $g_x$ is consistent with $g^\star(\vec{v})$),
    or $\DontCare$ if $\vec{v}$ is impossible or irrelevant (i.e., $g^\star(\vec{v}) = \DontCare$)}
\Let $\mathfrak{C} = $ empty mapping\;
\ForEach{assignment $\vec{v} \in \smash{\{0,1\}^{\dim(\vec{\bool{a}})}}$ of $\vec{\bool{a}}$}{
    \If(\tcp*[h]{impossible or irrelevant}){$g^\star(\vec{v}) = \DontCare$}{
        \Let $\mathfrak{C}(\vec{v}) = \DontCare$;
        \Continue;
    }
    \Let $\mathfrak{C}(\vec{v}) = \emptyset$\;
    \ForEach{assignment $\vec{u} \in \smash{\{0,1\}^{\dim(\vec{\bool{s}})}}$ of $\vec{\bool{s}}$}{
        \lIf{$g_x(\vec{v} \concat \vec{u}) = g^\star(\vec{v})$}{%
            \Let $\mathfrak{C}(\vec{v}) = \mathfrak{C}(\vec{v}) \cup \{ \vec{u} \}$%
        }
    }
}
\Return $\mathfrak{C}$;

\BlankLine\setcounter{AlgoLine}{0}\nonl
\Subroutine $\UpdateFeasibility(\vec{\bool{a}}, \vec{\bool{s}}, \mathfrak{C}, d, g_d)$\DontPrintSemicolon\;\PrintSemicolon
\Output{$\mathfrak{C}'$, a more constrained version of $\mathfrak{C}$ that reflects the effect of
    wiring variable $\bool{s}_d$ to the Boolean function $g_d$}
\Let $\mathfrak{C}' = $ empty mapping\;
\ForEach{assignment $\vec{v} \in \smash{\{0,1\}^{\dim(\vec{\bool{a}})}}$ of $\vec{\bool{a}}$}{
    \uIf(\tcp*[h]{impossible or irrelevant}){$\mathfrak{C}(\vec{v}) = \DontCare$}{
        \Let $\mathfrak{C}'(\vec{v}) = \DontCare$;
    }\Else(\tcp*[h]{only include settings consistent with $g_d$}){
        \Let $\mathfrak{C}'(\vec{v}) = \{ \vec{u} \in \mathfrak{C}(\vec{v}) \mid u_d = g_d(\vec{v}) \}$;
    }
}
\Return $\mathfrak{C}'$;

\BlankLine\setcounter{AlgoLine}{0}\nonl
\Subroutine $\PickSite(\vec{\bool{a}}, \vec{\bool{s}}, \mathfrak{C}, \set{I})$\DontPrintSemicolon\;\PrintSemicolon
\Output{index $d \in \set{I}$ as the next site to fix,
    and a partial Boolean function $g^\star_d$ of $\vec{\bool{a}}$,
    represented as a mapping $\smash{\{0,1\}^{\dim(\vec{\bool{a}})}} \to \{\DontCare, 0, 1\}$,
    specified in accordance with the feasibility map $\mathfrak{C}$}
\lForEach(\tcp*[f]{init accumulators for priority calculation}){$i \in \set{I}$}{%
    \Let $c_i = 0$%
}
\ForEach{assignment $\vec{v} \in \smash{\{0,1\}^{\dim(\vec{\bool{a}})}}$ of $\vec{\bool{a}}$}{
    \lIf(\tcp*[f]{impossible or irrelevant}){$\mathfrak{C}(\vec{v}) = \DontCare$}{%
        \Continue%
    }
    \ForEach{$i \in \set{I}$}{
        \Let $r = |\{ \vec{u} \in \mathfrak{C}(\vec{v}) \ \mid u_i = 1 \}| \;/\; |\mathfrak{C}(\vec{v})|$\;
        \Let $c_i = c_i + |r-0.5|$\tcp*{prioritize uneven splits}
    }
}
\Let $d = \max\arg_{i \in I} c_i$\;
\Let $g^\star_d = $ empty mapping\;
\ForEach{assignment $\vec{v} \in \smash{\{0,1\}^{\dim(\vec{\bool{a}})}}$ of $\vec{\bool{a}}$}{
    \If(\tcp*[h]{impossible or irrelevant}){$\mathfrak{C}(\vec{v}) = \DontCare$}{%
        \Let $g^\star_d(\vec{v}) = \DontCare$; \Continue;
    }
    \Let $B = \{ u_d \mid \vec{u} \in \mathfrak{C}(\vec{v}) \}$\tcp*{$0$/$1$ settings for $\bool{s}_d$, deduplicated}
    \uIf(\tcp*[h]{forced to choose one truth value}){$|B| = 1$}{
        \Let $g^\star_d(\vec{v}) = $ the only element of $B$;
    }\Else(\tcp*[h]{$|B| = 2$, so choose either $0$ or $1$ (note that $B \neq \emptyset$)}){
        \Let $g^\star_d(\vec{v}) = \DontCare$;
    }
}
\Return $(d, g^\star_d)$;

\end{algorithm2e}

\begin{Example}\label{ex:qr-hint-optimized-2}
Continue from \Cref{ex:qr-hint-optimized-1}, the derivation process for $r_1$ and $r_2$ are shown in an extended constraint table in \Cref{fig:constraint-table-full}. Given the previous procedure, we first give $r_1$ maximum flexibility for constructing its formula, which yields $a=1$. We then update the don't-cares accordingly before computing $r_2$. As for $r_2$, we follow the same procedure and derive the truth assignment based on the truth values of $r_1$. Finally, the procedure yields $c=3$. These are indeed the optimal fixes. 
\begin{figure}[ht]\scriptsize\setlength{\tabcolsep}{3pt}
  
    \begin{minipage}[b]{0.6\linewidth}\centering
    {\scriptsize
      \begin{tabular}[b]{|c|c|c|c|c|c|}\hline
        {\tt $a=1$} & {\tt $b=2$} & {\tt $c=3$}  & {\tt $r_1,r_2$} & {\tt $r_1$} & {\tt $r_2$} \\ \hline
        
        0 & 0 & 0 & 00,01 & 0 & $\DontCare \rightarrow 0$  \\
        0 & 0 & 1 & 00,01 & 0 & $\DontCare \rightarrow 1$ \\
        0 & 1 & 0 & 00 & 0 & 0 \\
        0 & 1 & 1 & 01,10,11 & $\DontCare \rightarrow 0$ & 1 \\
        1 & 0 & 0 & 10,11 & 1 & $\DontCare \rightarrow 0$ \\
        1 & 0 & 1 & 10,11 & 1 & $\DontCare \rightarrow 1$ \\
        1 & 1 & 0 & 01,10,11 & $\DontCare \rightarrow 1$ & 0 \\
        \rowcolor{blue!50}
        1 & 1 & 1 & 01,10,11 & $\DontCare \rightarrow 1$ & $\DontCare \rightarrow 1$ \\\hline
      \end{tabular}
      }
    \end{minipage}
  
  \caption{\label{fig:constraint-table-full}\small
  Extended constraint table for Example~\ref{ex:qr-hint-optimized-2} }
\end{figure}

In addition, running \oursys-optimized over \Cref{ex:syntax-tree} yields optimal fixes $a=b$ and $d>10 \land e<5$ for repair sites $x_4, (x_{10}, x_{12})$.
\end{Example}

\subsection{Proof of Lemma~\ref{lemma:where-correctness}}
\begin{proof}[Proof of \Cref{lemma:where-correctness}]
We prove the correctness of \WHERE-stage by proving the correctness of repair returned by \Cref{alg:find-where-repair}, i.e. applying the repair sites and fixes returned by \Cref{alg:find-where-repair} yields a new formula $P'$ such that $P^\star \Leftrightarrow P'$. The proof contains two steps:

\mypar{Step 1: \Cref{alg:create-bounds} returns the correct repair sites} Assume that there exists a set of fixes $\set{F}$ for a set of repair sites $\set{S}$ in $P$. By \Cref{lemma:correctness-create-bounds}, we can create a lower bound $P_{\bot}$ and an upper bound $P_{\top}$ such that $P' \in [P_{\bot}, P_{\top}]$, where $P'$ is the formula obtained by applying fixes to repair sites. Since $P' \Leftrightarrow P^\star$, $P^\star \in [P_{\bot}, P_{\top}]$. Thus, if a set of fixes exists for a set of repair sites, $P^\star$ must fall within the corresponding repair bounds at the root of $P$. This validates the procedure for determining repair sites.

\mypar{Step 2: \Cref{alg:derive-fixes} returns the correct fixes} Given a lower bound $P_{\bot}$ and an upper bound $P_{\top}$ for $P$  with respect to $\set{S}$ such that $P^\star \in [P_{\bot}, P_{\top}]$, by \Cref{lemma:correctness-derive-fixes}, we can derive a set of fixes $\set{F}$ for $\set{S}$ through \Cref{alg:derive-fixes}. 
\end{proof}

\subsection{Proof for Lemma~\ref{lemma:where-optimality}}
\begin{proof}
We use the induction over the structure of $P$ to prove the minimality of the fix $f$ for a single repair site $s$ in $P$. 

\mypar{Base case} The base case is simply when $s$ is $P$ (i.e. the entire $P$ is the repair site). In such case, $f$ is a minimal DNF of $P^\star$ returned by the Quine-McCluskey's method. Thus, removing any clause from $f$ causes $f \nLeftrightarrow P^\star$. 

\mypar{Induction Step} When $s$ is not $P$, there are three possible cases.

\mypar{Case 1} $P$ is in the form of $c_1 \land ... \land c_n$, where $c_i$ is the repair site. The target bounds for $c_i$ are derived to be $[P^\star, \top \land (P^\star \lor \lnot \bigwedge\limits_{j=1, j\neq i}^n c_j)]$ where the repair bounds of $\bigwedge\limits_{j=1, j\neq i}^n c_j$ is simply $[\bigwedge\limits_{j=1, j\neq i}^n c_j, \bigwedge\limits_{j=1, j\neq i}^n c_j]$ as it does not contain any repair site. Running Quine-McCluskey's method over the target bounds of $c_i$ yields a formula $f$, which is guaranteed to be in minimal DNF. Since we know $P^\star \Leftrightarrow \bigwedge\limits_{j=1, j\neq i}^n c_j \land f$ and $f$ is in minimal DNF, removing any of the clauses in $f$ would cause $\bigwedge\limits_{j=1, j\neq i}^n c_j \land f \Rightarrow P^\star$ but not vice versa. 

\mypar{Case 2} $P$ is in the form of $c_1 \lor ... \lor c_n$, where $c_i$ is the repair site. The target bounds for $c_i$ are derived to be $[\bot \lor (P^\star \land \lnot \bigvee\limits_{j=1, j\neq i}^n c_j), P^\star]$ where the repair bounds of $\bigvee\limits_{j=1, j\neq i}^n c_j$ are simply $[\bigvee\limits_{j=1, j\neq i}^n c_j, \bigvee\limits_{j=1, j\neq i}^n c_j]$. Running Quine-McCluskey's method over the target bounds of $c_i$ yields a formula $f$, which is guaranteed to be in minimal DNF. Since we know $P^\star \Leftrightarrow \bigvee\limits_{j=1, j\neq i}^n c_j \lor f$ and $f$ is in minimal DNF, removing any of the clauses in $f$ would cause $\bigvee\limits_{j=1, j\neq i}^n c_j \lor f \Rightarrow P^\star$ but not vice versa. 

\mypar{Case 3} $P$ is in the form of $\lnot c$. Here $P^\star \Leftrightarrow \lnot f$. Since $f$ is in minimal DNF, removing a clause results in $P^\star \nLeftrightarrow \lnot f$.
\end{proof}

\subsection{Proof of Lemma~\ref{lemma:correctness-create-bounds}}
\begin{proof}[Proof of \Cref{lemma:correctness-create-bounds}]
We use induction over the structure of $P$ to prove that for any subtree $x$ in $P$ (for which \CreateBounds\ is invoked), $\CreateBounds(x, \set{S}[x])$ returns a correct bound for $x$:
i.e., applying any repair to $\set{S}[x]$ in $x$ will result in a formula $x'$ bounded by $\CreateBounds(x, \set{S}[x])$.

\mypar{Base case}
Suppose $x$ is an atomic formula, \CreateBounds\ returns $[x,x]$ which bounds $x$. When $x$ is a repair site, \CreateBounds\ returns $[\False, \True]$, which certainly bounds $x$ or any Boolean expression with which we can replace $x$.

\mypar{Induction step}
Suppose $x$ is not atomic and is not itself a repair site.
Let $\Theta = \op(x)$ denote the logical operator at the root of $x$.
Every repair on $x$ (with the given $\set{S}[x]$) is obtained by (potentially) repairing each child of $x$, but without changing $\Theta$.
In other words, every repair $x$ results in $x' = \Theta_{c \in \Children(x)} c'$, where $c'$ is the result of some repair of $c$ at sites $\set{S}[c]$.
By the inductive hypothesis, $\forall c \in \Children(x): c' \in [l_c, u_c] = \CreateBounds(c, \set{S}[c])$.

There are two cases depending on $\Theta$.
In the case that $\Theta$ is $\land$ or $\lor$, since $\forall c \in \Children(x): l_c \Rightarrow c' \Rightarrow u_c$,
we have $\Theta_{c \in \Children(x)} l_c \Rightarrow \Theta_{c \in \Children(x)} c' \Rightarrow \Theta_{c \in \Children(x)} u_c$,
which means $x'$ is within the bound returned by Line~\ref{l:CreateBounds:and-or} of \CreateBounds.
In the case that $\Theta$ is $\lnot$, since $l_c \Rightarrow c' \Rightarrow u_c$,
we have $\lnot u_c \Rightarrow \lnot c' \Rightarrow \lnot l_c$,
which means $x'$ is within the bound returned by Line~\ref{l:CreateBounds:not} of \CreateBounds.
\end{proof}

\subsection{Proof of Lemma~\ref{lemma:correctness-derive-fixes}}
\begin{proof}[Proof of \Cref{lemma:correctness-derive-fixes}]
We use induction over the structure of $P$ to prove that for any subtree $x$ in $P$ for which $\DeriveFixes(x, \set{S}, l^\star, u^\star)$ is invoked:
\begin{itemize}
    \item (H1) $l^\star \Rightarrow u^\star$, and the bound $[l^\star, u^\star]$ implies (is equivalent or tighter than) the bound returned by $\CreateBounds(x, \set{S})$.
    \item (H2) The repair returned by $\DeriveFixes(x, \set{S}, l^\star, u^\star)$ yields some $x' \in [l^\star, u^\star]$.
\end{itemize}
Note that applying (H2) to the root of $P$ proves \Cref{lemma:correctness-derive-fixes}.

\mypar{Proving (H1) top-down}
The base case is when $x$ is the root of $P$; we only invoke $\DeriveFixes$ if $P^\star \in \CreateBounds(x, \set{S})$, so obviously $[l^\star, u^\star]$ implies $\CreateBounds(x, \set{S})$.
For the induction step, assuming that (H1) holds for $x$, we now show that (H1) for each child of $x$ for which \DeriveFixes\ is invoked.
There are three cases.

\mypar{Case 1} $x$ has form $\lnot c$.
Let $[l, u] = \CreateBounds(x, \set{S})$.
By the inductive hypothesis $l \Rightarrow l^\star \Rightarrow u^\star \Rightarrow u$.
Therefore $\lnot u \Rightarrow \lnot u^\star \Rightarrow \lnot l^\star \Rightarrow \lnot l$.
In other words, we call \DeriveFixes\ on $c$ with a bound (Line~\ref{l:DeriveFixes:not:start}) that implies $[\lnot u, \lnot l]$,
which is $\CreateBounds(c, \set{S}[c])$ by Line~\ref{l:CreateBounds:not} of \CreateBounds.

\mypar{Case 2} $x$ has form $c_1 \land ... \land c_n$. In this case, \Cref{alg:derive-fixes} divides the formula as $p_c \land p_{c'}$ where $p_c = c_i$, $p_{c'} = \bigwedge\limits_{j=1, j\neq i}^n c_j$.
We shall show that for $l_i^\star$ and $u_i^\star$ defined under Line~\ref{l:DeriveFixes:and:start}, $[l_i^\star, u_i^\star]$ implies $[l_i, u_i] = \CreateBounds(c_i, \set{S}[c_i])$ for all $i$; the cases for all other children are symmetric.
Clearly, $l_i \Rightarrow l^\star = l_i^\star$
and $u_i^\star = u_i \land (u^\star \lor \lnot u'_i) \Rightarrow u_i$.
Furthermore, note that:
\begin{align*}
l_i & \Rightarrow u_i;\\
l_i & \Rightarrow u^\star \lor \lnot u'_i;\\
l^\star & \xRightarrow{\text{\tiny ind.\ hypo. and Line~\ref{l:CreateBounds:and-or} of \CreateBounds}} u_i \land u'_i \Rightarrow u_0;\\
l^\star & \xRightarrow{\text{\tiny ind.\ hypo.}} u^\star \Rightarrow u^\star \lor \lnot u'_i.
\end{align*}
Therefore, $l_i^\star = l^\star \Rightarrow u_i \land (u^\star \lor \lnot u'_i) = u_i^\star$.

\mypar{Case 3} $x$ has form $c_1 \lor ... \lor c_n$. In this case, \Cref{alg:derive-fixes} divides the formula as $p_c \lor p_{c'}$ where $p_c = c_i$, $p_{c'} = \bigvee\limits_{j=1, j\neq i}^n c_j$.
We shall show that for $l_i^\star$ and $u_i^\star$ defined in the branch starting on Line~\ref{l:DeriveFixes:or:start}, $[l_i^\star, u_i^\star]$ implies $[l_i, u_i] = \CreateBounds(c_i, \set{S}[c_i])$ for all $i$; the cases for all other children are symmetric.
Clearly, $l_i \Rightarrow l_i \lor (l^\star \land \lnot l'_i) = l_i^\star$
and $u_i^\star = u^\star \Rightarrow u_i$.
Furthermore, note that:
\begin{align*}
l_i & \Rightarrow u_i;\\
l_i & \Rightarrow l_i \lor l'_i \xRightarrow{\text{\tiny ind.\ hypo. and Line~\ref{l:CreateBounds:and-or} of \CreateBounds}} u^\star;\\
l^\star \land \lnot l'_i & \Rightarrow u_i;\\
l^\star \land \lnot l'_i & \Rightarrow l^\star \xRightarrow{\text{\tiny ind.\ hypo.}} u^\star.
\end{align*}
Therefore, $l_i^\star = l_i \lor (l^\star \land \lnot l'_i) \Rightarrow u^\star = u_i^\star$.

\mypar{Proving (H2) bottom-up}
For the base case, when $x \in \set{S}$, assuming the correctness of \MinFix, \Cref{lemma:correctness-create-bounds} and (H1) ensure that $\MinFix(x, l^\star, u^\star)$ yields a repaired formula in $[l^\star, u^\star]$.

For the inductive step, assuming that (H2) holds for each child of $x$, we now show that (H2) holds for $x$.
There are three cases.

\mypar{Case 1} $x$ has form $\lnot c$.
By the inductive hypothesis, $\DeriveFixes$ on $c$ returns a repair that results in some $c'$ such that $\lnot u^\star \Rightarrow c' \Rightarrow \lnot l^\star$.
Clearly, the same repair, which is returned by $\DeriveFixes(x, \set{S}, l^\star, u^\star)$, changes $x$ to $\lnot c'$, which satisfies $l^\star \Rightarrow \lnot c' \Rightarrow u^\star$.

\mypar{Case 2} $x$ has form $c_1 \land ... \land c_n$.
By the inductive hypothesis, for all $1 \leq i \leq n$, $\DeriveFixes$ on $c_i$ returns a repair that results in some $c'_i$ such that $l^\star \Rightarrow c'_i \Rightarrow u_i \land (u^\star \lor \lnot u'_i)$.
The repair returned by $\DeriveFixes$ on $x$ results in $\bigwedge\limits_{i=1}^n c'_i$.
Clearly, $c'_1 \land ... \land c'_n \Leftarrow l^\star \land l^\star \Leftrightarrow l^\star$.
Also,
\begin{align*}
\bigwedge\limits_{i=1}^n c'_i & \Rightarrow  \bigwedge\limits_{i=1}^n \left(u_i \land (u^\star \lor \lnot u'_i)\right) \\
& \Leftrightarrow \bigwedge\limits_{i=1}^n u_i \land \left(u^\star \lor \left( \bigwedge\limits_{i=1}^n \lnot u'_i \right)\right)   \\
& \Leftrightarrow  \left(\bigwedge\limits_{i=1}^n u_i \land u^\star \right) \lor \left(\bigwedge\limits_{i=1}^n u_i \land \bigwedge\limits_{i=1}^n \lnot u'_i \right) \\
& \Leftrightarrow u^\star \lor \left(\bigwedge\limits_{i=1}^n u_i \land \bigwedge\limits_{i=1}^n \lnot u'_i \right) \\
& \Leftrightarrow u^\star \lor \bot \\ 
& \Leftrightarrow u^\star. 
\end{align*}


\mypar{Case 3} $x$ has form $c_1 \lor ... \lor c_n$.
By the inductive hypothesis, for all $1 \leq i \leq n$, $\DeriveFixes$ on $c_i$ returns a repair that results in some $c'_i$ such that $l_i \lor (l^\star \land \lnot l'_i) \Rightarrow c'_i \Rightarrow u^\star$.
The repair returned by $\DeriveFixes$ on $x$ results in $\bigvee\limits_{i=1}^n c'_i$.
Clearly, $c'_1 \lor ... \lor c'_n \Rightarrow u^\star \lor u^\star \Leftrightarrow u^\star$.
Also,
\begin{align*}
\bigvee\limits_{i=1}^n c'_i & \Leftarrow  \bigvee\limits_{i=1}^n (l_i \lor (l^\star \land \lnot l'_i)) \\
& \Leftrightarrow \bigvee\limits_{i=1}^n l_i \lor \left(l^\star \land \bigvee\limits_{i=1}^n \lnot l'_i \right) \\
& \Leftrightarrow \left(\bigvee\limits_{i=1}^n l_i \lor l^\star \right) \land \left(\bigvee\limits_{i=1}^n l_i \lor \bigvee\limits_{i=1}^n \lnot l'_i  \right) \\
& \Leftrightarrow l^\star \land \top \\ 
& \Leftrightarrow l^\star. 
\end{align*}
\end{proof}

\section{\GROUPBY\ Stage Supplement}

\subsection{Necessity of Fixing \GROUPBY}
We show that it is necessary to fix \GROUPBY. 
\begin{lemma}\label{lemma:aggr-necessity}
Consider two single-block SQL queries $Q_1$ and $Q_2$, where $Q_1$ has no \sql{GROUP} \sql{BY} or aggregation,
while $Q_2$ has \sql{GROUP} \sql{BY} and/or aggregation but no \sql{HAVING}.
$Q_1$ and $Q_2$ cannot be equivalent under bag semantics, assuming that no database constraints are present and there exists some database instance for which either $Q_1$ or $Q_2$ returns a non-empty result.
\end{lemma}

\begin{proof}[Proof of \Cref{lemma:aggr-necessity}]
Suppose for some instance $D$, both $Q_1$ and $Q_2$ return the same non-empty results.
Pick any $T \in \Tables(Q_1)$.
Create a new instance $D'$ by duplicating the contents of $T$ in the same table (i.e., doubling the multiplicity of each tuple in $T$) while keeping all other tables unchanged.
Since $Q_1$ has no \sql{GROUP} \sql{BY} or aggregation, the multiplicity of each tuple in $Q_1(D')$ will be increased by a factor of $2^c$ compared with that in $Q_1(D)$, where $c>1$ is the count of $T$ in $\Tables(Q_1)$ (which is multiset).
Hence, the size of $Q_1(D')$ is strictly larger than $Q_1(D)$.
On the other hand, consider $Q_2$, which has \GROUPBY\ and/or aggregation.
Between $Q_2(D')$ and $Q_2(D)$, the grouping of intermediate join result tuples remains the same, except the number of duplicates within each group.
Hence, the size of $Q_2(D')$ remains the same as $Q_2(D)$, which is the total number of groups.
Therefore, $Q_1(D')$ and $Q_2(D')$ are different.
\end{proof}

\subsection{Proof for Lemma~\ref{lemma:group-fix}}
\begin{proof}[Proof of \Cref{lemma:group-fix}]
    \mypar{Correctness} We prove the correctness of \GROUPBY-stage by showing $\vec{o} \setminus \Delta^- \union \Delta^+$ is equivalent to $\vec{\ostar}$. Assuming $\vec{o} \setminus \Delta^- \union \Delta^+$ is not equivalent to $\vec{\ostar}$, then among all possible pairs of tuples, there must exist $t_1, t_2$ such that $\bigwedge_i (o_i[t_1]\!=\!o_i[t_2]) \nLeftrightarrow \bigwedge_i (\ostar_i[t_1]\!=\!\ostar_i[t_2])$. This implies that $\IsSat(P[t_1] \land P[t_2] \land G^\star \land o_i[t_1]\!\neq\!o_i[t_2])$ (line~\ref{l:fix-grouping:redundant}) and/or $\IsSat(P[t_1] \land P[t_2] \land G \land \ostar_i[t_1]\!\neq\!\ostar_i[t_2])$ (line~\ref{l:fix-grouping:missing}) must be satisfiable, which contradicts the fact that \IsSat\ returns no false positive.

    \mypar{Strong Minimality of $\Delta^-$} Assuming $\Delta^- \nsubseteq \Delta^-_\circ$ (i.e., there exists an $o_x \in \vec{o}$ such that $o_x \in \Delta^-$ but $o_x \notin \Delta^-_\circ$), and $\vec{o} \setminus \Delta^-_\circ \union \Delta^+$ is equivalent to $\vec{\ostar}$. This implies that $o_x[t_1] = o_x[t_2]$ holds for all possible $t_1, t_2$, which contradicts the fact that $p$ is evaluated to true (otherwise $o_x$ should not be added to $\Delta^-$ according to the algorithm). Since \IsSat\ does not return false positive, $o_x$ does not exist and thus $\Delta^- \subseteq \Delta^-_\circ$.

    \mypar{Weak Minimality of $\Delta^+$} Assuming $\Delta^+ = \{\ostar_x\}$ and $\vec{o} \setminus \Delta^-_\circ$ is equivalent to $\vec{\ostar}$ (denoted by $\vec{o} \setminus \Delta^-_\circ \equiv \vec{\ostar}$). Following \Cref{alg:fix-grouping}, $\Delta^- \subseteq \Delta^-_\circ$ due to strong minimality. This indicates that $\vec{o} \Rightarrow \vec{\ostar}$ (i.e. $P[t_1] \land P[t_2] \land G \Rightarrow P[t_1] \land P[t_2] \land G^\star$) but $\vec{o} \nLeftarrow \vec{\ostar}$ (i.e. $P[t_1] \land P[t_2] \land G \nLeftarrow P[t_1] \land P[t_2] \land G^\star$). However, $\Delta^+ = \{\ostar_x\}$ indicates that $\IsSat(P[t_1] \land P[t_2] \land G \land \ostar_x[t_1] \neq \ostar_x[t_2])$ return true, where $\ostar_x$ is a conjunct in $G^\star$. This implies $\vec{o} \setminus \Delta^- \nRightarrow \vec{\ostar}$, thus further implying $\IsSat$ returns a false positive because $\vec{o} \setminus \Delta^-_\circ \equiv \vec{\ostar}$ was assumed at the beginning. Since \IsSat\ never returns false positive, no such $\ostar_x$ can exist in $\Delta^+$, and thus there does not exist a corresponding $\Delta^-_\circ$ such that $\vec{o} \setminus \Delta^-_\circ \equiv \vec{\ostar}$.

\end{proof}

\section{\HAVING\ Stage Supplement}

We use the following base context as default for testing satisfiability for \HAVING.
For brevity, the following assumes all values are integers;
if there are columns and literals of different domains, additional constraints will be added analogously.
We note that these constraints are \emph{not} intended to define the aggregation functions precisely;
rather, they encode only a subset of their properties that allow SMT solvers to deduce useful equivalences reasonably efficiently.

{\footnotesize
\newcommand{\XX}{\ensuremath{\mathbf{X}}}
\newcommand{\YY}{\ensuremath{\mathbf{Y}}}
\newcommand{\ZZ}{\ensuremath{\mathbf{Z}}}
\begin{align*}
\Context&: \left\{\;\begin{aligned}
\mathbf{X}, \mathbf{Y}, \mathbf{Z}, \mathbf{Ones} \text{ have type }\narrow{Array}(\mathbf{Z})\\
i, c \text{ have type }\narrow{Integer} \\
\sql{SUM}, \sql{AVG}, \sql{COUNT}, \sql{MAX}, \sql{MIN} \text{ have type }\narrow{Array}(\mathbb{Z}) \to \mathbb{Z}\\
\forall \XX, \YY, \ZZ: (\forall i: \XX[i] + \YY[i] = \ZZ[i]) \Rightarrow \sql{SUM}(\XX) + \sql{SUM}(\YY) = \sql{SUM}(\ZZ) \\ 
\forall \XX, \YY, \ZZ: (\forall i: \XX[i] - \YY[i] = \ZZ[i]) \Rightarrow \sql{SUM}(\XX) - \sql{SUM}(\YY) = \sql{SUM}(\ZZ) \\ 
\forall \XX, \YY, c: (\forall i: \XX[i] \times c = \YY[i]) \Rightarrow \sql{SUM}(\XX) \times c = \sql{SUM}(\YY)\\
\forall \XX, \YY, c: (\forall i: \XX[i] \div c = \YY[i]) \Rightarrow \sql{SUM}(\XX) \div c = \sql{SUM}(\YY)\\
    \text{// repeat the above 4 lines, replacing \sql{SUM} by \sql{AVG}}\\
\forall i: \mathbf{Ones}[i] = 1\\
\forall \XX: \sql{COUNT}(\XX) = \sql{COUNT}(\mathbf{Ones})\\
\forall \XX: \sql{SUM}(\XX) = \sql{AVG}(\XX) \times \sql{COUNT}(\mathbf{Ones})\\
    \text{// above assumes no \sql{NULL} values}\\
\forall \XX, i: \sql{MAX}(\XX) \ge \XX[i]\\
\forall \XX, i: \sql{MIN}(\XX) \le \XX[i]\\
\end{aligned}\;\right\}
\end{align*}
}
\section{\SELECT\ Stage Supplement}

\begin{algorithm2e}[t] \small
\caption{$\FixSelect(P, \vec{o}, \vec{\ostar})$}
\label{alg:fix-select}
\Input{a formula $P$ and two expression lists $\vec{o}$ and $\vec{\ostar}$}
\Output{a pair $(\Delta^-, \Delta^+)$,
where $\Delta^- \subseteq [1..\dim(\vec{o})]$ is a subset of indices of $\vec{o}$
and $\Delta^+ \subseteq [1..\dim(\vec{\ostar})]$ is a subset of indices of $\vec{\ostar}$}

\Let $\Delta^- = \emptyset$\;
\Let $\Delta^+ = \emptyset$\;
\Let $n = \min(\dim(\vec{o}),\dim(\vec{\ostar}))$\;
\ForEach{$o_i \in \vec{o}$, $\ostar_i \in \vec{\ostar}$, $1 \leq i \leq n$}{
    \If{$\IsSat_\Context(o_i\!\neq\!\ostar_i)$}{
        \Let $\Delta^- = \Delta^- \union \{i\}$\;
        \Let $\Delta^+ = \Delta^+ \union \{i\}$\;
    }
}

\ForEach{$o_i \in \vec{o}, n < i \leq \dim(\vec{o})$}{
    \Let $\Delta^- = \Delta^- \union \{i\}$\;
}

\ForEach{$\ostar_i \in \vec{\ostar}, n < i \leq \dim(\vec{\ostar})$}{
    \Let $\Delta^+ = \Delta^+ \union \{i\}$\;
}

\Return $(\Delta^-, \Delta^+)$\;
\end{algorithm2e}

The pseudocode for fixing \SELECT\ is shown in \Cref{alg:fix-select}. 

The correctness and optimality are the following:

\begin{lemma}\label{lemma:select-fix}
We say that two lists of \SELECT\ expression are equivalent if they produce the same set of columns in the same ordering. Let $(\Delta^-, \Delta^+) = \FixSelect(P, \vec{o}, \vec{\ostar})$. Assuming that subroutine $\IsSat_\Context$ returns no false positive, we have:

{\bf Correctness:} \oursys's \SELECT-stage hint leads to a fixed working query $Q_5$ that 1) passes the viability check ($\vec{o}$ and $\vec{\ostar}$ are equivalent); 2) satisfies $Q_5 \equiv Q^\star$. This applies to both \emph{SPJ} and \emph{SPJA} queries. 

{\bf Strong minimality of $(\Delta^-, \Delta^+)$ for \emph{SPJ}} Let $(\Delta^-_\circ, \Delta^+_\circ)$ denote the minimal $(\Delta^-, \Delta^+)$ respectively, then for any $(\Delta^-_\circ, \Delta^+_\circ)$ that make $\vec{o}$ and $\vec{\ostar}$ equivalent, $\Delta^- \subseteq \Delta^-_\circ, \Delta^+ \subseteq \Delta^+_\circ$.

\end{lemma}

\begin{proof}[Proof of \Cref{lemma:select-fix}]
    \mypar{Correctness} We prove the correctness by showing $\vec{o} \setminus \Delta^- \cup \Delta^+$ is equivalent to $\vec{\ostar}$. Since $\IsSat_\Context$ does not return false positive, $\vec{o}[i], \vec{\ostar}[i]$ are guaranteed to be added to $\Delta^-, \Delta^+$ respectively upon ``satisfiable'' or ``unknown'', and replacing $\vec{o}[i]$ with $\vec{\ostar}[i]$ guarantees the correctness of expression on position $i$. In addition, for any extra expressions in $\Delta^+, \Delta^-$ that do not have a counterpart in the other list, they are removed to ensure the number of expressions is the same between the \SELECT\ of $Q^\star, Q$. 

    \mypar{Strong minimality of $(\Delta^+, \Delta^-)$ in \emph{SPJ} queries} Assuming $\Delta^- \subsetneq \Delta^-_\circ$ (i.e. there exists an $o_x \in \vec{o}$ s.t. $o_x \in \Delta^-$ but $o_x \notin \Delta^-_\circ$), and $\vec{o} \setminus Delta^-_\circ \cup \Delta^+$ is equivalent to $\vec{o}$. This implies that either of the following is true: 1) $o_x$ is redundant in $Q$'s \SELECT\ and needs to be removed; 2) $\IsSat_\Context(o_x \neq \ostar_x)$ returns ``satisfiable''. In the former case, $o_x$ has to be removed and must be in $\Delta^-_\circ$ (otherwise $\Delta^-_\circ$ is not correct); in the latter case, $o_x \notin \Delta^-_\circ$ implies $\IsSat_\Context$ returns a false positive, which contradicts with our assumption. Thus $\Delta^- \subseteq \Delta^-_\circ$. $\Delta^+ \subseteq \Delta^+_\circ$ follows a similar proof. 
\end{proof}
\section{Experiments and User Study}

\subsection{Coverage Dataset}
The coverage dataset consists of 341 real wrong queries from students from an introductory database course at our institution, and \oursys\ can successfully fix 306 of them. The remaining 35 queries are not supported by \oursys\ due to the existence of set operations, outer joins, etc. The corresponding questions, solutions, and causes of errors are comprehensively shown in \Cref{tab:student-query-analysis}. 

The coverage of \oursys\ for the list of semantic errors proposed by Brass et al.~\cite{brass2006semantic} is shown in \Cref{tab:brass-error-list}. In summary, 25 of the 43 errors can be caught by \oursys, assuming the provided reference queries do not have any of the listed errors. Most importantly, the most frequent errors reported by \cite{brass2006semantic} are supported by \oursys. In addition, 17 out of the 25 supported errors are reflected in the coverage dataset. Note that a lot of errors in \cite{brass2006semantic} are efficiency and stylistic errors instead of logical errors (e.g., unnecessity of expressions in different clauses). 

\subsection{Schema and Questions in User Study Survey}

The following DBLP database schemas were used for the user study, we change the table name ({\tt inproceedings} $\rightarrow$ {\tt conference\_paper}, {\tt article} $\rightarrow$ {\tt journal\_paper}) in order to make them more intuitive for participants:
\begin{itemize}
    \item {\tt conference\_paper: (\underline{pubkey}, title, conference\_name, year, area)}
    \item {\tt journal\_paper: (\underline{pubkey}, title, journal\_name, year)}
    \item {\tt authorship: (\underline{pubkey}, \underline{author})}
\end{itemize}
The area attribute in the conference\_paper table can only be one of the following: "ML-AI", "Theory", "Database", "Systems" or "UNKNOWN". 

The questions and the correct queries for the user study are shown in Table~\ref{tab:user-study-queries}. The wrong queries and hints are shown in Table~\ref{tab:user-study-hints} (Hints from teaching assistants are in black, and hints from \oursys\ are in blue). All hints are shown in the same order as they were shown to the participants.

\begin{table*}[t]
  {    
  \centering
    \scriptsize
    \begin{tabular}{|p{20em}|p{40em}|}
      \hline
      \textbf{Question Statement} & \textbf{Correct Query}  \\\hline
      
      $Q_1$: Find names of the authors, such that among the years when he/she published both conference paper and journal paper, 2 of the published papers are at least 20 years apart.
      &
      SELECT i1.author \newline
      FROM conference\_paper i1, conference\_paper i2, journal\_paper a1, \newline
      journal\_paper a2, authorship au1, authorship au2, \newline
      authorship au3, authorship au4 \newline
      WHERE i1.pubkey = au1.pubkey AND i2.pubkey = au2.pubkey \newline
      AND a1.pubkey = au3.pubkey AND a2.pubkey = au4.pubkey \newline
      AND au1.author = au2.author AND au2.author = au3.author \newline
      AND au3.author = au4.author AND i1.year + 20 >= i2.year \newline
      AND i1.year = a1.year AND i2.year = a2.year \newline
      GROUP BY i1.author \\\hline
      
      $Q_2$: For each author who has published conference papers in the database area, find the number of their conference paper collaborators in the database area by years before 2018 (ignore the years when they have 0 collaborators). Your output should be in the format of (author, year, number of collaborators in that year).
      &
      SELECT t2.author, t1.year, COUNT(DISTINCT t3.author) \newline
      FROM conference\_paper t1, authorship t2, authorship t3 \newline
      WHERE t1.pubkey = t2.pubkey AND t3.pubkey = t1.pubkey  \newline
      AND t3.author <> t2.author AND t1.year < 2018  \newline
      AND t1.area = 'Database' \newline
      GROUP BY t2.author, t1.year \\\hline
      
      $Q_3$: Excluding publications in the year of 2015, find authors who publish conference papers in at least 2 areas.
      &
      SELECT t1.author \newline
      FROM conference\_paper t1, authorship t2, conference\_paper t3, authorship t4 \newline
      WHERE t1.pubkey = t2.pubkey AND t2.author = t4.author  \newline
      AND t3.pubkey = t4.pubkey AND t1.year = t3.year \newline
      AND t1.area <> t3.area AND t1.year <> 2015  \newline
      AND t1.area <> 'UNKNOWN' AND t3.area <> 'UNKNOWN' \newline
      GROUP BY t1.author \\\hline
      
      $Q_4$: Among the authors who publish in the Systems-area conferences, find the ones that have no co-authors on such publications (i.e. the author does not have any collaborator for any conference paper in systems area).
      &
      SELECT t2.author \newline
      FROM conference\_paper t1, authorship t2, authorship t3 \newline
      WHERE t1.pubkey = t2.pubkey \newline
      AND t2.pubkey = t3.pubkey AND t1.area = 'Systems' \newline
      GROUP BY t2.author \newline
      HAVING COUNT(DISTINCT t3.author) <= 1 \\\hline

    \end{tabular}
    \caption{Question statements and correct queries in the user study.}
    \label{tab:user-study-queries}
    }
\end{table*}

\begin{table*}[t]
  {    
  \centering
    \scriptsize
    \begin{tabular}{|p{35em}|p{25em}|}
      \hline
      \textbf{Wrong Queries} & \textbf{Hints}  \\\hline
      
      $Q_1$: \newline
      SELECT e.author  \newline
      FROM conference\_paper a, authorship e, conference\_paper b, authorship f,  \newline
      journal\_paper c, authorship g, journal\_paper d, authorship h  \newline
      WHERE a.pubkey = e.pubkey AND b.pubkey = g.pubkey  \newline
      AND c.pubkey = f.pubkey AND e.author = h.author  \newline
      AND d.pubkey = h.pubkey AND e.author = g.author  \newline
      AND f.author = h.author AND a.year + 20 > d.year  \newline
      GROUP BY e.author
      & 
      \textcolor{blue}{1. In WHERE: You should change "a.year + 20 > d.year" to some other conditions.} \\\hline

     $Q_2$: \newline
     SELECT a.author, year, COUNT(*)  \newline
     FROM conference\_paper, authorship, authorship a  \newline
     WHERE conference\_paper.pubkey = a.pubkey AND authorship.pubkey = a.pubkey  \newline
     AND a.author <> authorship.author AND year < 2018  \newline
     GROUP BY a.author, area, year, authorship.author  \newline
     HAVING area = 'Database' AND conference\_paper.year < 2018 
     &
     \textcolor{blue}{1. In GROUP BY: authorship.author is incorrect.} \newline
     \textcolor{blue}{2. In SELECT: COUNT(*) is incorrect.} \\\hline

     $Q_3$: \newline
     SELECT b.author  \newline
     FROM conference\_paper, authorship b, conference\_paper a, authorship  \newline
     WHERE conference\_paper.pubkey = authorship.pubkey AND a.year < 2015  \newline
     OR a.year > 2015 AND b.author = authorship.author  \newline
     AND a.pubkey = b.pubkey AND conference\_paper.year = a.year  \newline
     AND a.area <> conference\_paper.area AND a.area <> 'UNKNOWN'  \newline
     AND conference\_paper.area <> 'UNKNOWN'  \newline
     GROUP BY b.author
     &
     \textcolor{blue}{1. In WHERE, try to fix the whole condition by adding a pair of parentheses - in SQL AND takes higher precedence than OR (this fix alone should make the query correct)} \newline
     2. In WHERE, you are missing a pair of parentheses around a.year < 2015 OR a.year > 2015. \newline
     3. GROUP BY is incorrect. \newline
     4. GROUP BY is incorrect without an aggregate function. \\\hline

     $Q_4$: \newline
     SELECT a.author  \newline
     FROM authorship, conference\_paper, authorship a  \newline
     WHERE conference\_paper.pubkey = a.pubkey AND a.pubkey = authorship.pubkey  \newline
     GROUP BY a.author, conference\_paper.area  \newline
     HAVING conference\_paper.area = 'System' AND COUNT(DISTINCT a.author) <= 1
     &
     1. GROUP BY should not include t1.area.  \newline
     2. In HAVING, conference\_paper.area = 'System' should not appear.  \newline
     3. \textcolor{blue}{In HAVING, try to fix conference\_paper.area = 'System' (this plus another fix in HAVING will make the query right).}  \newline
     4. In HAVING, conference\_paper.area = 'System' should be = 'Systems'.  \newline
     5. \textcolor{blue}{In HAVING, try to fix COUNT(DISTINCT a.author) <= 1 (this plus another fix in HAVING will make the query right).}  \newline
     6. In HAVING, COUNT(DISTINCT a.author) <= 1 is referring to the same author attribute as the GROUP BY. \\\hline

    \end{tabular}
    \caption{Wrong queries and the hints provided (\oursys\ hints are in blue).}
    \label{tab:user-study-hints}
    }
\end{table*}

\begin{table*}[t]\centering

\scriptsize
\begin{tabular}{|c|c|c|c|l|l|} \hline
\multirow{6}{*}{ Question (a) } & Question & \multicolumn{4}{|l|}{Find the names of all beers served at James Joyce Pub.} \\\cline{2-6}
& Solutions & \multicolumn{4}{|l|}{SELECT beer
FROM serves
WHERE bar = 'James Joyce Pub';} \\\cline{2-6}
& \multirow{4}{*}{ Error Statistics } & \multicolumn{2}{|c|}{Total Wrong Query} & 22 & \\\cline{3-6}
& & \multicolumn{2}{|c|}{FROM} & 8 & 1. Wrong table; 2. Extra table (cross join with bar) \\\cline{3-6}
& & \multicolumn{2}{|c|}{WHERE} & 9 & Wrong bar name or typo \\\cline{3-6}
& & \multicolumn{2}{|c|}{SELECT} & 5 & SELECT * or bar alone instead of beer \\\hline
\multicolumn{6}{|c|}{} \\\hline
\multirow{6}{*}{ Question (b) } & Question & \multicolumn{4}{|l|}{Find names and addresses of bars that serve Budweiser at a price higher than 2.20.} \\\cline{2-6}
& Solutions & \multicolumn{4}{|l|}{SELECT name, address
FROM bar, serves
WHERE bar.name = serves.bar AND beer = 'Budweiser' AND price > 2.20;} \\\cline{2-6}
& \multirow{4}{*}{ Error Statistics } & \multicolumn{2}{|c|}{Total Wrong Query} & 126 & Note: 3 of them cannot be processed due to outer-join \\\cline{3-6}
& & \multicolumn{2}{|c|}{FROM} & 10 & Missing either Bar table or Serves table \\\cline{3-6}
& & \multicolumn{2}{|c|}{WHERE} & 96 & 1. Missing join condition; 2. Use >= instead of > \\\cline{3-6}
& & \multicolumn{2}{|c|}{SELECT} & 17 & 1. Missing Columns; 2. Column order is wrong \\\hline
\multicolumn{6}{|c|}{} \\\hline
\multirow{7}{*}{ Question (c) } & Question & \multicolumn{4}{|l|}{Find the names of drinkers who like Corona and frequent James Joyce Pub at least twice a week.} \\\cline{2-6}
& Solutions & \multicolumn{4}{|p{50em}|}{SELECT likes.drinker \newline
FROM likes, frequents \newline
WHERE likes.beer = 'Corona' AND likes.drinker = frequents.drinker \newline
AND frequents.bar = 'James Joyce Pub' AND frequents.times\_a\_week >= 2;} \\\cline{2-6}
& \multirow{5}{*}{ Error Statistics } & \multicolumn{2}{|c|}{Total Wrong Query} & 143 & Note: 20 of them cannot be processed due to usage of set operation, outer joins, complex subqueries \\\cline{3-6}
& & \multicolumn{2}{|c|}{FROM} & 11 & 1. Wrong table involved (serves); 2. Unnecessary drinker table (false positive, true error in SELECT/WHERE) \\\cline{3-6}
& & \multicolumn{2}{|c|}{WHERE} & 105 & 1. Missing join condition; 2. Using > instead of >=, or wrong number; 3. Missing condition on beer or bar \\\cline{3-6}
& & \multicolumn{2}{|c|}{SELECT} & 6 & SELECT * instead of name \\\cline{3-6}
& & \multicolumn{2}{|c|}{GROUP BY} & 1 & GROUP BY wrong columns \\\hline
\multicolumn{6}{|c|}{} \\\hline
\multirow{11}{*}{ Question (d) } & Question & \multicolumn{4}{|l|}{Find the name of each drinker who likes at least two beers.} \\\cline{2-6}
& Solution 1 & \multicolumn{4}{|l|}{SELECT drinker FROM likes GROUP BY drinker HAVING COUNT(*) >= 2;} \\\cline{2-6}
& Solution 2 & \multicolumn{4}{|l|}{SELECT DISTINCT l1.drinker
FROM likes l1, likes l2
WHERE l1.drinker = l2.drinker AND l1.beer <> l2.beer;} \\\cline{2-6}
& \multirow{8}{*}{ Error Statistics } & \multicolumn{2}{|c|}{Total Wrong Query} & 50 & Note: 12 of them cannot be processed due to usage of set operations \\\cline{3-6}
& & \multirow{4}{*}{ Solution 1 } & FROM & 1 & Wrong table \\\cline{4-6}
& & & GROUP BY & 1 & Group by 1 \\\cline{4-6}
& & & HAVING & 18 & 1. Using > instead of >=; 2. COUNT(DINSTINCT *) \\\cline{4-6}
& & & SELECT & 4 & Extra column COUNT \\\cline{3-6}
& & \multirow{3}{*}{ Solution 2 } & FROM & 5 & Extra/wrong tables (likes / frequents) \\\cline{4-6}
& & & WHERE & 2 & Wrong conditions: l1.beer = l2.beer, l1.drinker <> l2.drinker \\\cline{4-6}
& & & SELECT & 7 & Missing DISTINCT \\\hline
\end{tabular}
\caption{Student Query Statistics}\label{tab:student-query-analysis}
\end{table*}

\begin{table*}
\scriptsize
\begin{tabular}{|c|c|c|c|c|c|} \hline

\textbf{No.} & \multicolumn{2}{|l|}{\bf{Error}} & \textbf{Frequency in \cite{brass2006semantic}} & In \textbf{Students} Dataset &
 \textbf{\oursys\ Support} \\\cline{1-6}

1 & \multicolumn{2}{|l|}{Inconsistent condition} & $11.4\%$ & Y & \multirow{11}{*}{\begin{minipage}{2in}Logical errors \newline \oursys\ correctly gives hints\end{minipage}  } \\\cline{1-5}
3 & \multicolumn{2}{|l|}{Constant output columns} & $3.2\%$ & Y & \\\cline{1-5}
4 & \multicolumn{2}{|l|}{Duplicate output columns} & & Y & \\\cline{1-5}
5 & \multicolumn{2}{|l|}{Unused tuple variables} & $5.6\%$ & Y & \\\cline{1-5}
12 & \multicolumn{2}{|l|}{LIKE without wildcard} & & & \\\cline{1-5}
27 & \multicolumn{2}{|l|}{Missing join conditions} & $21.3\%$ & Y & \\\cline{1-5}
31 & \multicolumn{2}{|l|}{Comparison between different domains} & & Y & \\\cline{1-5}
33 & \multicolumn{2}{|l|}{DISTINCT in SUM and AVG} & & & \\\cline{1-5}
34 & \multicolumn{2}{|l|}{Wildcards without LIKE} & & Y & \\\cline{1-5}
37 & \multicolumn{2}{|l|}{Many duplicates} & $10.8\%$ & Y & \\\cline{1-5}
38 & \multicolumn{2}{|l|}{DISTINCT that might remove important duplicates} & & Y &  \\\hline

2 & \multicolumn{2}{|l|}{Unnecessary DISTINCT} & $3.7\%$ & Y & \multirow{11}{*}{ \begin{minipage}{2in}Efficiency/Stylistic issues \newline \oursys\ catch them if the reference queries are free from these errors\end{minipage} } \\\cline{1-5}
6 & \multicolumn{2}{|l|}{Unnecessary join} & $8.4\%$ & Y & \\\cline{1-5}
7 & \multicolumn{2}{|l|}{Tuple variables are always identical} & $3.2\%$ & & \\\cline{1-5}
15 & \multicolumn{2}{|l|}{Unnecessary aggregation function} & & & \\\cline{1-5}
16 & \multicolumn{2}{|l|}{Unnecessary DISTINCT in aggregation function} & & Y & \\\cline{1-5}
17 & \multicolumn{2}{|l|}{Unnecessary argument of COUNT} & & Y & \\\cline{1-5}
19 & \multicolumn{2}{|l|}{GROUP BY with singleton group} & $4.4\%$ & & \\\cline{1-5}
20 & \multicolumn{2}{|l|}{GROUP BY with only a single group} & & & \\\cline{1-5}
22 & \multicolumn{2}{|l|}{GROUP BY can be replaced by DISTINCT} & & Y & \\\cline{1-5}
24 & \multicolumn{2}{|l|}{Unnecessary ORDER BY term} & $10.8\%$ & & \\\cline{1-5}
32 & \multicolumn{2}{|l|}{Strange HAVING} & & &  \\\hline

8 & \multicolumn{2}{|l|}{Implied, tautological, or inconsistent subcondition} & $5.4\%$ & Y & \multirow{3}{*}{ \begin{minipage}{2in}Efficiency/Stylistic issues \newline \oursys\ does not consider them as errors\end{minipage}  }  \\\cline{1-5}
21 & \multicolumn{2}{|l|}{Unnecessary GROUP BY attribute} & & Y & \\\cline{1-5}
25 & \multicolumn{2}{|l|}{Inefficient HAVING} & & Y & \\\hline

9 & \multicolumn{2}{|l|}{Comparison with NULL} & & & \multirow{18}{*}{ Not supported by \oursys }  \\\cline{1-5}
10 & \multicolumn{2}{|l|}{NULL value in IN/ANY/ALL subquery} & & & \\\cline{1-5}
11 & \multicolumn{2}{|l|}{Unnecessarily general comparison operator} & & & \\\cline{1-5}
13 & \multicolumn{2}{|l|}{Unnecessarily complicated SELECT in EXISTS-subquery} & & & \\\cline{1-5}
14 & \multicolumn{2}{|l|}{IN/EXISTS condition can be replaced by comparison} & & & \\\cline{1-5}
18 & \multicolumn{2}{|l|}{Unnecessary GROUP BY in EXISTS subquery} & & & \\\cline{1-5}
23 & \multicolumn{2}{|l|}{UNION can be replaced by OR} & & & \\\cline{1-5}
26 & \multicolumn{2}{|l|}{Inefficient UNION} & & & \\\cline{1-5}
28 & \multicolumn{2}{|l|}{Uncorrelated EXISTS subquery} & & & \\\cline{1-5}
29 & \multicolumn{2}{|l|}{IN-Subquery with only one possible result value} & & & \\\cline{1-5}
30 & \multicolumn{2}{|l|}{Condition in the subquery that can be moved up} & & & \\\cline{1-5}
35 & \multicolumn{2}{|l|}{Condition on left table in left outer join condition} & & & \\\cline{1-5}
36 & \multicolumn{2}{|l|}{Outer join can be replaced by inner join} & & & \\\cline{1-5}
39 & \multicolumn{2}{|l|}{Subquery term that might return more than one tuple} & & & \\\cline{1-5}
40 & \multicolumn{2}{|l|}{SELECT INTO that might return more than one tuple} & & & \\\cline{1-5}
41 & \multicolumn{2}{|l|}{No indicator variable for argument that might be NULL} & & & \\\cline{1-5}
42 & \multicolumn{2}{|l|}{Difficult type conversion} & & & \\\cline{1-5}
43 & \multicolumn{2}{|l|}{Runtime error in datatype function. e.g. divided by 0} & & & \\\hline

\end{tabular}
\caption{List of Semantic Errors categorized by \oursys}\label{tab:brass-error-list}
\end{table*}


\end{document}